\documentclass[a4paper,twocolumn,11pt,accepted=2023-06-22]{quantumarticle}
\pdfoutput=1
\usepackage{amsmath}
\usepackage{amsthm}
\usepackage{amssymb}
\usepackage{graphicx}
\usepackage{hyperref}
\usepackage{color}
\usepackage{array}
\usepackage[makeroom]{cancel}
\usepackage[capitalize,nameinlink]{cleveref}
\usepackage{changes}
\usepackage{qcircuit}
\usepackage{braket}
\usepackage{bbm}
\usepackage{bm}
\usepackage{titlesec}
\usepackage{comment}
\usepackage{subfigure}
\usepackage{multirow}
\usepackage{cellspace} 
\usepackage{appendix}
\usepackage[numbers]{natbib}
\newcommand{\beq}{\begin{equation}}
\newcommand{\eeq}{\end{equation}}

\newcommand{\qromalgo}{\hyperref[eq:qrom_superposition_alg_maintxt]{Algorithm 1}}

\renewcommand{\theparagraph}{\arabic{paragraph}.}

\usepackage{empheq}

\newcommand{\ba}{\begin{align}}
\newcommand{\ea}{\end{align}}
\newcommand{\bq}{\bm{q}}
\newcommand{\bp}{\bm{p}}
\newcommand{\bG}{\bm{G}}
\newcommand{\bnu}{\bm{\nu}}

\newcommand{\bR}{\bm{R}}
\newcommand{\bb}{\bm{b}}
\newcommand{\val}{\text{val}}
\newcommand{\ion}{\mathrm{ion}}
\newcommand{\Ps}{\text{Ps}}

\newcommand{\qb}{\text{qrom}}
\newcommand{\dis}{$\text{Li}_2\text{FeSiO}_4$}
\newcommand{\limnfo}{$\text{Li}_{0.75}\text{MnO}_2\text{F}$}
\newcommand{\limnnio}{$\text{Li}_{0.75}\text{[Li}_{0.17}\text{Ni}_{0.25}\text{Mn}_{0.58}\text{]O}_2$}
\newcommand{\limno}{$\text{Li}_{0.5}\text{MnO}_3$}

\crefname{thm}{Theorem}{Theorems}
\crefname{dfn}{Definition}{Definitions}
\crefname{rmk}{Remark}{Remarks}
\crefname{lem}{Lemma}{Lemmas}
\crefname{cor}{Corollary}{Corollaries}

\theoremstyle{plain}  

\newtheorem{thm}{Theorem}[section]
\newtheorem{lem}[thm]{Lemma}
\newtheorem{innercustomlem}{Lemma}

\theoremstyle{definition}

\theoremstyle{remark}
\newtheorem{rmk}{Remark}[section]

\newenvironment{customlem}[1]
  {\renewcommand\theinnercustomlem{#1}\innercustomlem}
  {\endinnercustomlem}

\newcommand{\PREP}{\text{PREP}}
\newcommand{\SEL}{\text{SEL}}

\newcommand{\sgn}{\text{sgn}}
\newcommand{\error}{\varepsilon}

\newcommand{\mcB}{\mathcal{B}}

\newcommand{\mcG}{\mathcal{G}}

\newcommand{\mcT}{\mathcal{T}}

\newcommand{\mcX}{\mathcal{X}}

\newcommand{\mbbC}{\mathbb{C}}

\newcommand{\mbbZ}{\mathbb{Z}}

\newcommand{\irchi}[2]{\raisebox{\depth}{$#1\chi$}}
\DeclareRobustCommand{\rchi}{\mathpalette\irchi\relax}

\definecolor{JM}{RGB}{4,116,149}

\begin{document}
\title{Quantum simulation of battery materials using ionic pseudopotentials}
\author{Modjtaba Shokrian Zini}
\email{modjtaba@xanadu.ai}
\affiliation{Xanadu, Toronto, ON, M5G 2C8, Canada}
\author{Alain Delgado}
\affiliation{Xanadu, Toronto, ON, M5G 2C8, Canada}
\author{Roberto dos Reis}
\affiliation{Xanadu, Toronto, ON, M5G 2C8, Canada}
\author{Pablo A. M. Casares}
\affiliation{Xanadu, Toronto, ON, M5G 2C8, Canada}
\author{Jonathan E. Mueller}
\affiliation{Volkswagen AG, Berliner Ring 2, 38440 Wolfsburg, Germany}
\author{Arne-Christian Voigt}
\affiliation{Volkswagen AG, Berliner Ring 2, 38440 Wolfsburg, Germany}
\author{Juan Miguel Arrazola}
\affiliation{Xanadu, Toronto, ON, M5G 2C8, Canada}

\begin{abstract}
Ionic pseudopotentials are widely used in classical simulations of materials to model the effective potential due to the nucleus and the core electrons. Modeling fewer electrons explicitly results in a reduction in the number of plane waves needed to accurately represent the states of a system. In this work, we introduce a quantum algorithm that uses pseudopotentials to reduce the cost of simulating periodic materials on a quantum computer. We use a qubitization-based quantum phase estimation algorithm that employs a first-quantization representation of the Hamiltonian in a plane-wave basis. We address the challenge of incorporating the complexity of pseudopotentials into quantum simulations by developing highly-optimized compilation strategies for the qubitization of the Hamiltonian. This includes a linear combination of unitaries decomposition that leverages the form of separable pseudopotentials. Our strategies make use of quantum read-only memory subroutines as a more efficient alternative to quantum arithmetic. We estimate the computational cost of applying our algorithm to simulating lithium-excess cathode materials for batteries, where more accurate simulations are needed to inform strategies for gaining reversible access to the excess capacity they offer. We estimate the number of qubits and Toffoli gates required to perform sufficiently accurate simulations with our algorithm for three materials: lithium manganese oxide, lithium nickel-manganese oxide, and lithium manganese oxyfluoride. Our optimized compilation strategies result in a pseudopotential-based quantum algorithm with a total Toffoli cost four orders of magnitude lower than the previous state of the art for a fixed target accuracy. 
\end{abstract}
\maketitle

\section{Introduction}

Quantum computing is being actively studied as a potential method to accurately simulate materials and support the development of next-generation lithium-ion batteries~\cite{ho2018promise, rice2021quantum, kim2022fault, clinton2022towards, batterypaper, Rubin2023, sunderhauf2023quantum}. The driving motivation is that quantum algorithms are uniquely positioned to perform highly-accurate simulations without incurring prohibitive computational costs~\cite{mcardle2020quantum}. Nevertheless, considerable progress still needs to occur on both hardware and algorithms to make this promise a reality. The main theoretical challenges are the high cost of implementing quantum algorithms and the difficulty of identifying the applications. 

There has been considerable progress in developing quantum algorithms for simulating the properties of molecules and materials. A variety of different strategies have been proposed, ranging from variational approaches designed for noisy hardware with few qubits~\cite{kandala2017hardware, yoshioka2020variational, anselmetti2021local, clinton2022towards, arrazola2022universal}, algorithms tailored for early fault-tolerant quantum computers~\cite{lin2022heisenberg, wan2022randomized, wang2022quantum, ding2022even}, and variants of quantum phase estimation that require the full capabilities of large-scale fault-tolerant quantum computers~\cite{babbush2016exponentially, reiher2017elucidating, kivlichan2018quantum, su2021fault, motta2021low, lee2021even}. Particular attention has been devoted to improving the efficiency of quantum phase estimation and Hamiltonian simulation algorithms, which has led to an overall cost reduction of several orders of magnitude~\cite{reiher2017elucidating}. This progress has been fueled by innovations such as qubitization~\cite{low2017optimal,low2019qubitization,berry2019qubitization}, Hamiltonian factorization techniques~\cite{kivlichan2018quantum, motta2021low,von2021quantum, lee2021even}, interaction-picture simulations~\cite{low2018hamiltonian,kieferova2019simulating, rajput2022hybridized}, improved Trotter bounds~\cite{childs2021theory}, and first quantization methods~\cite{babbush2018low, su2021fault}.

Simulating bulk materials presents additional challenges beyond those associated with simulating finite molecules~\cite{urban2016computational}. Arguably, among the most pressing ones is the frequent need to use large unit cells (supercells). For example, in the context of lithium-ion batteries, large supercells are needed to predict the most stable phases of cathode materials. Their energies calculated for different values of the lithium-ion concentration determine the voltage profile of the battery cell~\cite{van1998first, seo2016structural, mccoll2022transition}. The size of the supercell is even more critical for the simulation of chemical reactions at the electrode-electrolyte interface~\cite{wang2018review}. This results in systems with many hundreds of electrons requiring a very large number of plane-wave basis functions to achieve high-accuracy simulations~\cite{hine2009supercell}.

In this work, we introduce a quantum algorithm that uses ionic pseudopotentials (PPs) to reduce the cost of using quantum phase estimation to simulate material, mirroring a strategy widely used for density functional theory simulations~\cite{schwerdtfeger2011pseudopotential, lejaeghere2016reproducibility}. Replacing the bare Coulomb potential due to the nucleus and the core electrons with an effective potential leads to a substantial reduction in both the number of electrons and plane waves needed to accurately represent the system. However, this is accomplished at the price of a more complicated pseudopotential operator describing the interaction between the valence electrons and the effective ionic cores. This greatly complicates the implementation of the resulting quantum algorithm and can negate the benefits of reducing the number of electrons and plane waves.

We tackle this challenge by deriving highly-optimized implementation strategies for qubitization-based quantum phase estimation in first quantization. We focus on the Hartwigsen-Goedecker-Hutter (HGH) pseudopotentials~\cite{hartwigsen1998} and develop a tailored linear combination of unitaries decomposition for the pseudopotential term of the Hamiltonian. We also carefully engineer the compilation of the qubitization encoding to reduce the implementation cost. To avoid costly quantum arithmetic, our methods make frequent use of quantum read-only memory (QROM) subroutines. To the best of our knowledge, this is the first example of a quantum algorithm that can incorporate pseudopotentials. Overall, simulations with our algorithm requires orders-of-magnitude fewer plane waves to reach convergence for typical materials than comparable all-electron calculations. This results in circuit-depth reductions of many orders of magnitude for a given target accuracy. 

However, while it is well-understood how one might use quantum computers to simulate key properties of batteries, such as equilibrium voltages~\cite{batterypaper}, it has not yet been established which concrete class of battery simulation would benefit the most from the capabilities of quantum computers. After all, developing better lithium-ion batteries requires solving a multitude of problems, and scientists are already equipped with sophisticated simulation techniques that can be run on powerful supercomputers. It is therefore crucial to identify problems where the limitations of classical methods are most severe, and whose solution would be most impactful to battery development. 

We propose an application of quantum computers for batteries that we contend meet these criteria: simulating lithium-excess cathode materials~\cite{zhang2022pushing}. These materials offer an avenue to dramatically increase the energy density of state-of-the-art cathodes. With theoretical capacities that are roughly twice as high as commercial batteries, they would make the driving range and cost of electric vehicles competitive with internal combustion engines~\cite{li2020high}. Unfortunately, lithium-excess materials suffer from substantial capacity loss even after a single charging cycle~\cite{gent2017coupling}. The rapid degradation of lithium-excess materials responsible for the capacity loss has been attributed to irreversible structural transformations of the material, but the relationship between the proposed redox mechanisms and the observed transformations remains a topic of active discussion as these processes are difficult to probe experimentally~\cite{zhang2022pushing}. To better understand these processes, researchers rely on density functional theory simulations. However, the available density functionals are not accurate enough to single out the dominant mechanisms driving the materials structural changes~\cite{zhang2022pushing}. This makes it difficult to develop solutions for capacity loss in lithium-excess cathode materials using classical computer simulations. 

To assess the potential for quantum computers to provide a solution to this problem, we perform a detailed estimation of the resources required to implement our quantum algorithm as applied to three lithium-excess cathode materials: lithium manganese oxide ($\text{Li}_2\text{MnO}_3$), lithium nickel-manganese oxide ($\text{Li}[\text{Li}_{0.17}\text{Ni}_{0.25}\text{Mn}_{0.58}]\text{O}_2$), and lithium manganese oxyfluoride ($\text{Li}_{0.75}\text{MnO}_2\text{F}$)~\cite{lim2015origins, eum2020voltage, mccoll2022transition}. In each of these cases, the use of ionic pseudopotentials allows a reduction of the number of electrons by a factor of two, and a reduction in the number of plane waves by roughly three orders of magnitude compared to all-electron simulations for a fixed target accuracy. This results in a total Toffoli cost for the algorithm that is about four orders of magnitude lower than the previous state-of-the-art~\cite{su2021fault}. 

The rest of this work is organized as follows. \cref{sec:background} provides background information on first-quantization quantum algorithms for electronic structure, the theory of ionic pseudopotentials, and the basic properties of quantum read-only memories. Our quantum algorithm is described in \cref{sec:algorithm}, outlining the linear combination of unitaries and qubitization subroutines that constitute the main technical contribution of this work. This is complemented with a detailed error analysis in \cref{sec:error_analysis_and_eff_value_lambda} and a calculation of the qubit and gate cost of the full algorithm in \cref{sec:cost}. We then study the application of the algorithm to the simulation of lithium-excess cathode materials in \cref{sec:applications}.

\section{Background}\label{sec:background}

This work is an interdisciplinary effort covering topics across computational chemistry, quantum computing, and lithium-ion batteries. In an effort to make this manuscript self-contained, this section provides background information on key concepts that will be used throughout.  

\subsection{The plane-wave electronic Hamiltonian in first quantization}
\label{sec:hamilt}
The ultimate goal of the quantum algorithm presented here is to solve the electronic structure problem:
\begin{equation}
H \Psi_0(\bm{r}_1, \dots, \bm{r}_\eta) = E_0 \Psi_0(\bm{r}_1, \dots, \bm{r}_\eta),
\label{eq:es}
\end{equation}
where $H$ is the Hamiltonian of $\eta$ interacting electrons, and $\Psi_0$ and $E_0$ are the ground-state wave function and energy, respectively. 

In the Born-Oppenheimer approximation~\cite{born1927quantum}, the electronic Hamiltonian, $H$, is given by
\begin{equation}
H = T + U + V,
\label{eq:AE_hamiltonian}
\end{equation}
where $T$ is the total kinetic energy operator, $U$ is the Coulomb potential due to the nuclei and $V$ is the electron-electron interaction term~\cite{kohanoff2006electronic}. We have listed the symbols used in this paper in \cref{appsec:list_notations}.

Plane-wave functions are a natural basis set to represent the electronic states in periodic materials. They can be used to encode the translational symmetry of crystal structures and allow us to derive closed-form expressions for the Hamiltonian matrix elements. Plane-wave functions are defined as
\begin{equation}
\varphi_p(\bm{r}) = \frac{1}{\sqrt{\Omega}} e^{i \bm{G}_p \cdot \bm{r}},
\label{eq:pw}
\end{equation}
where $\Omega$ is the volume of the material's unit cell and $\bm{G}_p$ is the reciprocal lattice vector
\begin{equation}
 \bm{G}_p = \left[ \sum_{i=1}^3 p_i b_{i_x}, \sum_{i=1}^3 p_i b_{i_y}, \sum_{i=1}^3 p_i b_{i_z} \right],
 \label{eq:pw_g}
 \end{equation}
 where $\bm{b}_1, \bm{b}_2, \bm{b}_3$ are the primitive vectors of the reciprocal lattice~\cite{ashcroft1976solid}. For a total number of plane waves $N$, the integer vectors $\bm{p}$ contained in the set
 \begin{equation}
 \mcG=\left[-\frac{N^{1/3}}{2}+1, \frac{N^{1/3}}{2}-1\right]^3,
\label{eq:pw_p}
\end{equation}
define a uniform grid of points in the reciprocal lattice.

Previous works~\cite{babbush2018low,babbush2019quantum, su2021fault, batterypaper} have argued in favour of using first quantization techniques to accommodate the large number of plane waves that are required for accurate simulations. We follow that strategy in this work. In first quantization, the plane wave representation of the operators $T$, $U$ and $V$ are given by~\cite{su2021fault, batterypaper}:
\begin{align}
T  =& \sum_{i=1}^\eta\sum_{p\in \mathcal{G}}\frac{\|\bm{G}_p\|^2}{2}\ket{\bm{p}}\bra{\bm{p}}_i, \label{eq:T}\\
U  =& -\frac{4\pi}{\Omega}\sum_{i=1}^\eta\sum_{q\in \mathcal{G}} \nonumber \\
& \frac{\sum_{I=1}^L Z_I e^{i\bm{G}_\nu \cdot \bm{R}_I}}{\|\bm{G}_\nu\|^2}\ket{\bm{q-\nu}}\bra{\bm{q}}_i \label{eq:U},\\
V  =& \frac{2\pi}{\Omega}\sum_{i\neq j}^\eta\sum_{p,q\in \mathcal{G}} \nonumber \\
& \sum_{\substack{\nu\in \mathcal{G}_0 \\ (\bm{p} + \bm{\nu}) \in \mathcal{G} \\ (\bm{q}-\bm{\nu}) \in \mathcal{G}  } } \frac{1} {\|\bm{G}_\nu\|^2}\ket{\bm{p+\nu}}\bra{\bm{p}}_i \ket{\bm{q}-\bm{\nu}}\bra{\bm{q}}_j \label{eq:V},
\end{align}
where $Z_I$ and $\bm{R}_I$ are respectively the atomic number and position of the $I$th atomic species, and $L$ denotes the number of atoms in the unit cell. In \cref{eq:U,eq:V},   $\bm{G}_\nu=\bm{G}_q-\bm{G}_p$ and $\bm{G}_\nu = \bm{G}_p - \bm{G}_s = \bm{G}_r - \bm{G}_q$, respectively, and $\mathcal{G}_0 = \mathcal{G} \setminus (0, 0, 0)$. The qubit representation of the $N$ plane waves uses 
\begin{align}
n_p = \lceil \log(N^{1/3}+1) \rceil
\end{align}
qubits to encode each component of the plane wave vector. Thus a total of $3\eta n_p$ qubits are required for the system register.

\subsection{Ionic pseudopotentials}
\label{sec:pp}
Performing accurate {\it all-electron} simulations of supercell structural models of battery materials is hampered by the huge number of plane waves that are needed to represent the core states and the valence states near the nucleus as sketched in~\cref{fig:pp}. Different strategies such as the augmented and orthogonalized plane wave methods~\cite{slater1937wave, wills2010full, herring1940new} have been proposed to overcome this limitation. However, a key step to retain the advantages of plane waves for materials simulations was taken by Phillips, Kleinman and Antoncik (PKA)~\cite{phillips1959new, antonvcik1959approximate}. Crucially, it follows from the PKA transformation that the nuclear potential and the core electrons can be replaced by an effective potential known as a {\it pseudopotential} that produces the same energies of the valence states. Furthermore, the associated pseudo wave functions superimpose the true valence wave functions outside the core region (see~\cref{fig:pp}), and can be accurately represented using a significantly smaller number of plane waves. This is an excellent approximation since core electrons populate deep energy states that do not influence neither the chemical bonding nor the redox processes in battery materials.

\begin{figure}[t]
\centering
\includegraphics[width=0.8 \columnwidth]{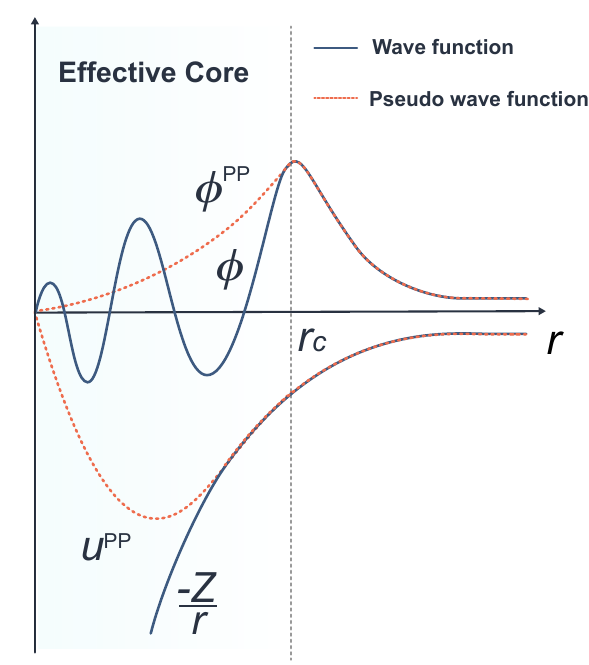}
\caption{Representation of the oscillating character of the all-electron wave function, $\phi$, versus the smooth behaviour of the pseudo wave function ,$\phi^\text{PP}$, in the core region $r<r_c$. The pseudo wave function is a solution of the pseudopotential, $u^\text{PP}$, and superimposes the valence wave function outside the core region. The two lower curves sketch the bare nuclear potential $-Z/r$ (solid line) and the pseudopotential $u^\text{PP}$ (dotted line).}
\label{fig:pp}
\end{figure}

Here, we focus on the use of \emph{ionic pseudopotentials} (PPs) which have proven to be highly accurate for simulating materials~\cite{lejaeghere2016reproducibility, garrity2014pseudopotentials, bennett2012discovery}. Different pseudization schemes, including the Hartwigsen, Goedecker and Hutter (HGH) PPs adopted in this work and described in more details in Sec.~\ref{ssec:hgh}, have been extensively benchmarked for a large set of elemental crystals~\cite{lejaeghere2016reproducibility}. Modern PPs exhibit small deviations (1-2.2 meV/atom) of the calculated equation of state with respect to analogous all-electron results. Furthermore, the transferability of ionic PPs have also been shown to reproduce the lattice constants of more complicated materials e.g., transition-metal oxides, with maximum root mean squared errors of the order of 2$\%$~\cite{garrity2014pseudopotentials, bennett2012discovery}. Interestingly, the authors in Ref.~\cite{garrity2014pseudopotentials} have noted that the differences between the PPs and the all-electron calculations are often comparable with the numerical uncertainties in the all-electron results themselves.

Ionic PPs are typically generated from all-electron atomic calculations performed using density functional theory. Due to the spherical symmetry of the Coulomb potential, the pseudo wave functions, $\phi_{lm}(\bm{r})=\phi_l(r)Y_{lm}(\phi, \theta)$, are eigenstates of the angular momentum operator and characterized by the quantum numbers $(l, m)$. Starting from a pseudo wave function ansatz, $\phi_l(r)$, the effective model potential, $u_l(r)$, is found by inverting the radial Schr\"odinger equation~\cite{troullier1990straightforward, hamann1979norm, kerker1980non, bachelet1982pseudopotentials}. Norm-conserving  PPs~\cite{troullier1991efficient} are obtained from pseudo wave functions enclosing the same charge as the true wave function in the core region $r<r_c$, where $r_c$ denotes a core radius around the nucleus. Finally, the ionic PP is obtained by removing the screening effects due to the valence electrons in the atom. This is referred to as ``unscreening'' the pseudopotential, which makes the PP transferable to different atomic environments~\cite{louie1982nonlinear, engel2001role}.  

The general form of the pseudopotential operator is \cite{kohanoff2006electronic, martin2020electronic}
\begin{equation}
u^\text{PP} = u^\text{loc} + u^\text{SL},
\label{eq:pp_op_2}
\end{equation}
where $u^\text{loc}:=u^\text{loc}(r)$ is a local potential, i.e., obtained by evaluating a simple function at point $r$. It represents a screened Coulomb potential which joins smoothly the all-electron atomic potential at some radius $r<r_c$. The second operator in~\cref{eq:pp_op_2} is defined as
\begin{equation}
u^\text{SL} = \sum_{l=0}^{l_\text{max}} \sum_{m=-l}^l \Delta u_l(r)\ket{lm} \bra{lm},
\label{eq:sl-pot}
\end{equation}
where $l_\text{max}$ is the maximum angular momentum of the core electrons and $\langle \bm{r} \vert lm \rangle = Y_{lm}(\phi, \theta)$ denotes the spherical harmonics. The $l$-dependent potential, $\Delta u_l(r) = u_l(r) - u^\text{loc}(r)$, is a short-ranged potential vanishing beyond the core radius. Moreover, in the asymptotic limit, $r \rightarrow \infty$, the full pseudopotential, $u^\text{PP}$, behaves as $-Z_\mathrm{ion}/r$, where $Z_\text{ion}=Z-\eta_\text{core}$ is the effective charge of the ionic core and $\eta_\text{core}$ is the number of core electrons.

The operator $u^\text{SL}$ in~\cref{eq:pp_op_2} has a semi-local character since its action on a given basis function, $\braket{ \bm{r}|\bm{q}} = \varphi_q(r, \phi, \theta)$, is local in the radial coordinates but involves an integral over the angular variables $(\phi, \theta)$~\cite{martin2020electronic}. Computing its plane wave matrix elements $u_{pq}^\text{SL}$ requires evaluating the radial integral \cite{kohanoff2006electronic},
\begin{equation}
I_{pq} = \int_0^\infty dr~r^2 j_l(G_pr) \Delta u_l(r) j_l(G_qr),
\label{eq:sl_me}
\end{equation}
where $j_l$ denotes the spherical Bessel functions and $G_p=\|\bm{G}_p \|$. 
The number of such integrals scales as $\mathcal{O}(LN^2)$, where $L$ is the number of atoms in the material's unit cell and $N$ is the total number of plane waves. Typically, $N$ can be very large in actual simulations and computing these matrix elements becomes computationally expensive~\cite{kohanoff2006electronic}.

To reduce this computational cost, Kleinman and Bylander (KB)~\cite{kleinman1982efficacious} proposed the separable pseudopotentials by expressing the radial operator $\Delta u_l(r)$ in a form that is separable in the radial variables. The KB construction was further modified by Bl\"och~\cite{blochl1990} to construct the non-local (NL) pseudopotential
\begin{align}
u^\mathrm{PP} & = u^\text{loc} + u^\text{NL},
\label{eq:kbb}
\end{align}
with the operator $u^\text{NL}$ given by 
\begin{equation}
u^\text{NL} = \sum_{l=0}^{l_\text{max}}\sum_{m=-l}^l \frac{\ket{\xi_{lm}} \bra{\xi_{lm}} }{\bra{\xi_{lm}} \phi_{lm} \rangle}.
\label{eq:nl}
\end{equation}
In~\cref{eq:nl} $\langle \bm{r} \vert \phi_{lm}\rangle=\phi_{l}(r)Y_{lm}(\phi, \theta)$ is a pseudo wave function solution of the model potential $u_l(r)$, and $\langle \bm{r} \vert \xi_{lm} \rangle = \xi_l(r) Y_{lm}(\phi, \theta)$ are projectors defined as
\begin{equation}
\xi_{lm}(\bm{r}) = \left\{ \varepsilon_l-\left[ -\frac{1}{2}\nabla^2 + u^\mathrm{loc}(r) \right] \right\} \phi_{lm}(\bm{r}),
\label{eq:nl_1}
\end{equation}
where $\varepsilon_l$ is the all-electron reference energy associated with the pseudo wave function $\phi_{lm}(\bm{r})$. Note from ~\cref{eq:nl} that computing the matrix elements of the separable potential $u_{pq}^\text{NL}$, as opposed to the matrix elements of a semi-local operator (\cref{eq:sl_me}), requires only to evaluate the product of the projection operations
\begin{equation}
\bra{\xi_{lm}} \varphi_q \rangle = \int d\bm{r}~\xi_{lm}^*(\bm{r}) \varphi_q(\bm{r}),
\label{eq:pp_kb_2}
\end{equation}
which scales linearly with the number of plane waves.

The non-local operator in~\cref{eq:nl} can be generalized to use two or more projectors per angular momentum quantum numbers \cite{vanderbilt1990},
\begin{equation}
u^\mathrm{NL} = \sum_{l=0}^{l_\text{max}}\sum_{m=-l}^l \left[ \sum_{i,j} B_{ij}\ket{\beta_i} \bra{\beta_j} \right]_{lm},
\label{eq:gnl}
\end{equation}
where the matrix $B_{ij}=\bra{\phi_i} \xi_j \rangle$ is used to define the projectors $\ket{\beta_i} = \sum_j (B^{-1})_{ji} \ket{\xi_j}$. Pseudopotential operators of this form are called \emph{generalized separable pseudopotentials}, and are routinely used in classical electronic structure calculations of materials \cite{lejaeghere2016reproducibility}. A further generalization of the non-local operator in~\cref{eq:gnl} is obtained by relaxing the norm conservation condition. This results in the so-called ultrasoft pseudopotentials~\cite{vanderbilt1990}, which we do not consider here.

\subsection{Qubitization-based quantum phase estimation}

The quantum algorithm presented in this work is a qubitization-based quantum phase estimation (QPE) algorithm for computing ground-state energies of periodic materials. In contrast to other techniques, qubitization does not introduce any further approximations in implementing the unitary that encodes the eigenvalues of the Hamiltonian. This makes it particularly appealing when the accuracy of the simulation is paramount. We briefly review the qubitization method and refer to~Refs.~\cite{low2019qubitization,batterypaper,su2021fault} for further details.

To qubitize a Hamiltonian, we first write it as a linear combinations of unitaries (LCU), 
\begin{align} \label{eq:simple_LCU}
H = \sum_\ell \alpha_\ell H_\ell,
\end{align}
where each $H_\ell$ is a unitary and the coefficients $\alpha_\ell>0$ are referred to as (unnormalized) \textit{selection probabilities}. The specific choice of an LCU has a large impact on the cost of qubitization, especially through the parameter $\lambda= \sum_\ell \alpha_\ell$~\cite{loaiza2022reducing}. 

We then define the qubitization operator
\begin{align}\label{eq:qubitization_Q}
Q=(2\ket{\bm{0}}\bra{\bm{0}}-\mathbbm{1})\text{PREP}_H^{\dagger}\text{SEL}_H\text{PREP}_H,
\end{align}
with the prepare and select unitaries given by
\begin{align}\label{eq:prep_H_formula}
\text{PREP}_H\ket{\bm{0}}\ket{\psi} &= \left(\sum_\ell\sqrt{\frac{\alpha_\ell}{\lambda}}\ket{\ell}\right)\ket{\psi}, \\
\label{eq:sel_H_formula}
\text{SEL}_H &= \sum_\ell\ket{\ell}\bra{\ell}\otimes H_{\ell},
\end{align}
Note that the reflection in $Q$ and $\PREP$ act on an auxiliary register $\ket{\bm{0}}$. The latter prepares the so-called $\PREP$ state $\sum_\ell\sqrt{\frac{\alpha_\ell}{\lambda}}\ket{\ell}$ with amplitudes given by the selection probabilities. It does not alter the system register $\ket{\psi}$, which is acted upon by $\SEL$ that applies the unitaries $H_\ell$. These subroutines satisfy the block-encoding equation
\begin{align}\label{eq:block_encoding_H}
    \bra{\bm{0}}\PREP_H^\dagger \ \SEL_H\ \PREP_H \ket{\bm{0}} = \frac{H}{\lambda}.
\end{align}
The operator $Q$ is block-diagonal, with each block $Q_k$ corresponding to a two-dimensional subspace $W_k$ spanned by an orthonormal basis $\text{span}\{\ket{\bm{0},\Phi_k}, \ket{\rho_k}\}$, where $\ket{\Phi_k}$ is an eigenstate of $H$ with eigenvalue $E_k$ and $\ket{\rho_k}$ is orthogonal to $\ket{\bm{0},\Phi_k}$. The term qubitization refers to these effective qubit subspaces. In its eigenbasis, $Q_{k}$ can be written as
\begin{align} \label{eq:qwalk_eigenvalues}
Q_{k} &= e^{i\theta_k}\ket{\theta_k}\bra{\theta_k}+e^{-i\theta_k}\ket{-\theta_k}\bra{-\theta_k},
\end{align}
where $\ket{\pm \theta_k}$ are the eigenstates, with $\theta_k=\arccos(E_k/\lambda)$. Thus, by applying QPE on $Q$ with an initial state $\ket{\bm{0}}\ket{\Psi_0}$, where $\ket{\bm{0}}\ket{\Psi_0}=\alpha\ket{\theta_0}+\beta\ket{-\theta_0}$ for some $\alpha,\beta$, we always recover the ground state and its energy since $\cos(\pm\theta_0)=E_0/\lambda$.

In the QPE algorithm, the unitary $Q$ is controlled on the state of auxiliary qubits, which increases the Toffoli cost of the algorithm. To avoid this, as shown in Ref.~\cite{babbush2018encoding}, one can use a reflection to have the inverse unitary applied when the auxiliary qubit is in state $\ket{0}$. The only requirement for this modification to work is for $\SEL$ to be self-inverse, which our algorithm satisfies.

When acting on an initial state $\ket{\psi}$, the QPE algorithm outputs an estimate of the ground-state energy $E_0$ with probability $p_0=|\braket{\psi_0|\psi}|^2$, where $\ket{\psi_0}$ is the ground state. The QPE routine needs to be repeated $O(1/p_0)$ times on average to retrieve $E_0$ with high probability. It is thus necessary to prepare an initial state with a sufficiently large overlap with the ground state. This is the initial state preparation problem: a crucial and daunting challenge for quantum algorithms. Strategies for preparing initial states have been studied in Refs.~\cite{sugisaki2016quantum, tubman2018postponing, von2021quantum, lee2022there}, and Ref.~\cite{batterypaper} described a method for preparing a Hartree-Fock state for periodic materials in first quantization. While we acknowledge the importance of developing better methods for initial state preparation, in this work we focus on the problem of reducing the cost of QPE.

\subsection{Quantum read-only memory (QROM)}\label{ssec:background_qrom}

Our main use of QROM in the quantum algorithm is to prepare arbitrary states of few qubits. There are a variety of methods in the literature for this task; we refer to~\cite{mcardle2022quantum} for an overview of such techniques. A QROM, or a data-lookup oracle, is an operator $O$ that reads a register $\ket{x}$ and outputs a corresponding bitstring $\ket{\theta_x}$ into an auxiliary register as: $O\ket{x}\ket{\bm{0}} = \ket{x}\ket{\theta_x}$~\cite{babbush2018encoding}. The output bitstring $\theta_x$ is precomputed and available in a data-lookup table. The non-Clifford gate cost of QROM is exponential in the size of $x$ and polynomial in the size of $\theta_x$. Parallelization can decrease the depth at the expense of using more qubits~\cite{babbush2018encoding,low2018trading}.\\

We follow~\cite{low2018trading} in our description of state preparation using QROM. Consider the target state $\ket{\psi}=\sum_x a_x \ket{x}$, where we assume $a_x\geq 0$ for simplicity. For any bit-string $y$ of length $w\le n$, where $n$ is the total number of qubits, denote by $p_y$ the probability that the first $w$ qubits of $\ket{\psi}$ are in state $\ket{y}$. Define also $ \cos(\theta_y)= \sqrt{p_{y0}/p_y}$, where $y0$ is the bitstring $y$ followed by $0$, and the QROM oracles $O_w \ket{y}\ket{\bm{0}} = \ket{y}\ket{\theta_y}$, outputting the classically precomputed and tabulated $\theta_y$ up to $b$ bits of precision for each $w$. The QROM oracles can be used iteratively to prepare any state $\ket{\psi}=\sum_{x}a_x \ket{x}$ as follows. By induction, for $1\le w\le n$, do:
\begin{align}\label{eq:qrom_superposition_alg_maintxt}
\begin{split}
&\ket{\psi_w} = \sum_{y\in\{0,1\}^w}\sqrt{p_y}\ket{y}\ket{0}\ket{\bm{0}}\\
&\underset{O_w}{\mapsto} \sum_{y\in\{0,1\}^w}\sqrt{p_y}\ket{y}\ket{0}\ket{\theta_y}
\\
&\underset{R_{w}}{\mapsto} \sum_{y\in\{0,1\}^w}\sqrt{p_y}\ket{y}\left(\cos(\theta_y)\ket{0}+\sin(\theta_y)\ket{1}\right)\ket{\theta_y}
\\
&\underset{O^\dag_w}{\mapsto} \sum_{y\in\{0,1\}^{w+1}}\sqrt{p_y}\ket{y}\ket{\bm{0}}=\ket{\psi_{w+1}}\ket{\bm{0}},
\end{split}
\end{align}
where $R_{w}$ is a one-qubit rotation on the $(w+1)$-th qubit controlled on the state of the auxiliary system, and as a slight abuse of notation we use $\ket{\bm{0}}$ to denote all-zero states of different number of qubits. While we assumed non-negative amplitudes $a_x \ge 0$, the complex phases of $\ket{\psi}$ can be implemented by storing $\phi_x=\text{arg}[a_x/|a_x|]$ and applying it in the final iteration. In our applications, $\phi_x =\pm 1$, meaning the phase is always real and the last step is a simple $Z$ gate. For future reference throughout the text, this entire state preparation procedure is called \qromalgo.

The precision of the rotation angle is the main source of error, and we have the following result for the total error, proved in \cref{app:QROM_superposition}:
\begin{customlem}{\ref{lem:qrom_superposition_error}}\label{lem:qrom_superposition_error_repetition}
The error in the state preparation \qromalgo~is $2^{-b}\pi n$.
\end{customlem}

There are three different types of QROM oracles that can be used in \qromalgo~as proposed in~\cite{low2018trading}. Two of these, called \textsc{Select} and \textsc{SelSwapDirty}, are of interest to us. The latter is our terminology for the oracle described in~\cite[Fig. 1d]{low2018trading}, also called QROAM by \cite{berry2019qubitization}. $\textsc{Select}$ is mostly used when $n$ is small while $b$ is large. One important property of \textsc{SelSwapDirty} is the space-depth trade-off that it offers. With the help of dirty qubits, i.e., qubits that do not need to be initialized to any specific state, the depth of the circuit can be lowered by parallelizing the controlled-SWAP gates used in the \textsc{Swap} subroutine of \textsc{SelSwapDirty}. Furthermore, the dirty qubits are returned to their initial state, therefore any qubit not undergoing simultaneous computation in the quantum circuit can be borrowed as a dirty qubit.  We briefly overview the cost of these routines in \cref{table:qroms_cost}, and we refer the reader to \cref{app:QROM_all_apps} for more details. \\
{
	\renewcommand{\arraystretch}{1.1}
	\begin{table}
 \centering
	\resizebox{\columnwidth}{!}{
		\begin{tabular}{|Sc|Sc|Sc|Sc|}
  \hline
			\shortstack{Operation\\$\;$}&\shortstack{Additional qubits\\$\;$}&\shortstack{Toffoli Depth\\$\;$}&\shortstack{Toffoli count\\ $\le\cdot+\mathcal{O}(\log\cdot)$}\\
			\hline
			\textsc{Select}&$b+\lceil\log{N}\rceil$&$N$&$N$
			\\\hline
			{\textsc{SelSwapDirty}}&$b(\beta+1) +\lceil\log{N}\rceil$&$2\lceil\frac{N}{\beta}\rceil+3\lceil\log{\beta\rceil}$&$2\lceil\frac{N}{\beta}\rceil+3b\beta$\\\hline
		\end{tabular}}
		\caption{(\cite[Table II]{low2018trading}) Gate and qubit cost of \textsc{Select} and \textsc{SelSwapDirty} QROMs. The space-depth trade-off is determined by $\beta\in[1,N]$. Note that $b\beta$ qubits of the~\cref{fig:qroms_circuit}b implementation are dirty, while $b+\lceil \log(N) \rceil $ are clean qubits.}
		\label{table:qroms_cost_main_text}
	\end{table}
}

\section{Quantum algorithm}\label{sec:algorithm}
We now describe the pseudopotential-based quantum phase estimation algorithm. We begin with the construction of the plane-wave representation of the pseudopotential operators describing the ionic cores in the material, where we derive closed-form expressions for the matrix elements of the local and non-local potentials (see \cref{sec:pp_me}). We then proceed to describe the LCU decomposition of the Hamiltonian and the implementation of the qubitization operator. These procedures exploit the structure of the pseudopotential to optimize the qubitization of the modified operator $U$, as discussed in \cref{ssec:lcu}.

\subsection{Plane wave matrix elements of the pseudopotential operator}
\label{ssec:hgh}
We focus on the description of the separable Gaussian pseudopotentials proposed by Hartwigsen, Goedecker and Hutter (HGH)~\cite{hartwigsen1998}. The HGH pseudopotentials are relatively easy to define and have proven to be transferable and accurate. 

The HGH local potential is defined as
\begin{equation}
u^\mathrm{loc}(r) = -\frac{Z_\mathrm{ion}}{r}\mathrm{erf(\alpha r)} + \sum_{i=1}^4 C_i (\sqrt{2}\alpha r)^{2i-2}~e^{-(\alpha r)^2}
\label{eq:hgh_loc}
\end{equation}
where $\mathrm{erf}(\alpha r)$ is the error function with $\alpha=\frac{1}{\sqrt{2}r_\mathrm{loc}}$, $r_\mathrm{loc}$ is a local radius parameter giving the charge distribution in the core, and the $C_i$ are tabulated coefficients~\cite{hartwigsen1998}. On the other hand, the non-local part $u^\text{NL}$ is given by a sum of separable terms~\cite{hartwigsen1998}
\begin{align}
u^\text{NL} = & \sum_{lm} \sum_{i, j = 1}^3 \vert \beta_i^{(lm)} \rangle B_{ij}^{(l)} \langle \beta_j^{(lm)} \vert,
\label{eq:hgh_nl}
\end{align}
where $\langle \bm{r} \vert \beta_i^{(lm)} \rangle = \beta_i^{(l)}(r)Y_{lm}(\phi, \theta)$. The radial functions $\beta_i^{(l)}(r)$ are Gaussian-type projectors given by
\begin{equation}
\beta_i^{(l)}(r) = A_l^i r^{l+2i-2} \mathrm{exp}\left[ -\frac{1}{2}\left(\frac{r}{r_l}\right)^2 \right],
\label{eq:hgh_gaussians}
\end{equation}
where the $A_l^i$ is a constant defined in \cref{ssec:non_local_me} and the radii $r_l$ give the range of the $l$-dependent projectors. The optimized values of the coefficients $B_{ij}^{(l)}$ and the radii $r_l$ are reported in Ref.~\cite{hartwigsen1998}.\\

The plane-wave matrix elements of the local and non-local components of the HGH pseudopotentials are derived in \cref{sec:pp_me}. For an ion located at the coordinates $\bm{R}$, the matrix elements of the local potential are given by
\begin{align}
&u_{pq}^\mathrm{loc}(\bm{R}) =  \frac{4\pi}{\Omega} e^{i \bm{G}_\nu \cdot \bm{R}}~e^{-(G_\nu r_\text{loc})^2/2} \bigg\{-\frac{Z_{\mathrm{ion}}}{G_\nu^2} \nonumber\\
&+ \frac{\sqrt{\pi}}{2} \Big[ C_1 r_\text{loc}^3 + C_2(3r_\text{loc}^3 - 5r_\text{loc}^5 G_\nu^2) \nonumber\\
&+ C_3(15r_\text{loc}^3 - 10r_\text{loc}^5 G_\nu^2 + r_\text{loc}^7 G_\nu^4) \nonumber \\
&+ C_4(105r_\text{loc}^3 - 105r_\text{loc}^5 G_\nu^2 + 21 r_\text{loc}^7 G_\nu^4 - r_\text{loc}^9 G_\nu^6) \Big] \bigg\}
\label{eq:uloc_me}
\end{align}
where $\bm{G}_\nu = \bm{G}_q-\bm{G}_p$.\\

The matrix elements of the non-local operator in~\cref{eq:hgh_nl} are derived in \cref{ssec:non_local_me}. Typically, electronic structure calculations are performed using one or two projectors. For the sake of simplicity, we consider the case of one projector per angular momentum $(i=j=1)$. In this case, the plane wave matrix elements for the HGH non-local potential are given by
\begin{align}
u^\mathrm{NL}_{pq}(\bm{R}) = &\frac{4\pi}{\Omega}  e^{i \bm{G}_\nu \cdot \bm{R}} \Bigg\{ 
4B_0 r_0^3~e^{-(G_p^2+G_q^2)r_0^2/2} \nonumber\\
& + \frac{16B_1 r_1^5}{3}~(\bm{G}_p \cdot \bm{G}_q)~e^{-(G_p^2+G_q^2)r_1^2/2}
\nonumber \\
& + \left[ \frac{32B_2 r_2^7}{15}~(\bm{G}_p \cdot \bm{G}_q)^2 \right. \nonumber\\ 
& \left. + \frac{32B_2 r_2^7}{45}~(G_p G_q)^2\right]e^{-(G_p^2+G_q^2)r_2^2/2} \Bigg\},
\label{eq:unl_me_1}
\end{align}
where the coefficient $B_l:=B_{11}^{(l)}$. To see roughly how the projector expression in \cref{eq:hgh_nl} could lead to the expression above, we show one example of how the equation decomposes to a sum of projections. Consider the first term in \cref{eq:unl_me_1} $4B_0 r_0^3~e^{-(G_p^2+G_q^2)r_0^2/2}\ket{\bp}\bra{\bq}$. When summed over $\bp,\bq$, this can be expressed as a scalar multiple of the projection onto the Gaussian superposition state
\begin{align}\label{eq:psi_I_o_example}
    \sum_{\bp \in \mcG }e^{-G_p^2r_0^2/2}\ket{\bp}.
\end{align}
A similar rewriting applies to other terms in \cref{eq:unl_me_1}, and involves projection onto (derivatives) of Gaussian superpositions. See \cref{app:U_NL_LCU} for more details.\\

\subsection{The pseudopotential Hamiltonian}
\label{sec:pp_u}
By including the ionic pseudopotentials, the all-electron problem defined by the Hamiltonian in~\cref{eq:AE_hamiltonian} transforms into a valence-only electron problem where $\eta_\text{val}=\sum_{I=1}^L Z_{\text{ion}_I}$. This results in a substantial reduction of the total number of electrons, plane waves $N$, and the overall norm of the Hamiltonian describing the valence electrons.

In this approach, the expressions for the operators $T$ and $V$ (\cref{eq:T,eq:V}) remain formally identical. However, the operator $U$ accounting for the electron-nuclei interactions in the all-electron case needs to be defined using the effective pseudopotentials describing the ionic cores:
\begin{align}
U = \sum_{i=1}^\eta \sum_{I=1}^L \frac{-Z_I}{\vert \bm{r}_i - \bm{R}_I \vert } \rightarrow U = \sum_{i=1}^{\eta_\text{val}} \sum_{I=1}^L u^\text{PP}(\bm{r}_i, \bm{R}_I).
\label{eq:u_pp}
\end{align}
Using~\cref{eq:kbb}, the plane-wave representation of the operator $U$ becomes
\begin{align}
& U = U_\text{loc}+U_\text{NL}, \\
& U_\text{loc} := \sum_{i=1}^{\eta_\text{val}} \sum_{p,q=1}^N \sum_{I=1}^L u_{pq}^\text{loc}(\bm{R}_I) \vert \bm{p} \rangle \langle \bm{q} \vert_i, \\
& U_\text{NL} := \sum_{i=1}^{\eta_\text{val}} \sum_{p,q=1}^N \sum_{I=1}^Lu_{pq}^\text{NL}(\bm{R}_I) \vert \bm{p} \rangle \langle \bm{q} \vert_i,
\end{align}
where $u_{pq}^\text{loc}(\bm{R}_I)$ and $u_{pq}^\text{NL}(\bm{R}_I)$ are the matrix elements given by \cref{eq:uloc_me,eq:unl_me_1}, respectively. This yields the {\it pseudopotential} Hamiltonian
\begin{align}\label{eq:pseudopotential_hamiltonian}
 H = T+ U_{loc}+U_{NL} + V. 
 \end{align}

\subsection{Linear Combination of Unitaries}\label{ssec:lcu}

We study the structure of the matrix entries of each of the four operators in the pseudopotential Hamiltonian when deriving the LCU for $H$ as in \cref{eq:pseudopotential_hamiltonian}. There are two main differences from the setting in~\cite{su2021fault}: (i) We consider the general case of non-cubic lattices, that is, the primitive vectors $\bm{a_i}$ have different lengths and are not orthogonal. Hence, the reciprocal lattice vector
\begin{align}
\bG_p &= \sum_{\omega=1}^3 p_\omega \bb_\omega = \nonumber\\
& \frac{2\pi}{\Omega}\big (p_1(\bm{a}_2 \times \bm{a}_3)+p_2(\bm{a}_3 \times \bm{a}_1)+ p_3(\bm{a}_1 \times \bm{a}_2)\big)\label{eq:recip_vectors},
\end{align}
can no longer be substituted by $\bG_p = 2\pi \Omega^{-1/3}\bp$, and (ii) While the local term $U_{loc}$ resembles the operator $U$ in the all-electron setting, the non-local term $U_{NL}$ has a radically different structure. Dealing with the complexity of this non-local term is one of the biggest challenges we face in deriving an efficient decomposition.

We briefly show how the LCUs are derived for $T,V,$ and $U_{loc}$. The case of $U_{NL}$ is explained at a high level, with details appearing in \cref{app:U_NL_LCU}. Hereafter, we make use of the following conventions. We use $\eta:= \eta_{\val}$ to denote the number of valence electrons in the pseudopotential Hamiltonian. To make our future discussions more precise, we define the compound index $I=(t, t_j)$ where $t$ indicates the atomic species of the $I$-th ionic core, and $t_j$ enumerate the cores of that type (reading for example as the `third oxygen atom'). To make the reference to nuclei $I$ more explicit, we may sometimes use $t_I$ and $t_{I,j}$.\\

\textit{LCU for $T$.} Recall that $T$ is a diagonal operator with entries $\frac{G_p^2}{2}$. We rewrite 
\begin{equation}
    G_p^2 = \bG_p \cdot \bG_p = \sum_{\omega,\omega'} \langle \bb_\omega, \bb_{\omega'}\rangle p_\omega p_{\omega'},
\end{equation} 
and take the binary expansion 
\begin{align}
p_\omega &= (-1)^{p_{\omega,n_p-1}} \sum_r 2^r p_{\omega,r}, \nonumber \\  p_{\omega'} &=(-1)^{p_{\omega',n_p-1}}\sum_s 2^s  p_{\omega',s},
\end{align}
where $p_{\omega',s}$ is the $s$-th bit of $p_{\omega}$. Note the signed integer representation, where the $(n_p-1)$-th bit determines the sign of the coordinate. One can apply the same trick as in the orthonormal lattice case \cite{su2021fault}, rewriting $p_{\omega,r}p_{\omega',s}= (1+(-1)^{p_{\omega,r}p_{\omega',s}+1})/2$. By doing so, we reach the following decomposition:
\begin{align}\label{eq:lcu_T}
\begin{split}
    &T = \sum_{j=1}^\eta \sum_{\bp \in \mcG} \frac{G_p^2}{2}\ket{\bp}_j \bra{\bp}_j =\\
    &\frac{1}{4} \sum_{j=1}^\eta \sum_{1 \le \omega,\omega' \le 3} |\langle \bb_\omega, \bb_{\omega'} \rangle| \sum_{1\le r,s \le n_p-2} 2^{r+s} \sum_{b \in \{0,1\}}\sum_{\bp \in \mcG} \\ &  (-1)^{\sgn(\langle \bb_\omega, \bb_{\omega'} \rangle p_\omega p_{\omega'})}(-1)^{b(p_{\omega,r} p_{\omega',s} + 1)} \ket{\bp}_j \bra{\bp}_j.
\end{split}
\end{align}
 where $\sgn(x) = 0$ if $x\ge 0$ and is equal to $1$ otherwise. Notice the appearance of the inner product in $(-1)^{\sgn(\langle \bb_\omega, \bb_{\omega'} \rangle)}$, which yields one for an orthogonal lattice. The unitaries are in the last line of the equation. The LCU in \cref{eq:lcu_T} is also used for the kinetic term of the all-electron calculations as the lattices in our case studies are no longer orthonormal.
\begin{rmk}\label{rmk:lcu_t_b_getting_dropped}
When the lattice is orthogonal, the sum corresponding to $b=0$ is a multiple of the identity, which can be omitted in the qubitization of $H$. This shifting is used in our case studies with orthogonal lattices, giving an improvement over the LCU proposed in Ref.~\cite{su2021fault}, as it decreases the value of $\lambda$ and thus the simulation cost.
\end{rmk}

\textit{LCU for $V$.} The unitaries used here are signed translations of the lattice by the momentum vector $\bnu$. This strategy applies, given that the matrix entry $V_{(\bp'_i,\bq'_j),(\bp_i,\bq_j)}$ of $V$ is nonzero only when $ \bp'_i-\bp_i = \bnu = \bq_j-\bq'_j \in \mcG_0$, and its value depends only on $\bnu$ (see \cite[App. E.2]{batterypaper}). This term identically appears in the all-electron setting. Furthermore, in that case, as the term $\frac{1}{G_\nu^2}$ is shared by $U$ and $V$ selection probabilities, the LCU for $U$ is similarly derived. See \cref{app:Ineq_test} for more details. The LCU is given by

\begin{align}\label{eq:lcu_V}
&V = \sum_{\bm{\nu}\in \mathcal{G}_0}\frac{\pi}{\Omega G_\nu^2}\sum_{i\neq j=1}^\eta\sum_{b\in\{0,1\}}\sum_{\bp,\bq \in \mcG} \nonumber  \\
&(-1)^{b([\bp+\bnu \notin \mathcal{G}]\vee[\bq-\bnu \notin \mathcal{G}])}\ket{\bp+\bnu}_i\bra{\bp}_i\ket{\bq-\bnu }_j\bra{\bq}_j,
\end{align}
where as before the unitaries are denoted in the last line of the equation.\\

\textit{LCU for $U_{loc}$.} The LCU for the local term is derived similarly to $V$, given its symmetry with respect to any translation of the lattice by $\bnu \in \mcG_0$. However, the selection probabilities are different, as $U_{loc}$ entries include an additional exponentially decaying term in the numerator, denoted by $\gamma_I(G_\nu)$ and defined below. The resulting LCU, with unitaries corresponding to the last line of the equation, is given by

\begin{align}\label{eq:lcu_loc}
    &U_{loc} = \sum_{j=1}^\eta\sum_{\bm{\nu}\in \mathcal{G}_0}\sum_{I=1}^{L}\frac{2\pi |\gamma_I(G_\nu)|}{\Omega G_\nu^2}\sum_{b\in\{0,1\}} \sum_{\bm{q}\in \mathcal{G}} \nonumber \\ &(-1)^{\sgn(\gamma_I(G_\nu))+b[(\bm{q}-\bm{\nu})\notin \mathcal{G}]+1}e^{i\bm{G}_{\nu}\cdot \bm{R}_I}\ket{\bm{q}-\bm{\nu}}\bra{\bm{q}}_j
\end{align}
where
\begin{align}
\begin{split}\label{eq:dfn_gamma_I}
    \gamma_I(G_\nu) &:= \Big(e^{-(G_\nu^2 r_{loc}^2)/2} \big( -Z_\ion+ \nonumber\\
    &(C_1+3C_2)\frac{\sqrt{\pi}r_{loc}^3G_\nu^2}{2} -C_2\frac{\sqrt{\pi}r_{loc}^5G_\nu^4}{2} \big)\Big).
    \end{split}\\
\end{align}
The parameters $r_{loc},C_1,C_2$ were previously introduced in \cref{eq:uloc_me}. They depend on the atomic type $t_I$, so the function $\gamma_I$ is determined by $t_I$. We use this fact to change the notation to $\gamma_t$ when the context is clear. \\

\textit{LCU for $U_{NL}$}.
The main insight in deriving an LCU for the non-local term is to exploit the projector representation of the operator as in \cref{eq:hgh_nl}. We then break down each projection $P$ into a sum of the identity and a reflection operator $\frac{1}{2}\mathbbm{1} - \frac{1}{2}(\mathbbm{1}-2P)$. This leads to a smaller $\lambda$ and more efficient strategies for the qubitization of $U_{NL}$. This decomposition also gives rise to the identity terms that lead to a shifted Hamiltonian. For convenience, we still refer to $U_{NL}$ and $H$ by the same name after the shifting. A full derivation can be found in \cref{app:U_NL_LCU}, and the resulting LCU is given by 
\begin{align}
    \begin{split}\label{eq:lcu_NL}
    &U_{NL} = \sum_{j=1}^\eta \sum_{I,\sigma} |c_{I,\sigma}|\\
    & (-1)^{\sgn(c_{I,\sigma})}R(\bR_I)^\dagger(\mathbbm{1}-2\ket{\Psi_{I,\sigma}}\bra{\Psi_{I,\sigma}})R(\bR_I),
    \end{split}
\end{align}
where
\begin{align}\label{eq:R_R_I_phase_NL}
   R(\bR_I)\ket{\bp} = e^{i\bG_p \cdot \bR_I}\ket{\bp}.
\end{align}
As before, the second line in the equation for the LCU are the unitaries in the decomposition. The coefficients $c_{I,\sigma}$ are complicated expressions representing sums over Gaussian terms, which we define fully in \cref{app:U_NL_LCU}. For example, the simplest one corresponds to
\begin{align}
c_{I,0} := \frac{-8\pi r_0^3 B_0(\sum_{\bp} e^{-G_p^2r_0^2})}{\Omega  }.
\end{align}
These coefficients depend on the atomic specie $t_I$ and the label $\sigma$ that denotes the type of Gaussian superposition defining the states $\ket{\Psi_{I,\sigma}}$. There are eleven different choices for $\sigma\in \{0,\ldots,10\}$. All corresponding Gaussian states are derived in \cref{app:U_NL_LCU}, with $\ket{\Psi_{I,\sigma}}$ for $\sigma =0$ shown in \cref{eq:psi_I_o_example}. Notice we need to use this sign function as the selection probabilities must be positive. 

We make a few comments on the generalizability of the LCUs. Our preference in \cref{eq:dfn_gamma_I} to only include $C_1,C_2$, and not the higher order terms in \cref{eq:uloc_me} is motivated by the fact that $C_3=C_4=0$ for all of our case studies. Nevertheless, the algorithm subroutines that are relevant to $\gamma_I(G_\nu)$ apply without changes to any material for which $C_3,C_4 \neq 0$. In fact, the LCUs above are applicable to any pseudopotential as defined in \cref{eq:u_loc_mast_eq,eq:gnl_7}. Indeed, for the local term, the LCU is identical to \cref{eq:lcu_loc}. For the non-local term, if one includes more than one projector, then the LCU will involve more Gaussian superpositions $\ket{\Psi_{I,\sigma}}$, some of which do not appear in our explicit derivation in \cref{app:U_NL_LCU}. The generalization of our subroutines is straightforward as well, especially given the fact that we employ QROM for the preparation of superpositions, where the functional defining the amplitudes could change according to the specific parameters of the pseudopotential.

\subsection{Breakdown of the qubitization subroutines}
The qubitization of $H=T+U_{loc}+U_{NL}+V$ is roughly broken down into a qubitization of each term. This means defining PREP and SEL subroutines for each of the four operators, for instance SEL$_T$ for the kinetic energy term. We use this notation from now on. Henceforth, `AE' refers to the all-electron setting, with Hamiltonian $H=T+U+V$ given by \cref{eq:T,eq:U,eq:V} and no assumption on the lattice. In our discussion of PREP and SEL, we mention the needed adjustments to the AE algorithm when the lattice is not orthonormal. The all-electron setting with an orthonormal lattice, hereafter referred to by `OAE', has been qubitized in \cite{su2021fault,batterypaper}. 
 
\subsubsection{Prepare operator (PREP)}\label{sssec:the_prep_subroutine_of_qubitization}
Below, the target PREP state for the pseudopotential Hamiltonian is shown, as implied by their LCUs in \cref{eq:lcu_T,eq:lcu_V,eq:lcu_loc,eq:lcu_NL}. The error analysis is done in \cref{sec:error_analysis_and_eff_value_lambda}. We refer to \cite[Eq. (48)]{su2021fault} for the similar equation in the OAE setting.

\begin{widetext}
\begin{align}\label{eq:master_prep_state}
    &\PREP_H \ket{\bm{0}} = \\\label{eq:prep_mcX}
    &\Big(\sum_{\rchi \in \{0,1\}^2} \sqrt{\frac{\lambda_{\rchi}}{\lambda}}\ket{\rchi}_{\mcX} \Big) \\\label{eq:prep_b_omega_omega_prime}
    &\otimes\Big(\sum_{\omega,\omega' \in \{1,2,3\}} \frac{\sqrt{|\langle \bb_\omega, \bb_{\omega'} \rangle|}}{\sqrt{\sum_{\omega,\omega' } |\langle \bb_\omega, \bb_{\omega'} \rangle|}} \ket{\omega,\omega',\sgn(\langle \bb_\omega, \bb_{\omega'} \rangle)}_f\Big) \\\label{eq:prep_r_s_binary_T}
    &\otimes  \frac{1}{2^{n_p-1}-1}\Big(\sum_{r,s=0}^{n_p-2} 2^{(r+s)/2} \ket{r}_g \ket{s}_h \Big) \\\label{eq:prep_b_c_d_e_eta}
    &\otimes\ket{+}_b \otimes \frac{1}{\sqrt{\eta}}\Big( \sqrt{\eta-1} \ket{0}_c\sum_{i\neq j=1}^{\eta}\ket{i}_d\ket{j}_e+\ket{1}_c\sum_{j=1}^{\eta}\ket{j}_d\ket{j}_e \Big) \\\label{eq:prep_v_state}
    &\otimes\Big( \sqrt{\frac{P_{\nu,V}}{\lambda_{\nu,V}}} \ket{0}_{j_V}\sum_{\bm{\nu}\in \mathcal{G}_0}\frac{1}{G_\nu}\ket{\bm{\nu}}_{k_V}+  \sqrt{1-P_{\nu,V}}\ket{1}_{j_V}\ket{\bm{\nu}^\perp}_{k_V} \Big)\\\label{eq:prep_loc_state}
    &\otimes\Big( \sum_{I,\bm{\nu}\in \mathcal{G}_0}\frac{|\gamma_I(G_\nu)|^{1/2}}{\lambda_{\nu,loc}^{1/2}G_\nu}\ket{\bm{\nu}}_{k_{loc}}\ket{I}_{k'_{loc}}\ket{\sgn(\gamma_I(G_\nu))}_{s_{loc}} \Big)  \\\label{eq:prep_nl_state}
    &\otimes  \Big(\sum_{I,\sigma} \frac{\sqrt{|c_{I,\sigma}|}}{\sqrt{\sum_{I,\sigma} |c_{I,\sigma}|}}\ket{I}_{k'_{NL}}\ket{\sigma}_{k_{NL}}\ket{\sgn(c_{I,\sigma})}_{s_{NL}} \Big).
\end{align}
\end{widetext}

Each register is denoted by a subscript, such as $\mcX$ which labels the first register. While we borrow techniques from \cite{su2021fault} to prepare some registers, there are also adjustments and new registers:
\begin{itemize}
    \item We have to qubitize four operators instead of three.
    \item As the lattice is no longer orthonormal, the selection probabilities are more general, and in cases like $T$, they include the inner product of the reciprocal lattice vectors.
    \item When using pseudopotentials, the selection probabilities of $U_{loc}$ and $V$ are not proportional to each other, hence their corresponding momentum state superposition cannot be shared. This is in contrast to the AE case, where the momentum state superposition of $U$ and $V$ is shared.
    \item The preparation of the state in \eqref{eq:prep_loc_state} for the local term is via QROM, instead of the usual inequality test involving quantum arithmetic \cite[Sec. II.C.]{su2021fault}. The same holds for the register storing the amplitudes $\sqrt{|c_{I,\sigma}|}$ for $U_{NL}$.
\end{itemize}
We give a sketch of the preparation of the state of each register, with more details provided in \cref{app:the_details_of_prep_and_sel}.
\begin{enumerate}
\item The state of register $\mcX$ is made of two qubits, and is a superposition prepared by QROM with amplitudes $\sqrt{\lambda_{\rchi}/\lambda}$, where $\lambda_{\rchi}$ is the sum of the selection probabilities of the LCU for $\rchi$. We enumerate $T,V$ with $\rchi = 00,01$ and $U_{loc},U_{NL}$ with $\rchi = 10,11$. This state enables us to end up with the desired $\frac{H}{\lambda} = (T+U_{loc}+U_{NL}+V)/\lambda$ instead of $\frac{T}{\lambda_T} + \frac{U_{loc}}{\lambda_{loc}} + \frac{U_{NL}}{\lambda_{NL}}  + \frac{V}{\lambda_V}$.

\item The state $f$ is made of five qubits, two for each of $\omega, \omega'$, and one storing $\sgn(\langle \bb_\omega, \bb_{\omega'} \rangle)$. This is prepared using QROM and is a part of the PREP state of $T$. There is a failure probability in the preparation of this state, where ineligible states are flagged by an additional qubit. This is not shown above to avoid cluttering.

\item The superposition in registers $g,h$ is the same one from OAE. It is prepared by implementing a unary state using controlled-Hadamards and bit flips~\cite[Fig. II.]{batterypaper}. These two registers along with $f,b,e$ form the PREP state of $T$.

\item The state of register $b$ is given by the $\ket{+}$ eigenstate of Pauli $X$. It is used in the PREP state of $T,U_{loc},V$, where the unitaries in the LCUs are signed permutations. Note that $b$ is not part of the PREP state of $T$ when the lattice is orthogonal (\cref{rmk:lcu_t_b_getting_dropped}).

\item The states of registers $c,d,e$ also appear in the OAE setting, and form the superposition over the electrons, which is needed in first quantization as part of the PREP state of all four operators. The technique to prepare the uniform superpositions is available in~\cite[App. A.2]{lee2021even}, and is followed by checking whether $i = j$, with the result stored in the additional qubit of register $c$.

\item The superposition over the momentum $\bnu$ corresponding to $V$ is given in registers $j_V,k_V$. This is prepared using an inequality test (\cref{app:Ineq_test}) followed by amplitude amplification to increase the probability of success, flagged by $\ket{0}_{j_{V}}$. The procedure is very similar to the OAE setting \cite[Sec. II.C.]{su2021fault}, except that we use QROM instead of quantum arithmetic to compute one side of the inequality test.

\item The superposition with registers $k_{loc},k'_{loc},s_{loc}$ correspond to the momentum state of $U_{loc}$. The register $k'_{loc}$ enumerates all nuclei $I=(t,t_{j})$. While for $V$ the similar superposition is created using an inequality test, according to our simulations, the exponentially small term $\gamma_I(G_\nu)$ makes the rejection of the inequality test happen too often, leading to excessive rounds of amplitude amplifications that increase the Toffoli cost. Therefore QROM is used to prepare almost the entire superposition. Note that $\gamma_I$ only depends on the atomic species, so QROM prepares the superposition over register $k_{loc},s_{loc}$ and the atomic species index $t$ in the register $k'_{loc}$. What remains to be done is a uniform superposition over all nuclei $t_{j}$ of each specie $t$, which is created using the technique in \cite[App. A.2]{lee2021even}.

\item The registers $k_{NL},k'_{NL},s_{NL}$ give a superposition for the PREP state of $U_{NL}$. The coordinate $I$ in register $k'_{NL}$ enumerates all nuclei, with the same splitting mentioned above for $k'_{loc}$. To prepare \cref{eq:prep_nl_state}, recall that $c_{I,\sigma}$ only depends on the atomic type $t_I$ and the Gaussian superposition type $\sigma$. Therefore, a QROM produces the superposition over $\ket{t}_{k'_{NL}}\ket{\sigma}_{k_{NL}}\ket{\sgn(c_{t,\sigma})}$. Then a uniform superposition over the $t_j$ nuclei of type $t$ is implemented, giving the desired \cref{eq:prep_nl_state}. 
\end{enumerate}

\subsubsection{Select operator (SEL)}\label{ssec:the_sel_subroutine_of_qubitization}

We mostly borrow the corresponding implementation in the OAE case \cite{su2021fault} for every SEL operator except $\SEL_{NL}$, while adjusting for general lattices and the different registers holding the momentum state for $U_{loc}$ and $V$. We devote more explanation to $\SEL_{NL}$ as it is the operator with no similar precedent in the literature. Nevertheless, this section is not a detailed compilation, especially for $\SEL_{NL}$ which is the most involved; we refer to \cref{app:the_details_of_prep_and_sel} for more details.

There are some commonalities among all SEL operators, which we briefly discuss. The action of each operator is controlled on a register that flags the success of the corresponding state preparation. 
For example for $T$, we need to check three conditions:
\begin{itemize}
        \item The state $\ket{\rchi}_{\mcX}$ is equal to $\ket{00}_{\mcX}$,
        \item The register $c$ flags the success of $i=j$ in registers $d,e$ (i.e., $\ket{1}_c$),
        \item The ancilla attached to register $f$ flags the meaningful basis states ($\omega \neq 4$ and $\omega' \neq 4$) in the superposition prepared by QROM.
    \end{itemize}
If all the above conditions hold, then we get $\ket{0}_T$, and get $\ket{1}_T$ otherwise, in which case SEL$_T$ acts as the identity. Checking correctness of the PREP states can be performed using a few logical gates (Toffoli, CNOT, X). Note that these checks are performed as part of the $\PREP$ procedure, but we find it more informative to introduce them here.

The second design shared by all select operators is a common CSWAP circuit and its inverse. This circuit first copies $\ket{\bp}_i$ (controlled on $\ket{i}_d$) into an auxiliary register, swaps it back into its place after the relevant SEL operations are carried out, then does the same for $\ket{\bp}_j$ controlled on $\ket{j}_e$. This ensures that the SEL operations are all done on a single auxiliary register, obviating the need for the far more costly controlled operations directly on the system register \cite[Eq. (E27)]{batterypaper}. In the AE case, this is actually the costliest part of implementing the select operator. We now discuss how to implement each select operator, the sum of which is the desired $\SEL_H$.\\

1. $\SEL_T$: The transformation is given by
    \begin{align}\begin{split}
        &\ket{0}_{T} \ket{+}_b\ket{j}_e \ket{\omega,\omega',\sgn(\langle \bb_\omega,\bb_{\omega'}\rangle)}_f\\
        &\ket{r}_g \ket{s}_h \ket{\bp}_j \to\\
        &(-1)^{b(p_{\omega,r}p_{\omega',s}+1)+\sgn(\langle \bb_\omega,\bb_{\omega'}\rangle p_\omega p_{\omega'})} \ket{0}_{T} \ket{+}_b \ket{j}_e \\
        &\otimes\ket{\omega,\omega',\sgn(\langle \bb_\omega,\bb_{\omega'}\rangle)}_f \ket{r}_g \ket{s}_h \ket{\bp}_j.
    \end{split}
    \end{align}
    The action above is essentially a phase, which is a combination of multi-controlled Z gates. The sign of the inner products appearing in the phase is the term making the distinction from the OAE case (\cite[Eq. (49)]{su2021fault}). \\

2. $\SEL_V$: The implementation follows the OAE setting \cite[Eq. (51)]{su2021fault}:
    \begin{align}\begin{split}
        &\ket{0}_{V} \ket{b}_b \ket{i}_d \ket{j}_e \ket{0}_c \ket{0}_{j_V} \ket{\bnu}_{k_V} \ket{\bp}_i \ket{\bq}_j \to \\
        &(-1)^{b([(\bp+\bnu) \not \in \mcG] \lor [(\bq-\bnu) \not\in \mcG])} \ket{0}_{V} \ket{b}_b \ket{i}_d \ket{j}_e  \\
        &\otimes \ket{0}_c \ket{0}_{j_V}\ket{\bnu}_{k_V} \ket{\bp+\bnu}_i \ket{\bq-\bnu}_j.
        \end{split}
    \end{align}
    The addition and subtraction along with the phase implementation are all controlled on $\ket{0}_V$. The non-diagonal superposition $\ket{i}_d\ket{j}_e$ over the electrons, flagged by $\ket{0}_c$, is acted upon by $\SEL_V$, while the rest of the SEL operators only act on the $e$ register. \\

3. $\SEL_{loc}$: The operator acts in two main steps:
    \begin{align}\begin{split}
        &\ket{0}_{loc} \ket{b}_b \ket{j}_e \ket{\bnu}_{k_{loc}} \ket{I}_{k'_{loc}}  \\
        &\otimes \ket{\sgn(\gamma_I(G_\nu))}_{s_{loc}} \ket{\bm{0}}_{\bR}   \ket{\bp}_j \to \\
        &(-1)^{b[(\bp-\bnu) \not \in \mcG]} \ket{0}_{loc} \ket{b}_b \ket{j}_e \ket{\bnu}_{k_{loc}} \ket{I}_{k'_{loc}} \\
        &\otimes\ket{\sgn(\gamma_I(G_\nu))}_{s_{loc}} \ket{\bR_I}_{\bR} \ket{\bp-\bnu}_j.
    \end{split}
    \end{align}
    This first step illustrates a controlled sign, along with a controlled subtraction $\ket{\bp}_j\rightarrow \ket{\bp-\bnu}_j$. It also shows a QROM that reads $\ket{x}_{loc}\ket{I}_{k'_{loc}}$ and outputs $\bR_I$ into the register $\ket{\cdot}_{\bR}$ only if $x=0$. This is then used to implement the phase $e^{i \bG_\nu \cdot \bR_I}$, following the techniques in \cite[Eq. (50)]{su2021fault}. The second step maps the previous state to
    \begin{align}
    \begin{split}
    &\to e^{i \bG_\nu \cdot \bR_I} (-1)^{b[(\bp-\bnu) \not \in \mcG]+\sgn(\gamma_I(G_\nu))} \ket{0}_{loc} \ket{b}_b \ket{j}_e\\
    &\otimes  \ket{\bnu}_{k_{loc}} \ket{I}_{k'_{loc}} \ket{\sgn(\gamma_I(G_\nu))}_{s_{loc}}  \ket{0}_{\bR} \ket{\bp-\bnu}_j.
    \end{split}
    \end{align}
    Here we erase the register $\bR$ by taking the inverse of the QROM, and finish by applying $\text{Z}_{s_{loc}}$ controlled on $\ket{0}_{loc}$. As a remark, the AE circuit for $\SEL_U$ is exactly the same as its OAE implementation in \cite{su2021fault}.\\
    
4. $\SEL_{NL}$: The unitary from the LCU in \cref{eq:lcu_NL} is
\begin{align}
    (-1)^{\sgn(c_{I,\sigma})}R(\bR_I)^\dagger(\mathbbm{1}-2\ket{\Psi_{I,\sigma}}\bra{\Psi_{I,\sigma}})R(\bR_I),
\end{align}
which can be broken down as a series of transformations, which we describe below. To avoid cluttering, we present only the registers involved in each stage. The first transformation applies a phase as follows:
\begin{align}
\begin{split}
     &\ket{0}_{NL} \ket{\bq}_j \ket{I}_{k'_{NL}}  \ket{\bm{0}}_{\bR} \to \\
     &\ket{0}_{NL}\ket{\bq}_j  \ket{I}_{k'_{NL}}  \ket{\bR_I}_{\bR} \to \\
     &e^{i \bG_q \cdot \bR_I} \ket{0}_{NL}\ket{\bq}_j \ket{I}_{k'_{NL}}   \ket{\bR_I}_{\bR}.
     \end{split}
\end{align}
To implement this transformation, first a QROM reads $\ket{x}_{NL}\ket{I}_{k'_{NL}}$ and outputs $\bR_I$ into $\ket{\cdot}_{\bR}$ only if $x=0$. Then the phase $e^{i \bG_q \cdot \bR_I}$ is applied, similar to how $e^{i \bG_\nu \cdot \bR_I}$ was applied for $U_{loc}$. The next stage is to apply the reflection onto a Gaussian state $\ket{\Psi_{I,\sigma}}$. To do so, we need to prepare the state and apply a reflection:
\begin{align}\label{eq:reflection_SEL_NL_text}
    U_{I,\sigma}(\mathbbm{1}-2\ket{0}_{NL}\ket{\bm{0}}\bra{\bm{0}}\bra{0}_{NL})U_{I,\sigma}^\dagger,
\end{align}
where $U_{I,\sigma}\ket{\bm{0}} = \ket{\Psi_{I,\sigma}}$ acts on the same register as $\ket{\bq}_j$. Note that the reflection also includes the flag qubit $\ket{0}_{NL}$, ensuring that $\SEL_{NL}$ acts by identity if the basis state has not been successfully prepared for the non-local term. 

This reflection is the most expensive part of $\SEL_H$, so it is worthwhile to discuss strategies to reduce its cost.  Many materials have either orthogonal or partially orthogonal lattices, i.e., when a lattice vector is orthogonal to the other two. This crystallographic feature is prevalent among many materials of interest, including the ones utilized as cathode materials. For instance, about half of the crystal structures available in the Materials Project database ~\cite{jain2013commentary} have orthogonal or partially orthogonal lattices. Assuming this, the Gaussian state can always be decomposed into the tensor product of three one-dimensional (1D) Gaussian states, or a 1D+2D Gaussian state, respectively. Since QROM cost rises exponentially with the number of read qubits, it is important to exploit this decomposition. As a result, for orthogonal (and partially orthogonal lattices), one has three (two) QROMs acting in parallel and all reading $n_p$ ($n_p$ and $2n_p$) qubits, instead of one QROM reading $3n_p$ qubits. As an example of the decomposition, we have the following for \cref{eq:psi_I_10_example} when the lattice is orthogonal:
\begin{align}
    &\ket{\Psi_{I,10}} = \sum_{\bp_1}  e^{-G_{p,1}^2r_2^2/2}\ket{\bp_1} \otimes  \nonumber \\
    &\sum_{\bp_2}  G_{p,2}e^{-G_{p,2}^2r_2^2/2}\ket{\bp_2}\otimes \sum_{\bp_3}  G_{p,3} e^{-G_{p,3}^2r_2^2/2}\ket{\bp_3}.
\end{align}
where $G_{p,\omega}$ is the $\omega$ coordinate of $\bG_p$. Once the reflection in \cref{eq:reflection_SEL_NL_text} is implemented on $\ket{\bq}_j \ket{0}_j$, it leads to a superposition of the form $\sum_{\bp \in \mcG} \phi_{I,\sigma}(\bp) \ket{\bp}_j$ for some amplitudes $\phi_{I,\sigma}(\bp)$. 

The final steps are (i) the application of the phase $e^{-i\bG_p \cdot \bR_I}$, which is done similar to its inverse at the beginning, (ii) the erasure of $\bR_I$ from the register $\bR$ by the inverse of the QROM that created it, and (iii) the Z gate $\text{Z}_{s_{NL}}$ on $s_{NL}$ controlled on $\ket{0}_{NL}$. Focusing on the notable registers, the final result is:
\begin{align}
    & (-1)^{\sgn(c_{I,\sigma})}\ket{0}_{NL} \ket{\sigma}_{k_{NL}}\ket{I}_{k'_{NL}}\ket{\sgn(c_{I,\sigma})}_{s_{NL}}\otimes    \nonumber \\
    &\Big( \sum_{\bp \in \mcG} e^{i (\bG_q-\bG_p) \cdot \bR_I} \phi_{I,\sigma}(\bp) \ket{\bp}_j  \Big).
\end{align}

Finally, we explain our choice of QROM for $U_{I,\sigma}$ in further detail. First, while the preparation of a discrete Gaussian state such as $\ket{\Psi_{I,0}} = \sum_{\bp}  e^{-G_p^2r_0^2/2}\ket{\bp}$ has been specifically treated in earlier~\cite{kitaev2008wavefunction}, we have to also prepare higher-order derivatives of such states, for example: 
\begin{align}\label{eq:psi_I_10_example}
    \ket{\Psi_{I,10}} = \sum_{\bp}  (G_{p,2}G_{p,3})e^{-G_p^2r_2^2/2}\ket{\bp}, 
\end{align}
Even for a diagonal covariance matrix corresponding to an orthogonal lattice, the preparation method in~\cite{kitaev2008wavefunction} is inefficient compared to QROM as it assumes an arithmetic oracle.

Furthermore, the specific case of a discrete Gaussian state with a non-diagonal covariance matrix has not been properly investigated and optimized. While the work in~\cite{kitaev2008wavefunction} provides some ideas like a simple basis change, the details regarding non-orthogonal lattices are far more complicated and our estimates show that the algorithm does not yield the actual state in a cost-efficient way. In our range of applications, more recent methods like inequality test coupled with quantum arithmetics~\cite{su2021fault} fail at providing a good balance of the product of number of qubits and number of gates. A more sophisticated preparation method called state preparation without coherent arithmetic~\cite{mcardle2022quantum} is promising. However, it is also more complicated for cost and error analysis and crucially provides less parallelization and qubit/gate trade-off opportunities, which we frequently exploit to reduce overall cost of the algorithm.

\section{Error analysis and the effective value of $\lambda$}\label{sec:error_analysis_and_eff_value_lambda}

The quantum phase estimation algorithm targets a maximum total error that we denote by $\error$. To achieve this, we need to identify all individual sources of error in the algorithm. We use $\error_X$ to denote each source of error, where $X$ will be replaced by a label describing the type of error. These errors are ultimately related to finite precision operations respectively using $n_X$ bits. The choice of $n_X$ further determines the number of qubits and non-Clifford gates used in the algorithm, and is involved in identifying the \textit{effective} value of the normalization factor $\lambda$ after qubitization. Below, we review the different sources of error, and discuss how to compute $\lambda$. We relegate the detailed derivations to \cref{app:lambdas,app:errors}.

\subsection{Overview of errors and finite size approximations}\label{sec:overview_of_errors}

Many of our errors are related to the precision of the rotation angles $\theta$ in \qromalgo~when building a superposition using QROM. Below, we show the complete list of all qubitization errors:
\begin{enumerate}
    \item $\error_{\rchi}$ is the error due to using $n_{\rchi}$ bits for the precision of the rotation angles necessary to build the superposition of register $\mcX$ (\cref{eq:prep_mcX}) using QROM.
    \item $\error_B$ is a similar error, due to using $n_{B}$ bits for preparing the superposition of register $f$ in \cref{eq:prep_b_omega_omega_prime}.
    \item $\error_{NL}$ is the error due to using $n_{NL}$ bits in the QROM for building the PREP state for $U_{NL}$ in \cref{eq:prep_nl_state}.
    \item $\error_{M_V}$ is the error due to using $n_{M_V}$ bits in the QROM computing the inequality test for preparing the PREP$_V$ state in \cref{eq:prep_v_state}.
    \item $\error_{M_{loc}}$ is the error due to using $n_{M_{loc}}$ bits for preparing the local PREP state \cref{eq:prep_loc_state} with QROM.
    \item $\error_{\Psi}$ is the error due to using $n_{\Psi}$ bits for building the superpositions $\ket{\Psi_{I,\sigma}}$ using QROM.
    \item $\error_{R} \le \error_{R,loc}+\error_{R,NL}$ is the error due to the finite size register $n_R$ used to represent $\bR_I$ in register $\bR$ for the implementation of $\SEL_{loc}$ and $\SEL_{NL}$.
\end{enumerate}
Finally, let $\error_{\text{QPE}}$ be the error from quantum phase estimation. To achieve an approximation $\error$ of the ground state energy, it is necessary and sufficient that
\begin{align}\label{eq:error_inequality}
    \error^2 \ge& ~ \error_{\text{QPE}}^2 + \nonumber \\&(\error_{\rchi}+\error_B+\error_{NL}+\error_R+\error_{M_V}+\error_{M_{loc}} + \error_{\Psi})^2.
\end{align}
The proof follows a similar argument to the OAE setting \cite[Thm. 4]{su2021fault}.

Each $\error_X$ is upper bounded by an expression depending on $n_X$, which thus determines the value of $n_X$ needed to obtain an error $\error_X$. We estimate all errors in \cref{app:errors}, with results summarized in \cref{tab:error_profile}. As an example, following the Lemma in \cref{ssec:background_qrom}, it can be shown that
\begin{align}\label{eq:error_rchi_main_txt_example}
    \error_{\rchi} \le \frac{4\pi }{2^{n_{\rchi}}}\lambda, 
\end{align}
which implies that
\begin{align}\label{eq:n_rchi_main_txt_example}
    n_{\rchi} = \left\lceil\log\left( \frac{4\pi\lambda}{\error_{\rchi}}\right)\right\rceil,
\end{align}
gives an error less than or equal to $\error_{\rchi}$.

\subsection{The effective value of $\lambda$}\label{ssec:effective_lambda}

The value of $\lambda$ determines the scaling in the block-encoding of $H$ and is a significant factor in the cost that is highly dependent on the chosen LCU and compilation strategy. Therefore, its accurate computation is important. The algorithm targets a PREP state that is different from the theoretical one implied by the LCUs, and the resulting effective $\lambda$ is closely related but technically different from the one implied by the LCUs. Note that this is also a feature of previous work (\cite[Thm. 4]{su2021fault}).

Computing the effective $\lambda$ requires finding the effective one for each of the four operators $T, U_{loc}, U_{NL}, V$. To do so, we identify any failure/success probability embedded in the PREP implementation. Here, failure refers to a basis state itself being inadmissible. This appears when preparing a uniform superposition over a basis that is not a power of two, or when an inequality test is rejected. Even upon success, the obtained amplitudes are usually different from the desired ones. For example, we use an inequality test to prepare amplitudes $G_\nu^{-1}$ for \cref{eq:prep_v_state}, and after success, we get some complicated expression approximating $G_\nu^{-1}$ (see \cref{eq:error_M_V_derivation}).

Overall, the $\lambda$ for each operator will roughly look like a sum of effective selection probabilities $\sum \tilde{\alpha_\ell}$ divided by a product of success probabilities $\prod_i P_i$. Then the effective total value $\lambda = \lambda_{T}+\lambda_{loc}+\lambda_{NL}+\lambda_{V}$ can be derived. We do this in \cref{app:lambdas} and summarize the values in \cref{tab:lambda_vals}. As an example, we show $\lambda_T$ below. This is computed after the compilation of the LCU of $T$ in \cref{eq:lcu_T}. In the numerator of the expression below, we have the explicit value for $\sum \tilde{\alpha_\ell}$. In the denominator, we have a success probability related to the preparation of the uniform superposition over pairs of electrons in \cref{eq:prep_b_c_d_e_eta}:
\begin{align}
\lambda_T = \frac{\eta  2^{2n_p-3}\sum_{\omega,\omega' \in \{1,2,3\}} |\langle \bb_\omega, \bb_{\omega'} \rangle|  }{\Ps(\eta,b_r)^2}.
\end{align}

\section{Gate and qubit cost}\label{sec:cost}
Having described the main steps of the algorithm and identified the sources of error, we now quantify the number of gates and qubits needed to run the full procedure. We follow standard practice in the literature and focus on Toffoli gates since they constitute the bulk of non-Clifford gates used in the algorithm.

The compilation of the algorithm involves numerous strategies, each requiring its own separate cost estimate. This leads to an extensive analysis that cannot be summarized in simple expressions. We have thus gathered all Toffoli costing calculations in the Appendix \cref{tab:gatecosts} and demonstrate the derivation of each in \cref{app:gate_costings}. These formulas are implemented in code at \url{https://github.com/XanaduAI/pseudopotentials} and used to perform resource estimation calculations. Here we focus on highlighting the most expensive steps of the algorithm.

First, we recall the rough estimate for the cost of a qubitization-based QPE algorithm:
\begin{align}
\left\lceil \frac{\pi \lambda}{\error_{\text{QPE}}}\right\rceil (2\PREP_{cost} + \SEL_{cost}),
\end{align}
where we explicitly include the qubitization costs of $\PREP$ and $\SEL$. The factor of two appears because in the qubitization operator we apply PREP and its complex conjugate. The multiplicative factor $\lambda/\error$ is the largest contributor to the total cost. For example, for a system with $N=10^5$ plane waves and a hundred electrons (\cref{table:mat_dis}), this fraction gives a factor of about $10^9$, while the qubitization cost contributes roughly a factor of $10^5$. 

The two most costly subroutines in the qubitization part of the algorithm are:
\begin{itemize}
    \item Preparing $\sum_{I,\bnu} \frac{|\gamma_I(G_\nu)|^{1/2}}{G_\nu} \ket{\bnu,I,\sgn(\gamma_I(G_\nu))}$, which is part of the PREP of the local term. 
    \item Performing the reflection on the Gaussian superpositions $\ket{\Psi_{I,\sigma}}$, which is part of SEL of the non-local term.
\end{itemize}
The expressions for the Toffoli cost of these steps are roughly given by
\begin{align}
    &2\big(2\left\lceil \frac{2^{3n_p+\tau+1}}{\beta_{loc}}\right\rceil+ 3\beta_{loc}n_{M_{loc}}(3n_p)\big),\\
    \begin{split}
    &6\big(8\left\lceil \frac{2^{n_{\text{QROM}}+1}-2^{n_{\text{QROM}}-n_p}}{\beta_{\Psi}}\right\rceil + 3\beta_{\Psi} n_{\Psi} n_p\big)  \\
    & +30\big(2\left\lceil \frac{2^{n_{\text{QROM}}+1}-2^{n_{\text{QROM}}-n_p}}{\beta_{\Psi}'}\right\rceil + 3\beta_{\Psi}' n_{\Psi} n_p\big),
    \end{split}
\end{align}
respectively, where $n_{\text{QROM}} = n_p + \lceil \log(\mcT)\rceil + 2$, where $\mcT$ is the number of atomic species in the cell. Given how the parameters $\beta_{loc},\beta_{\Psi},\beta_{\Psi}'$ divide the leading terms, their values are quite important in lowering the cost. However, these scalars have to satisfy constraints based, for example, on the number of available dirty qubits. The space-time trade-offs of QROM during the resource estimation is explained further in \cref{app:res_est_det_results}. 

In our calculations, each of these constitutes the largest share of $\PREP_{cost}$ and $\SEL_{cost}$, respectively. While $\SEL_{cost}$ is slightly cost-dominant for a smaller number of plane waves, the trend is reversed as $N$ grows beyond $10^5$. Consequently, one significant obstacle in further optimizing the qubitization algorithm is the preparation of \cref{eq:prep_loc_state}, where $\gamma_I(G_\nu)$ is the arithmetically involved term in \cref{eq:dfn_gamma_I}. 

The derivation of the qubit cost of the algorithm is of similar complexity to the gate cost. We provide a full analysis in \cref{app:qubit_cost}, where we reuse any uncomputed and clean qubits whenever possible. We also make the distinction between clean and dirty qubits as the latter is implicated in QROMs. Although our counting is different from \cite{su2021fault,batterypaper}, our results in \cref{sec:resource_estimation} follow the same behaviour: the overwhelming contribution to the qubit cost is the size of the system register, which requires $3\eta n_p$ qubits. As an example, the total clean qubit cost for simulating a system of 408 electrons with $N=10^5$ plane waves is about 9,892, while the system register uses $3\eta n_p = 3\cdot 408 \cdot 6= 7,344$ clean qubits (\cref{table:mat_limno}).

\section{Application: lithium-excess cathode materials}
\label{sec:applications}
In this section, we focus on the quantum simulation of lithium-excess (Li-excess) cathode materials, which have been recently proposed for higher-capacity cathodes~\cite{zhang2022pushing, campeon2021fundamentals, jiao2022achieving}. By replacing a fraction of the transition metals with Li atoms, Li-excess materials can potentially offer up to twice the specific capacity, i.e., the total amount of charge stored per unit mass, of conventional cathodes~\cite{yabuuchi2019material}. For example, $\text{Li}_2\text{Mn}\text{O}_3$ has a theoretical capacity of 460 $\text{mAh}~\text{g}^{-1}$ responsible for the voltage of $\sim 4$ V~\cite{yabuuchi2011detailed}. This yields a specific energy of $(460~\text{mAh}~\text{g}^{-1} \times 4~\text{V}) = 1840~\text{Wh}~\text{Kg}^{-1}$ which is well above the specific energy ($800~\text{Wh}~\text{Kg}^{-1}$ at the cathode material level) required to enable full driving performance and significantly reduce the cost of electric vehicles~\cite{li2020high}.  However experiments reveal average voltages of $\sim 3.8$ V and a much smaller capacity of about $180~\text{mAh}~\text{g}^{-1}$ following the first charging cycle~\cite{wang2013atomic}. This voltage decay and the abrupt capacity loss, which is common to other Li-excess materials, is an important obstacle in designing higher-capacity batteries. 

The voltage profile of a cathode material, i.e., the voltage $V(x)$ measured as a function of the concentration of the Li ions $x$, provides relevant information about the Li insertion process~\cite{van2020rechargeable}. The typical voltage profile of Li-excess materials is sketched in~\cref{fig:v-profile}. A distinguishing feature of this profile is a long plateau region following the initial sloping curve during charge. This is indicative of the cathode material transitioning from a solid solution phase to a region where two phases of the material coexist~\cite{van2020rechargeable}. Moreover, the hysteresis loop reveals that the removal of Li ions is accompanied by irreversible structural transformations causing the capacity loss observed experimentally~\cite{zhang2022pushing}.
\begin{figure}[!ht]
\centering
\includegraphics[width=1 \columnwidth]{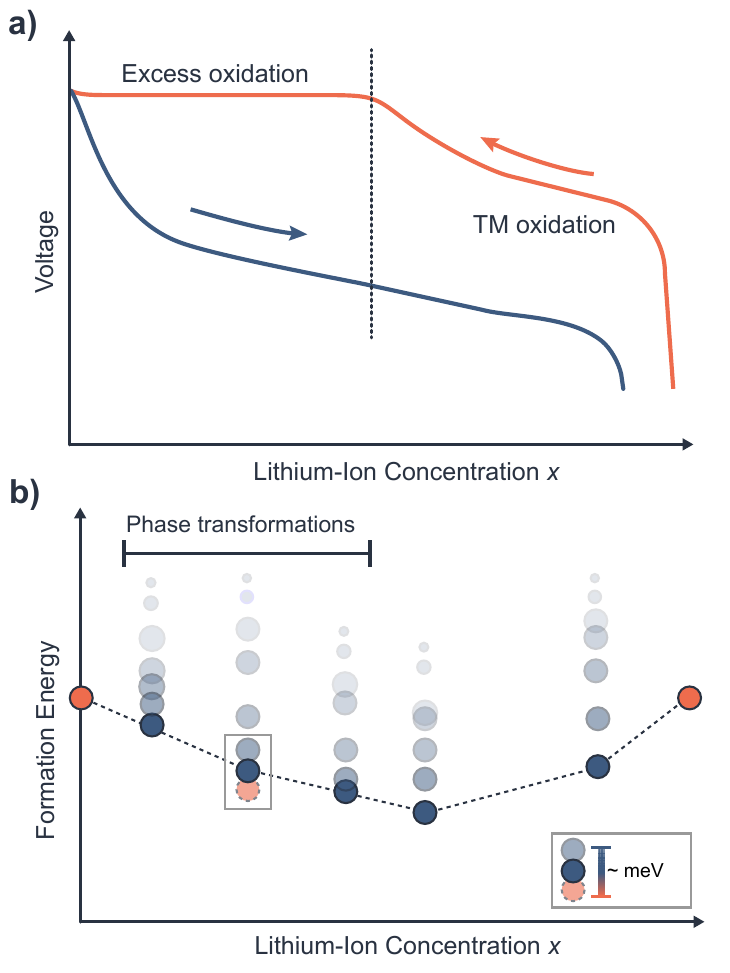}
\caption{a) Sketch of a typical voltage profile for Li-excess materials. The curves depict the charging (orange) and discharging (dark-blue) processes. In the transition metal oxidation region, the extraction of Li ions is charge-compensated by electrons provided by the transition metals. Further Li extraction is thought to be possible via the oxidation of oxygen atoms (excess oxidation). Phase transformations occur within the plateau region. b) A typical formation energy plot with respect to the Li-ion concentration $x$ in a Li-excess cathode. The brighter (dimmed) blue circles correspond to the more (less) stable phases in this representation. The dotted line indicates the convex hull consisting of the most stable phases. The orange circles correspond to the fully lithiated and delithiated phases where {\it ab initio} calculations are strongly supported by experimental data. Differences between stable, unstable, and inaccurate calculations (orange dimmed circle) are in the energy scale of millielectronvolts (meV), as depicted in the inset.}
\label{fig:v-profile}
\end{figure}

The capacity gain of Li-excess materials is most frequently attributed to the oxidation of oxygen anions ($\text{O}^{2-} \rightarrow \text{O}^{n-}, n<2$) (anion redox)~\cite{seo2016structural}. This redox mechanism and alternative mechanisms have been proposed and discussed extensively in Ref.~\cite{zhang2022pushing} along with their relation to structural transforms resulting in materials degradation and performance loss. While important questions remain open, there seems to be a consensus that the increase of oxygen hole states and lithium vacancies destabilizes the metal-oxide chemical bonds and leads to the formation of oxygen dimers~\cite{chen2016lithium, mccoll2022transition}. Furthermore, oxygen dimerization cooperatively favors the migration of transition metals to Li-vacancy sites in the structure, which is the main process driving the structural transformations in Li-excess materials~\cite{mohanty2014unraveling, kleiner2018origin}. 

It follows that identifying the most stable phases of the cathode material for compositions with excess Li is crucial for determining the dominant redox mechanisms and, more importantly, for deriving potential strategies (e.g., doping, modifying the crystal structure) to retain more reversible capacity. The relevant quantity for ranking the stability of possible phases of the delithiated material is the formation energy computed for a given Li ion concentration~\cite{urban2016computational}. For example, in the case of $\text{Li}_2\text{Mn}\text{O}_3$, this is given by~\cite{lim2015origins}
\begin{align}
E_f(x) = E(\text{Li}_x\text{Mn}\text{O}_3) &- \frac{x}{2}E(\text{Li}_2\text{Mn}\text{O}_3) \nonumber \\
& - \left(1-\frac{x}{2}\right)E(\text{Li}_0\text{Mn}\text{O}_3),
\label{eq:form-energy}
\end{align}
where $E(\text{Li}_x\text{Mn}\text{O}_3)$ is the ground-state energy of the material in a given phase, and $E(\text{Li}_2\text{Mn}\text{O}_3)$ and $E(\text{Li}_0\text{Mn}\text{O}_3)$ are respectively the total energies of the stable reference materials at the end points of the charging curve ($x=2$ and $x=0$). The difference between the formation energies of different phases is typically on the scale of 1 meV $\approx 0.04~ \text{mHa}$~\cite{urban2016computational, lee2014structural, lim2015origins, mccoll2022transition}, which sets the required accuracy for computing the energies needed to evaluate~\cref{eq:form-energy}.

The release of molecular oxygen, $\text{O}_2$, upon Li extraction has been observed in experiments and also theoretically predicted~\cite{zhang2022pushing}. Quite recently, McColl {\it et al.}~\cite{mccoll2022transition} investigated the stability of the delithiated phases of the disordered rocksalt structures of $\text{Li}_2\text{MnO}_2\text{F}$ and found that the formation of $\text{O}_2$, following transition metal migration, is the dominant redox mechanism. The thermodynamical driving force for this reaction is the oxygen-vacancy formation energy~\cite{lee2014structural,lim2015origins}, defined as 
\begin{align}
E_{\text{O}_2}(x) = E(\text{Li}_x\text{MnO}_{2-\delta}\text{F}) & - E(\text{Li}_x\text{MnO}_{2}\text{F}) \nonumber \\
& + \frac{\delta}{2}E(\text{O}_2), 
\label{eq:form-o2}
\end{align}
where $\delta$ denotes the number of oxygen atoms removed per formula unit and $E(\text{O}_2)$ is the ground-state energy of the oxygen molecule. 

Different strategies have been proposed to mitigate the capacity loss of Li-excess materials~\cite{li2014k+, qiu2016gas, shi2016mitigating, shin2018alleviating, ku2018suppression, pei2020reviving, chen2020highly, eum2020voltage}. Most of them rely on modifying the composition and/or the atomic structure of the material to suppress transition metal migration in the delithiated cathode. To that end, computing accurate site energies, i.e., the ground-state energy of the material as the transition metal occupies different lattice sites, has proven to be useful~\cite{eum2020voltage}. Furthermore, from the ground-state energy of the transition state (TS) of the material along the transition metal migration path~\cite{neb_method}, it is possible to compute the activation energy $E_\text{TS}-E_0$ to describe the kinetic pathways and, crucially, determine the reversibility of the structural transformations of the material.

The accuracy of the electronic structure method used to simulate Li-excess materials is key as it impacts the entire computational methodology used to investigate these materials~\cite{zhang2022pushing, seo2016structural, eum2020voltage, mccoll2022transition}. Due to the practical computational limitations of more accurate methods, we only have access to approximate density functional theory (DFT) methods which can introduce significant errors as they rely on approximate parametrizations of the exchange-correlation density functional. In particular, the simulation of key battery properties requires computing the difference of ground-sate energies of materials with very different electronic structures, see for example~\cref{eq:form-energy} to compute the formation energy. In these scenarios, DFT approximations do not benefit from cancellation of errors, and very accurate energies need to be computed. Standard local and semi-local approximations to the density functional can not properly describe the strong on-site electronic correlations between the $d$ electrons in the transition metal~\cite{urban2016computational}. 

This problem is partially mitigated by adding a Hubbard-like term in the Kohn-Sham Hamiltonian, the so-called DFT+U method~\cite{jain2011formation}. However, the value of the Hubbard parameter $U$, typically obtained from experimental data~\cite{seo2016structural, eum2020voltage, mccoll2022transition}, is strongly system-dependent which limits applicability of the method to explore new materials. In addition, authors in Ref.~\cite{seo2016structural} noted that DFT+U cannot predict the band structure of Li-excess materials with the required accuracy. Instead, they used a hybrid functional~\cite{heyd2004efficient} which incorporates a fraction of the exact exchange from Hartree-Fock theory as part of the exchange-correlation functional. This approach also depends on an adjustable parameter selecting the amount of HF exchange included in the calculations, and finding its optimal value is an issue if no experimental data is available. Finally, Zhang {\it et al.}~\cite{zhang2022pushing} have recently pointed out that these approximations may break down in the presence of oxygen-oxygen bonding which is one of the most important relaxation process in delithiated Li-excess materials~\cite{mccoll2022transition}.

In summary, relevant electronic structure calculations for Li-excess materials include formation energies of delithiated phases, oxygen-vacancy formation energies, site energies, and activation energies for kinetic pathways. Any of these simulations can be reduced to a series of ground-state energy calculations, each of which can be performed using our quantum algorithm. Note that the proposed algorithm is a full first-principles approach that, for a given structural model of the target material, can be used to compute its ground-state energy with guaranteed precision using a fault-tolerant quantum computer. It does not depend on any semi-empirical parameter and it can be used to simulate any material consisting of atomic species for which HGH PPs are accessible. We now study the resources required to implement our quantum algorithm for the ground-state energy calculations for Li-excess materials.

{
	\renewcommand{\arraystretch}{1.5}
	\begin{table*}[!ht]
	\centering
 \resizebox{\textwidth}{!}{
		\begin{tabular}{|Sc|Sc|Sc|Sc|}
  \hline
			\shortstack{{\bf Material}} & \shortstack{\limno} & \shortstack{\limnnio} & \shortstack{\limnfo}\\
			\hline
			{\bf Crystal system / space group} & Monoclinic / [C 2/m] & Hexagonal / [$\text{P}6_3\text{mc}$]  & Cubic / [Fm3m]
			\\\hline
   			{\bf Supercell size} & $2 \times 2 \times 1$ & $2 \times 3 \times 2$ & $3 \times 2 \times 2$ \\ \hline
   \shortstack{\\{\bf Lattice vectors} ($\text{\AA}$) } & \shortstack{$\vec{a}_1=(10.02, 0, 0)$ \\ $\vec{a}_2=(0, 17.32, 0)$ \\ $\vec{a}_3=(-1.6949, 0, 4.7995)$}  & \shortstack{$\vec{a}_1=(5.7081, 0, 0)$ \\ $\vec{a}_2=(-4.2811, 7.4151, 0)$ \\ $\vec{a}_3=(0, 0, 19.6317)$} & \shortstack{$\vec{a}_1=(12.48, 0, 0)$ \\ $\vec{a}_2=(0, 8.32, 0)$ \\ $\vec{a}_3=(0, 0, 8.32)$} \\ \hline
   {\bf Supercell volume} ($\text{\AA}^3$) & 832.9405 & 830.9604 & 863.8955 \\ \hline
   $N_\text{atoms}$ & 72 (8 Li, 16 Mn, 48 O) & 90 (22 Li, 14 Mn, 6 Ni, 48 O) & 76 (12 Li, 16 Mn, 16 F, 32 O) \\ \hline
   $\eta^\text{PP}~(\eta^\text{AE})$ & 408 (808) & 468 (968) & 428 (836) \\ \hline
   $N^\text{PP}~(N^\text{AE})$ & 55,473 ($5.8 \times 10^8$) & 67,767 ($8.7 \times 10^8$) & 57,655 ($6.4 \times 10^8$) \\ \hline
		\end{tabular}}
		\caption{Crystal lattice parameters for the materials selected to perform resource estimation. Recall that $\eta$ denotes the number of electrons in the supercell structural models and $N$ is the number of plane waves required to converge the ground-state energy of the material at the level of density functional theory. The superscripts PP and AE are used to differentiate between the pseudopotential and all-electron cases, respectively. The structural models of the selected materials were built according to Refs.~\cite{lim2015origins, eum2020voltage, mccoll2022transition}.}
		\label{table:materials}
	\end{table*}
}

\section{Resource estimation}
\label{sec:resource_estimation}
We now estimate the number of qubits and Toffoli gates needed to implement our pseudopotential-based algorithm for three different Li-excess materials: lithium manganese oxide ($\text{Li}_2\text{MnO}_3$), lithium nickel-manganese oxide ($\text{Li}[\text{Li}_{0.17}\text{Ni}_{0.25}\text{Mn}_{0.58}]\text{O}_2$), denoted as LLNMO, and lithium manganese oxyfluoride ($\text{Li}_{2}\text{MnO}_2\text{F}$). To that end we have built supercell structural models of these materials as described in detail in Refs.~\cite{lim2015origins, eum2020voltage, mccoll2022transition}. The atomic models are shown in \cref{fig:mat-structures}. The lattice parameters of the supercells are summarized in \cref{table:materials}, and the procedure to delithiate the pristine materials is described in \cref{app:models}. 
\begin{figure*}[!ht]
\centering
\includegraphics[width= 1\textwidth]{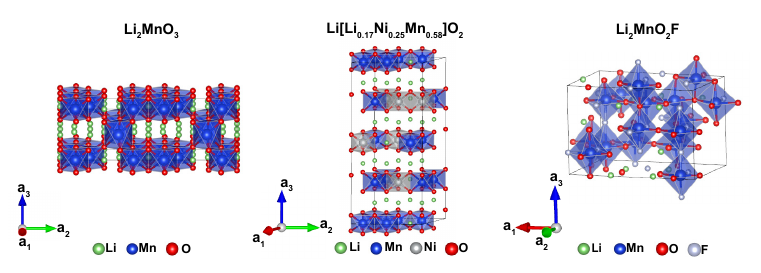}
\caption{Representation of the atomic structures of the Li-excess materials selected to perform the resource estimation of the quantum algorithm. These structural models corresponds to the supercells described in \cref{table:materials}. The polyhedra depict the octahedrally coordinated transition metals. These figures were generated using the VESTA package~\cite{vesta}.}
\label{fig:mat-structures}
\end{figure*}

The resource estimation results targeting chemical accuracy are shown in \cref{table:materials_NPP_vs_NAE_cost}. All calculations are performed using our resource estimation software, which is available at \url{https://github.com/XanaduAI/pseudopotentials}. This code allows us to compute the resources of both the pseudopotential and the all-electron algorithms for any material. In all cases, we see that the pseudopotential-based algorithm has gate counts that are roughly four orders of magnitude lower than in the all-electron case. Qubit numbers are also considerably lower, roughly needing less than half as many qubits. 
{
	\renewcommand{\arraystretch}{1.5}
\begin{table*}[ht]
\centering
		\begin{tabular}{|Sc||Sc|Sc||Sc|Sc|}
 \hline
		\multirow{2}{*}{Material}& \multicolumn{2}{Sc||}{Qubit cost}&\multicolumn{2}{Sc|}{Toffoli cost}\\
        \cline{2-5}
        & PP & AE &  PP & AE \\\hline
		\limno  & \textbf{9808} & 24974  & $\mathbf{5.00\times 10^{15}}$ & 2.16$\times 10^{19}$
        \\\hline
        \limnnio & \textbf{11130} & 29784  & $\mathbf{4.84\times 10^{15}}$ & 3.67$\times 10^{19}$
		\\\hline
		\limnfo  & \textbf{10260} & 26629 &  $\mathbf{3.87\times 10^{15}}$ & 1.18$\times 10^{19}$
		\\\hline
        \dis & \textbf{2650} & 4859 & $\mathbf{6.38 \times 10^{13}}$ & 1.59$\times 10^{17}$\\\hline
	\end{tabular}
	\caption{Resource estimation using $N^{\text{PP}}$ and $N^{\text{AE}}$ plane waves as specified in \cref{table:materials} for the pseudopotential (PP) and the all-electron (AE) algorithm, respectively. For \dis, we have $N^{\text{PP}} = 19,549, N^{\text{AE}} = 5.46\times 10^{7}$ and $\eta^{\text{PP}}=100,\eta^{\text{AE}}=150$. The rest of parameters needed are supplied in \cite{batterypaper}. Better numbers are indicated in bold. Qubit cost corresponds to logical qubits.}
	\label{table:materials_NPP_vs_NAE_cost}
\end{table*}
}

We note that the cost of the pseudopotential (PP) and all-electron algorithms are comparable even when using the same number of planewaves. For \limno~and $N=10^5$, the Toffoli cost for PP is $2.68 \times 10^{14}$, and for the all-electron is $1.44\times 10^{14}$. We refer to the tables in \cref{app:res_est_det_results} for comparisons for each material. The cost competitiveness for simulations using the same number of plane waves is not just the result of only having to use half the number of electrons with the pseudopotential, but also a result of our tailored compilation strategy, which involves a suitable LCU to lower the value for $\lambda$, and the use of QROM for the state preparation method. When targeting the same energy accuracy, the cost of the all-electron algorithm is significantly higher as the number of electrons increases, and a very large number of plane waves are needed to properly describe the quantum system (see~\cref{table:materials}).

The number of plane waves used to perform resource estimation are reported in \cref{table:materials}. These numbers were estimated by performing a convergence analysis of the total energy of the materials using density functional theory calculations with the Perdew-Burke-Ernzerhof (PBE) exchange–correlation functional~\cite{perdew1996generalized} (see \cref{app:NPWs}). Using pseudopotentials, we achieve convergence with roughly $6\times 10^4$ plane waves. For the all-electron calculations, we need approximately $7 \times 10^8$ plane waves.

Calculations using the HGH pseudopotentials were performed using the {\sc Quantum ESPRESSO} package~\cite{giannozzi2009quantum}, and the total energy was computed for the Gamma ($\Gamma$)-point only, which is sufficient to represent the large systems studied in this work~\cite{payne1992iterative}. A cutoff energy of $70~\text{Ry}$ $\approx$ 950 eV was set for both $\text{Li}_{0.5}\text{MnO}_3$ and $\text{Li}_{0.75}\text{MnO}_2\text{F}$, while a cutoff of 80 Ry $\approx$ 1000 eV was set for $\text{Li}[\text{Li}_{0.17}\text{Ni}_{0.25}\text{Mn}_{0.58}]\text{O}_2$ (see Appendix Fig.~\ref{fig:pw-classical}). For the all-electron case, we used the localized augmented plane-wave plus local orbitals (LAPW+lo) method as implemented in the {\sc WIEN2k} code~\cite{blaha2020wien2k}. In this case, a larger number of plane waves are needed to achieve convergence of the total energy. The most important parameters that have to be considered are the muffin-tin radii ($R_\text{MT}$) and the plane-wave cutoff ($K_\text{max}$) for the expansion of the wave function in the interstitial. The basis set cutoff parameter is defined by the product $R_\text{MT}K_\text{max}$, which was set to 7 for all the structures in this work. In all cases, the convergence energy was calculated without applying structure relaxation as this would not significantly affect the total number of plane waves needed to achieve convergence.

Using the data in \cref{table:materials}, we can measure the Toffoli and logical qubit cost of our algorithm. These estimates are benchmarked against the all-electron setting for the same materials. Results are also shown for the dilithium iron silicate cathode \dis, studied earlier in \cite{batterypaper}. We only report the cost of performing one round of quantum phase estimation. A more comprehensive analysis of total cost should include the overhead due to a limited overlap of the initial state and the cost of initial state preparation. We can also consider the T-gate cost of the algorithm, but these are at least an order of magnitude lower than that of the qubitization-based QPE~\cite{batterypaper}.

We also report the depth of the circuits along with the clean portion of the qubit cost in \cref{app:res_est_det_results} illustrating further the advantage of using pseudopotentials and QROMs in our algorithm. We compute this for the following reasons. First, studying depth is the first step towards answering the more complex and interesting question of algorithmic runtime. Second, compared to other state preparation methods, QROM offers far more space-depth trade-off flexibility, something we would like to leverage. Also, given the high costs of the algorithm, it is reasonable to assume that capable hardware would also offer the possibility of simultaneous Toffoli applications. 

To estimate the Toffoli depth, we have to specify certain parameters such as the maximum number $n_{\text{dirty}}$ of dirty qubits available for QROM, the maximum allowed number $n_{\text{tof}}$ of simultaneous Toffoli applications. More details are provided in \cref{app:res_est_det_results}. Notice that due to the relatively significant use of QROMs, the PP-based algorithm stands to benefit more than the AE algorithm on circuit depth reduction. As an example, for \limnnio, the clean qubit cost and Toffoli depth for PP are respectively 11,130 and 9.59$\times 10^{14}$, while for AE they are 29,784 and 3.59$\times 10^{19}$.

Lastly, note that while computing depth gets us one step closer to physical resource estimation, we are still estimating resources at the logical level, i.e. estimating the depth of the logical circuit. One of the assumptions in the logical setting regarding QROM is the all-to-all connectivity of the logical qubits, an assumption that could break down in realistic hardware. However, such geometrical constraints can be overcome through clever architecture designs and lattice surgery techniques \cite{litinski2018lattice,litinski2019game}, showing only a logarithmic overhead of Clifford gates is incurred to bring together far away encoded logical qubits on the hardware.

\section{Conclusions}\label{sec:conclusion}
This work introduced the first example of a quantum algorithm that makes use of ionic pseudopotentials to simulate periodic materials. Our main technical contributions are a collection of highly-optimized compilation strategies for qubitization-based quantum phase estimation in first quantization. They are designed to reduce the cost of implementing the algorithm despite the mathematical complexity of pseudopotential operators. A key ingredient is the use of quantum read-only memories to avoid performing complicated arithmetic operations on a quantum computer. This also helps navigate tradeoffs between the number of qubits and number of gates in the algorithm, which can be exploited to reduce costs. Overall, we reduce the cost of the quantum algorithm by about four orders of magnitude compared to the all-electron approach described in Ref.~\cite{su2021fault} when applied to simulating lithium-excess materials for a fixed target accuracy of the ground-state energy.

Using these tailored quantum algorithms, we estimated the number of qubits and Toffoli gates needed to simulate Li-excess cathode materials proposed in the literature. In each of these cases, a large supercell is necessary to ensure the quality of the simulation, which considerably increases the cost of the all-electron approach, making the use of pseudopotentials even more important. On the other hand, even though our algorithm is applicable to any lattice, it benefits from the orthogonality of the lattice vectors, as this reduces the cost of quantum read-only memory techniques for preparing relevant Gaussian superposition states. It is desirable to study strategies where the cost of implementing the quantum algorithm is less dependent on the specific structure of the lattice.

The accuracy of the quantum algorithm depends on the quality of the pseudopotentials. Therefore, in selecting the Hartwigsen-Goedecker-Hutter pseudopotential we rely on extensive benchmarking by the density functional theory community that have identified it as accurate and transferable~\cite{lejaeghere2016reproducibility}. Additionally, we have constrained our analysis to separable pseudopotentials with one projector per angular momentum channel in the non-local component of the pseudopotential operator. Possible extensions of this work might consider more than one projector and the use of ultrasoft pseudopotentials~\cite{vanderbilt1990} to further reduce the number of plane-wave basis functions as observed in classical algorithms.

Quantum computing offers the potential to perform high-accuracy simulations of strongly-correlated systems of unprecedented size. This is an outstanding challenge for classical methods, that either cannot offer the same accuracy guarantees, or incur prohibitive costs. Still, to realize the potential advantages of quantum computing, further efforts are invaluable to reduce the cost of quantum algorithms. For example, the quality of the initial state preparation method should be addressed, since a single round of quantum phase estimation may need to be repeated too many times for states with poor overlap. 

Overall, we have demonstrated that the method of pseudopotentials can be effectively applied to reduce the cost of quantum algorithms to simulate battery materials that require large supercell structural models. This can potentially unlock the application of quantum computing to address more complicated processes involving doped materials as well as chemical reactions at the electrode-electrolyte interface.

\section{Acknowledgments}
\label{sec:acknowledgments}
We thank Matthew Kiser, Stepan Fomichev, Soran Jahangiri, and Nathan Killoran for their fruitful comments.

\bibliographystyle{plainnat}
\bibliography{main}

\newpage
\onecolumngrid
\appendix
\vspace{1cm}
\begin{center}
    \large{\textbf{Appendix}}
\end{center}
\section{Plane wave matrix elements of the pseudopotential operator}
\label{sec:pp_me}
In this section we derive the expressions for the matrix elements of the local and non-local components of the pseudopotential operator described in~\cref{ssec:hgh} in a plane-wave basis.

\subsection{Matrix elements of the local potential}
\label{ssec:loc_pot_me}
For an ion located at the coordinates $\bm{R}$, the plane-wave matrix element of the local potential is defined by the integral
\begin{equation}
u_{pq}^\mathrm{loc}(\bm{R}) = \frac{1}{\Omega} \int d\bm{r}^3 u^\text{loc}(||\bm{r}-\bm{R}||) e^{i \bm{G}_\nu\cdot \bm{r}},
\label{eq:u_loc_1}
\end{equation}
where $\bm{G}_\nu = \bm{G}_q-\bm{G}_p$. By changing the variable $\bm{r} \to \bm{r} + \bm{R}$ we obtain
\begin{equation}
u_{pq}^\mathrm{loc}(\bm{R}) = \frac{1}{\Omega}e^{i \bm{G}_\nu \cdot \bm{R}} \int d\bm{r}^3 u^\text{loc}(r) e^{i \bm{G}_\nu\cdot \bm{r}},
\label{eq:u_loc_2}
\end{equation}
with $r = \|\bm{r}\|$. In spherical coordinates, the integral above transforms as
\begin{align}
u_{pq}^\mathrm{loc}(\bm{R}) & = \frac{2\pi}{\Omega} e^{i \bm{G}_\nu \cdot \bm{R}} \int_0^\infty dr~u^\text{loc}(r)\int_0^\pi d\theta \cos(G_\nu r \cos(\theta))~r^2\sin(\theta).
\label{eq:u_loc_3}
\end{align}
To integrate the angular variable we use
\begin{align}\label{eq:u_loc_4}
\mathrm{cos}[G_\nu r \mathrm{cos}(\theta)] r^2 \mathrm{sin}(\theta) = r \frac{d}{d\theta} \left[ - \frac{\mathrm{sin}(G_\nu r \mathrm{cos}(\theta))}{G_\nu} \right],
\end{align}
which simplifies \cref{eq:u_loc_3} to the integral over the radial variable
\begin{align}
u_{pq}^\mathrm{loc}(\bm{R}) = \frac{4\pi}{\Omega} & e^{i \bm{G}_\nu \cdot \bm{R}} \left[ \frac{1}{G_\nu} \int_0^\infty dr~r~u^\text{loc}(r)~\text{sin}(G_\nu r) \right].
\label{eq:u_loc_5}
\end{align}
Note that~\cref{eq:u_loc_5} is general and it can be used to compute the matrix element of any given local potential $u^\text{loc}(r)$. For the case of the Hartwigsen-Goedecker-Hutter (HGH) pseudopotential we insert~\cref{eq:hgh_loc} into~\cref{eq:u_loc_5} to obtain the following expression:
\begin{align}
u_{pq}^\mathrm{loc}(\bm{R}) & = \frac{4\pi e^{i \bm{G}_\nu \cdot \bm{R}}}{\Omega}  \left[ \frac{-Z_\mathrm{ion}}{G_\nu} \int_0^\infty dr ~\text{erf}(\alpha r)~\text{sin}(G_\nu r) + \sum_{i=1}^4 \frac{C_i}{G_\nu} \int_0^\infty dr~r~(\sqrt{2}\alpha r)^{2i-2}~e^{-(\alpha r)^2}~\text{sin}(G_\nu r) \right]
\label{eq:loc_6}
\end{align}
where $\alpha=\frac{1}{\sqrt{2} r_\text{loc}}$. The first integral is computed as follows:
\begin{equation}
     \kern-6pt -\frac{Z_\mathrm{ion}}{G_\nu} \int_0^\infty dr ~\text{erf}(\alpha r)~\text{sin}(G_\nu r) =  -\frac{Z_{\mathrm{ion}}e^{-G_\nu^2\alpha^{-2}/4}}{G_\nu^2}.
    \label{eq:loc_7}
\end{equation}

To evaluate the second term we compute the integrals:
\begin{align}
&\frac{C_1}{G_\nu} \int_0^\infty dr~r~e^{-(\alpha r)^2}~\text{sin}(G_\nu r) = C_1 \frac{\sqrt{\pi}}{2} r_\text{loc}^3 e^{-(G_\nu r_\text{loc})^2/2}, \label{eq:loc_8} \\
&\frac{2C_2}{G_\nu} \int_0^\infty dr~r~(\alpha r)^2~e^{-(\alpha r)^2}~\text{sin}(G_\nu r) = C_2 \frac{\sqrt{\pi}}{2} \left[3 r_\text{loc}^3 - r_\text{loc}^5 G_\nu^2 \right] e^{-(G_\nu r_\text{loc})^2/2}, \label{eq:loc_9} \\
& \frac{4C_3}{G_\nu} \int_0^\infty dr~r~(\alpha r)^4~e^{-(\alpha r)^2}~\text{sin}(G_\nu r) = C_3 \frac{\sqrt{\pi}}{2} \left[15 r_\text{loc}^3 - 10 r_\text{loc}^5 G_\nu^2 + r_\text{loc}^7 G_\nu^4 \right] e^{-(G_\nu r_\text{loc})^2/2}, \label{eq:loc_10} \\
& \frac{8C_4}{G_\nu} \int_0^\infty dr~r~(\alpha r)^6~e^{-(\alpha r)^2}~\text{sin}(G_\nu r) = C_4 \frac{\sqrt{\pi}}{2} \left[105 r_\text{loc}^3 - 105 r_\text{loc}^5 G_\nu^2 + 21 r_\text{loc}^7 G_\nu^4 - r_\text{loc}^9 G_\nu^6 \right] e^{-(G_\nu r_\text{loc})^2/2}. \label{eq:loc_11}
\end{align}

Using \cref{eq:loc_7} and Eqs. ~(\ref{eq:loc_8}) to ~(\ref{eq:loc_11}) we obtain a final expression for the matrix elements of the local potential
\begin{align}\label{eq:u_loc_mast_eq}
u_{pq}^\mathrm{loc}(\bm{R}) = \frac{4\pi}{\Omega} e^{i \bm{G}_\nu \cdot \bm{R}}~e^{-(G_\nu r_\text{loc})^2/2} \bigg\{-\frac{Z_{\mathrm{ion}}}{G_\nu^2} + \frac{\sqrt{\pi}}{2} \Big[& C_1 r_\text{loc}^3 + C_2(3r_\text{loc}^3 - 5r_\text{loc}^5 G_\nu^2) \nonumber \\
& + C_3(15r_\text{loc}^3 - 10r_\text{loc}^5 G_\nu^2 + r_\text{loc}^7 G_\nu^4) \nonumber \\
& + C_4(105r_\text{loc}^3 - 105r_\text{loc}^5 G_\nu^2 + 21 r_\text{loc}^7 G_\nu^4 - r_\text{loc}^9 G_\nu^6) \Big] \bigg\}.
\end{align}

\subsection{Matrix elements of the non-local potential}
\label{ssec:non_local_me}

The matrix elements of the non-local potential $u^\text{NL}$ defined in ~\cref{eq:hgh_nl} are given by
\begin{align}
u^\mathrm{NL}_{pq} = \sum_{lm} \sum_{i,j}~\langle \varphi_p \vert \beta_i^{(lm)} \rangle~B_{ij}^{(l)}~\langle \beta_j^{(lm)} \vert \varphi_q \rangle,
\label{eq:gnl_2}
\end{align}
where $\langle \bm{r} \vert \beta_i^{(lm)} \rangle =\beta_i^{(lm)}(\bm{r}-\bm{R}) = \beta_i^{(l)}(||\bm{r}-\bm{R}||)Y_{lm}(\phi,\theta)$, the coefficients $B_{ij}^{(l)}$ are reported in Ref.~\cite{hartwigsen1998}, and $\bm{R}$ denotes the coordinates of an ion in the material. Here $Y_{lm}(\phi,\theta)$ denote spherical harmonics. By changing the variable $\bm{r} \to \bm{r} + \bm{R}$ the overlap integral $\langle \varphi_p \vert \beta_i^{(lm)} \rangle$ writes as
\begin{equation}
\langle \varphi_p \vert \beta_i^{(lm)} \rangle = e^{- i \bm{G}_p \cdot \bm{R}} \int d\bm{r} \varphi_p(\bm{r}) \beta_i^{(lm)}(\bm{r}).
\label{eq:gnl_4}
\end{equation}
The plane wave function $\varphi_p(\bm{r})$ can be expanded in terms of spherical harmonics $Y_{lm}(\phi, \theta)$ and Bessel functions $j_l(G_pr)$ \cite{abramowitz1964handbook}
\begin{equation}
\varphi_p(\bm{r}) = \frac{4\pi}{\sqrt{\Omega}} \sum_{lm} i^l j_l(G_pr) Y_{lm}^*(\hat{\bm{G}}_p) Y_{lm}(\hat{\bm{r}}),
\label{eq:pw_expansion}
\end{equation}
where $\hat{\bm{r}}$ and $\hat{\bm{G}}_p$ are unitary vectors. Using this expansion we compute the overlap in~\cref{eq:gnl_4} as  $\langle \varphi_p \vert \beta_i^{(lm)} \rangle =$
\begin{align}
 & \frac{4 \pi e^{- i \bm{G}_p \cdot \bm{R}}}{\sqrt{\Omega}}  \sum_{l'm'} (-i)^{l'} Y_{l'm'}(\hat{\bm{G}}_p) \int_0^\infty dr~r^2 j_{l'}(G_pr) \beta_i^{(l)}(r) \int_0^{2\pi} d\phi \int_0^\pi d\theta~\mathrm{sin}(\theta) Y_{l'm'}^*(\phi, \theta) Y_{lm}(\phi, \theta) \nonumber\\ \\
& = \frac{4 \pi}{\sqrt{\Omega}} e^{- i \bm{G}_p \cdot \bm{R}} (-i)^{l} Y_{lm}(\hat{\bm{G}}_p) \int dr~r^2 j_{l}(G_pr) \beta_i^{(l)}(r).
\label{eq:gnl_5}
\end{align}
Similarly, the projection $\langle \beta_j^{(lm)} \vert \varphi_q \rangle$ is given by
\begin{equation}
\langle \beta_j^{(lm)} \vert \varphi_q \rangle = \frac{4 \pi}{\sqrt{\Omega}} e^{i \bm{G}_q \cdot \bm{R}} i^l Y_{lm}^*(\hat{\bm{G}}_q) \int dr~r^2 j_{l}(G_qr) \beta_j^{(l)}(r).
\label{eq:gnl_6}
\end{equation}
Plugging Eqs. \eqref{eq:gnl_5}-\eqref{eq:gnl_6} into \cref{eq:gnl_2} we obtain an expression for computing the matrix elements in \cref{eq:gnl_2} for a given set of projectors $\beta_i^l(r)$
\begin{align}
u^\mathrm{NL}_{pq}(\bm{R}) = & \frac{4\pi}{\Omega} e^{i \bm{G}_\nu \cdot \bm{R}} \sum_{l=0}^{l_\mathrm{max}} (2l+1) P_l(\hat{\bm{G}}_p \cdot \hat{\bm{G}}_q) \sum_{i,j}~\langle p \vert i \rangle_l~B_{ij}^{(l)}~\langle j \vert q \rangle_l
\label{eq:gnl_7},
\end{align}
where $\bm{G}_\nu=\bm{G}_q-\bm{G}_p$, $P_l$ is a Legendre polynomial, and $\langle p \vert i \rangle_l$ is an overlap integral over the radial coordinates given by
\begin{equation}
\langle p \vert i \rangle_l = \int dr~r^2 j_{l}(G_pr) \beta_i^{(l)}(r).
\label{eq:gnl_8}
\end{equation}
To compute the matrix elements of the non-local component of the Hartwigsen-Goedecker-Hutter (HGH) pseudopotenial, we use the HGH projectors defined as:
\begin{equation}
\beta_i^{(l)}(r) = A_l^i r^{l+2i-2} \mathrm{exp}\left[ -\frac{1}{2}\left(\frac{r}{r_l}\right)^2 \right],
\label{eq:app_hgh_gaussians}
\end{equation}
where the radii $r_l$ give the range of the $l$-dependent projectors, and
\begin{equation}
A_l^i= \frac{\sqrt{2}}{r_l^{l+(4i-1)/2}\Gamma(l+\frac{4i-1}{2})},
\label{eq:cnst_hgh_pp}
\end{equation}
with $\Gamma$ denoting the gamma function. For the case of one projector ($i=1$) and $l_\text{max} \leq 2$, the overlap integrals defined by~\cref{eq:gnl_8} are calculated as:
\begin{align}
\langle p \vert 1 \rangle_0 & =  \int_0^\infty dr~r^2 \frac{\mathrm{sin}(G_pr)}{(G_p r)}~\beta_1^{(0)}(r) = 2 r_0^{3/2}~e^{-(G_pr_0)^2/2}, \label{eq:hgh_9} \\ \nonumber \\
\langle p \vert 1 \rangle_1 & = \int_0^\infty dr~r^2 \left[ \frac{\mathrm{sin}(G_p r)}{(G_p r)^2} - \frac{\mathrm{cos}(G_pr)}{(G_p r)} \right]~\beta_1^{(1)}(r) =\frac{4}{3} r_1^{5/2}~G_p~e^{-(G_pr_1)^2/2},
\label{eq:hgh_10} \\ \nonumber \\
\langle p \vert 1 \rangle_2 &= \int_0^\infty dr~r^2 \left[ \frac{3\mathrm{sin}(G_pr)}{(G_p r)^3}  - \frac{\mathrm{sin}(G_pr)}{(G_p r)} - \frac{3\mathrm{cos}(G_pr)}{(G_p r)^2} \right]~ \beta_1^{(2)}(r)=\frac{8}{15}r_2^{7/2}~G_p^2~e^{-(G_pr_2)^2/2} \label{eq:hgh_11}.
\end{align}
Using~\eqref{eq:hgh_9}-\eqref{eq:hgh_11} we compute the product of the projections $\braket{p|i}_l  \braket{i|q}_l$ entering~\cref{eq:gnl_7} to obtain the final expression for the matrix elements
\begin{align}
u^\mathrm{NL}_{pq}(\bm{R}) = \frac{4\pi}{\Omega}  e^{i \bm{G}_\nu \cdot \bm{R}} \Bigg\{ 
& 4B_0 r_0^3~e^{-(G_p^2+G_q^2)r_0^2/2} + \frac{16B_1 r_1^5}{3}~(\bm{G}_p \cdot \bm{G}_q)~e^{-(G_p^2+G_q^2)r_1^2/2} +
\nonumber \\
& \left[ \frac{32B_2 r_2^7}{15}~(\bm{G}_p \cdot \bm{G}_q)^2 + \frac{32B_2 r_2^7}{45}~(G_p G_q)^2\right]e^{-(G_p^2+G_q^2)r_2^2/2} \Bigg\},
\label{eq:hgh_12}
\end{align}
where the coefficient $B_l:=B_{11}^{(l)}$ in~\cref{eq:gnl_7}.
\begin{figure}[!ht]
\centering
\includegraphics[width=1 \columnwidth]{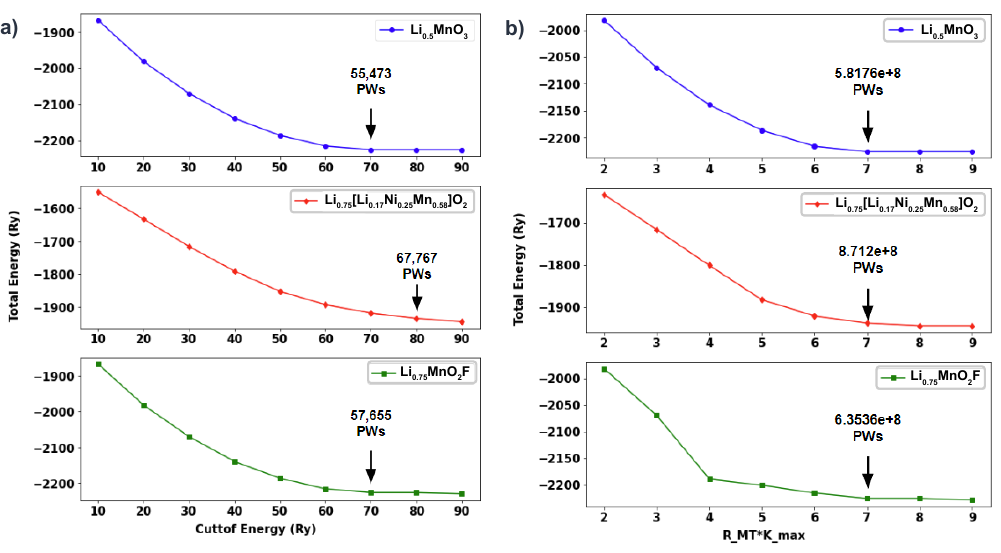}
\caption{Convergence of the ground-state energy for the selected materials computed using DFT at the level of the PBE-GGA exchange-correlation functional~\cite{perdew1992atoms}. a) Energies computed using plane waves and the HGH pseudopotentials as implemented in the {\sc Quantum ESPRESSO} package. b) Energies obtained for the all-electron calculations using {\sc WIEN2k} as a function of the cutoff parameters $R_\text{MT}$ and $K_\text{max}$. The number of plane waves used to perform the resource estimations for each material is indicated with arrows.}
\label{fig:pw-classical}
\end{figure}

\section{Structural models of lithium-excess materials}\label{app:models}
Here we provide more details on how structural models were built for the lithium-excess materials studied in this work. These models are used for determining the number of plane waves needed for convergence of the total energy, computed using density functional theory, and for resource estimation of the quantum algorithm. 

For $\text{Li}_{0.5}\text{MnO}_3$, we first modify the position of the Mn atoms in the pristine structure to generate the spinel-like phase simulated in Ref~\cite{lim2015origins}. Thus, we delithiated the material following the results reported in Ref.~\cite{lee2014structural}. The structural model for the LLNMO material corresponds to the most stable structure of the $\text{Li}_{0.75}[\text{Li}_{0.17}\text{Ni}_{0.25}\text{Mn}_{0.58}]\text{O}_2$ material predicted in Ref.~\cite{eum2020voltage}. Similarly, we removed one lithium layer from the supercell to estimate the resources to compute the site energies reported in Ref.~\cite{eum2020voltage}. Finally, the structural model for the delithiated $\text{Li}_{0.75}\text{MnO}_2\text{F}$ material was obtained by randomly taking Li ions off the pristine material until we match the Li-ion concentration of $x=0.75$. The relaxed structure of the fully lithiated material was taken from Ref.~\cite{sharpe2020redox}.

\section{Number of plane waves used to perform resources estimation of the quantum algorithm}\label{app:NPWs}

The purpose of estimating the number of plane waves is to have an appropriate number that we can use to perform resource estimations of the quantum algorithm. To this aim, we estimate the number of plane waves needed to perform a classical calculation using the structures studied in this work. We conducted an energy convergence test by changing the cutoff energy and calculating the total energy at the $\Gamma$-point (Fig.~\ref{fig:pw-classical}). Each point in Fig.~\ref{fig:pw-classical}a plot corresponds to the total energy value, calculated by using the {\sc Quantum ESPRESSO} package, for a given kinetic energy cutoff. Fig.~\ref{fig:pw-classical}b) depicts a similar analysis conducted using the {\sc WIEN2k} program, an all-electron approach, where the cutoff in the reciprocal space is given by the ($R_\text{MT}$ and $K_\text{max}$) parameter. The number of plane waves was chosen to ensure that the total energy variation was less than 1 kcal/mol = 0.043 eV (chemical accuracy).

\section{Linear combination of unitaries for the non-local potential}\label{app:U_NL_LCU}
This section provides details on the construction of the linear combination of unitaries (LCU) decomposition for the non-local term of the pseudopotential operator.  The LCU for the non-local term exploits the projector representation of the operator, and breaks down each projection $P$ to the sum of unitaries $\frac{1}{2}\mathbbm{1} - \frac{1}{2}(\mathbbm{1}-2P)$. Below we demonstrate how this is done in details. Recall
\begin{align}
U_{NL} = & \sum_{j=1}^\eta \sum_{I=1}^L f_I(\bp,\bq) \ket{\bp}_j \bra{\bq}_j, \\ \label{eq:f_Ipq}
f_I(\bp,\bq) = \frac{4\pi}{\Omega}e^{i \bG_{\bnu} \cdot \bR_I} \Bigg\{& 4r_0^3B_0~e^{-(G_p^2+G_q^2)r_0^2/2} + \frac{16r_1^5B_1}{3}~(\bm{G}_p \cdot \bm{G}_q)~e^{-(G_p^2+G_q^2)r_1^2/2} \nonumber \\
& + \left[ \frac{32 r_2^7B_2}{15}~(\bm{G}_p \cdot \bm{G}_q)^2 + \frac{32r_2^7B_2}{45}~(G_p G_q)^2\right]e^{-(G_p^2+G_q^2)r_2^2/2} \Bigg\}.
\end{align}
We decompose $U_{NL}$ by focusing on each individual term in the function $f_I$. Define:
\begin{align}\label{eq:Psi_I,0}
    \ket{\Psi_{I,0}} &:= \frac{1}{\sqrt{\sum_{\bp} e^{-G_p^2r_0^2}}}\sum_{\bp} e^{-G_p^2r_0^2/2}\ket{\bp}.
\end{align}
Now for the term in \cref{eq:f_Ipq} including $r_0$, we define
\begin{align}
    U_{NL,0} &:= \frac{4\pi }{\Omega} \sum_{j=1}^\eta \sum_{I=1}^L \left(\sum_{\bp} e^{-G_p^2r_0^2}\right)4r_0^3B_0  R(\bR_I)^\dagger\ket{\Psi_{I,0}}\bra{\Psi_{I,0}}R(\bR_I) \nonumber \\
    &= \sum_{I=1}^L\frac{8\pi r_0^3B_0 \eta   (\sum_{\bp} e^{-G_p^2r_0^2})}{\Omega  }\mathbbm{1}+ \sum_{j=1}^\eta \sum_{I=1}^L \frac{-8\pi r_0^3B_0(\sum_{\bp} e^{-G_p^2r_0^2})}{\Omega  }R(\bR_I)^\dagger(\mathbbm{1}-2\ket{\Psi_{I,0}}\bra{\Psi_{I,0}})R(\bR_I),
\end{align}
where $R(\bR_I)\ket{\bp} = e^{i\bG_p \cdot \bR_I}\ket{\bp}$ is a phase action and we have suppressed the electron index $j$ for the phase and for the register in state $\ket{\Psi_{I,0}}$. Thus, in the formula above, using the same notations of \cref{eq:lcu_NL} in the main text, for $\sigma=0$ we have 
\begin{align}
c_{I,0} := \frac{-8\pi r_0^3B_0 (\sum_{\bp} e^{-G_p^2r_0^2})}{\Omega }.
\end{align}
The sign of $c_{I,0}$ only depends on $B_0$, i.e., on the atomic type of $I$. For the next term in $U_{NL}$ involving $r_1$, we adopt the notation $\sigma=(1,\omega)$. For $\omega \in \{x,y,z\}$, we define 
\begin{align}\label{eq:Psi_I,1}
    \ket{\Psi_{I,1,\omega}} &:= \frac{1}{\sqrt{\sum_{\bp} G_{p,\omega}^2e^{-G_p^2r_1^2}}}\sum_{\bp}  G_{p,\omega} e^{-G_p^2r_1^2/2}\ket{\bp},
\end{align}
where $\bG_p = (G_{p,1}, G_{p,2}, G_{p,3})$, and we have the following decomposition depending on the coordinate $\omega$:
\begin{align}
    U_{NL,1,\omega} &:= \frac{4\pi }{\Omega} \sum_{j=1}^\eta \sum_{I=1}^L \left(\sum_{\bp} G_{p,\omega}^2e^{-G_p^2r_1^2}\right)\frac{16r_1^5B_1}{3}  R(\bR_I)^\dagger\ket{\Psi_{I,1,\omega}}\bra{\Psi_{I,1,\omega}}R(\bR_I) \nonumber\\
    &= \sum_{I=1}^L\frac{32\pi r_1^5B_1 \eta  (\sum_{\bp} G_{p,\omega}^2 e^{-G_p^2r_1^2})}{3\Omega  } \mathbbm{1} + \nonumber \\
    &\sum_{j=1}^\eta \sum_{I=1}^L \frac{-32\pi r_1^5B_1(\sum_{\bp} G_{p,\omega}^2e^{-G_p^2r_1^2})}{3\Omega  }R(\bR_I)^\dagger(\mathbbm{1}-2\ket{\Psi_{I,1,\omega}}\bra{\Psi_{I,1,\omega}})R(\bR_I)
\end{align}
implying
\begin{align}
    c_{I,1,\omega} =  \frac{-32\pi r_1^5B_1(\sum_{\bp} G_{p,\omega}^2e^{-G_p^2r_1^2})}{3\Omega  }.
\end{align}
For the term including the scalar $32r_2^7B_2/45$, corresponding to the index $\sigma=(2,0)$, we define 
\begin{align}\label{eq:Psi_I,2,0}
    \ket{\Psi_{I,2,0}} &:= \frac{1}{\sqrt{\sum_{\bp} G_p^4e^{-G_p^2r_2^2}}}\sum_{\bp}  G_p^2 e^{-G_p^2r_2^2/2}\ket{\bp},
\end{align}
and note the following decomposition:
\begin{align}
    U_{NL,2,0} &:= \frac{4\pi }{\Omega} \sum_{j=1}^\eta \sum_{I=1}^L (\sum_{\bp} G_p^4e^{-G_p^2r_2^2}) \frac{32r_2^7B_2}{45}  R(\bR_I)^\dagger\ket{\Psi_{I,2,0}}\bra{\Psi_{I,2,0}}R(\bR_I) \nonumber \\
    &= \sum_{I=1}^L\frac{64\pi r_2^7 B_2\eta (\sum_{\bp} G_p^4 e^{-G_p^2r_2^2})}{45\Omega  }\mathbbm{1}+ \nonumber \\
    &\sum_{j=1}^\eta \sum_{I=1}^L \frac{-64\pi r_2^7B_2(\sum_{\bp} G_p^4e^{-G_p^2r_2^2})}{45\Omega  }R(\bR_I)^\dagger(\mathbbm{1}-2\ket{\Psi_{I,2,0}}\bra{\Psi_{I,2,0}})R(\bR_I),
\end{align}
implying
\begin{align}
    c_{I,2,0} = \frac{-64\pi r_2^7 B_2 (\sum_{\bp} G_p^4e^{-G_p^2r_2^2})}{45\Omega}.
\end{align}
Finally, for the term including the scalar $\frac{32r_2^7B_2}{15}$ in \cref{eq:f_Ipq}, corresponding to the index $\sigma=(2,(\omega,\omega'))$ where $\omega\neq\omega' \in \{x,y,z\}$, we first define
\begin{align}\label{eq:Psi_I,2,w}
    \ket{\Psi_{I,2,(\omega,\omega')}} &:= \frac{1}{\sqrt{\sum_{\bp} (G_{p,\omega}G_{p,\omega'})^2e^{-G_p^2r_2^2}}}\sum_{\bp}  (G_{p,\omega}G_{p,\omega'})e^{-G_p^2r_2^2/2}\ket{\bp} 
\end{align}
and derive the following decomposition
\begin{align}
    U_{NL,2,(\omega,\omega')} &:= \frac{4\pi }{\Omega} \sum_{j=1}^\eta \sum_{I=1}^L\left(\sum_{\bp} (G_{p,\omega}G_{p,\omega'})^2e^{-G_p^2r_2^2}\right) \frac{32r_2^7B_2}{15}  \ket{\Psi_{I,2,(\omega,\omega')}}\bra{\Psi_{I,2,(\omega,\omega')}}  \nonumber \\ 
    &= \sum_{I=1}^L\frac{64\pi r_2^7B_2 \eta   \left(\sum_{\bp} (G_{p,\omega}G_{p,\omega'})^2e^{-G_p^2r_2^2}\right)}{15\Omega  }\mathbbm{1}+ \nonumber \\
    &\sum_{j=1}^\eta \sum_{I=1}^L \frac{-64\pi r_2^7B_2\left(\sum_{\bp} (G_{p,\omega}G_{p,\omega'})^2e^{-G_p^2r_2^2}\right)}{15\Omega  }R(\bR_I)^\dagger(\mathbbm{1}-2\ket{\Psi_{I,2,(\omega,\omega')}}\bra{\Psi_{I,2,(\omega,\omega')}})R(\bR_I),
\end{align}
implying
\begin{align}
    c_{I,2,(\omega,\omega')} = \frac{-64\pi r_2^7 B_2\left(\sum_{\bp} (G_{p,\omega}G_{p,\omega'})^2e^{-G_p^2r_2^2}\right)}{15\Omega }.
\end{align}
This finishes the LCU for $U_{NL}$:
\begin{align}
U_{NL} = U_{NL,0}+\sum_{\omega} U_{NL,1,\omega} + U_{NL,2,0} + \sum_{\omega\neq \omega'} U_{NL,2,(\omega,\omega')}.
\end{align}
In each of the above decompositions, there is a multiple of the identity, which are not considered in the qubitization as we can shift the Hamiltonian by the appropriate scalar. Throughout the text, we use the alternative indexing $0\le \sigma \le 10$ to index the operators above in the order they were derived.

\section{QROM: application, parallelization and costs}\label{app:QROM_all_apps}
This section of the appendix provides more details on how quantum read-only memories (QROMs) can be used to prepare arbitrary superposition states. It also explains the cost of implementing a QROM and the space-time tradeoffs that arise. 

\subsection{Using QROM to prepare superpositions}\label{app:QROM_superposition}
As mentioned in \cref{ssec:background_qrom}, there are three different types of QROMs that one can use. Two of these, called \textsc{Select} and \textsc{SelSwapDirty}, are of interest to us. \cref{fig:qroms_circuit} provides an overview of their circuit implementation. We now prove the error estimate in \cref{lem:qrom_superposition_error}. We use the same notations as in \cref{ssec:background_qrom}, which we briefly recall here. The target state is $\ket{\psi}=\sum_x a_x \ket{x}$, where we assume $a_x\geq 0$ for simplicity. For any bit-string $y$ of length $w\le n$, define $p_y = \sum_{\text{prefix}_w(x)=y} |a_x|^2$, $\cos(\theta_y)= \sqrt{p_{y0}/p_y}$, and the QROM oracles $O_w \ket{y}\ket{\bm{0}} = \ket{y}\ket{\theta_y}$, outputting $\theta_y$ up to $b$ bits of precision for each $w$. The precise form of \cref{lem:qrom_superposition_error} assumes an exact preparation for the synthesis of rotation $R$.
\begin{figure}[!h]
	\centering
 \resizebox{\textwidth}{!}{
	\begin{tabular}[t]{ll}
		\begin{tabular}[t]{lc}
			a)&			\includegraphics[trim={0 0 0 0},clip]{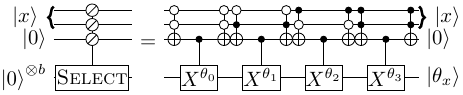}\vspace{0.2cm}
		\end{tabular}&
		\begin{tabular}[t]{lc}
			b)&		\includegraphics[trim={0 0 0 0},clip]{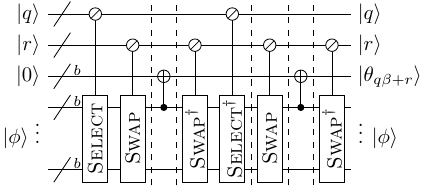}
		\end{tabular}
	\end{tabular}}
	\caption{(\cite[Fig. 1]{low2018trading}) (a) The \textsc{Select} $\sum_{x=0}^{N-1}\ket{x}\bra{x}\otimes X^{\theta_x}$ with $N=4$. The symbol $\oslash$ indicates control by a number state. This QROM variant is the most simple and useful for reading a small number of qubits, and it has no garbage. Its cost becomes prohibitive when reading a large number of qubits $\log(N)$. (b) The \textsc{SelSwapDirty} QROM is a cost-friendly alternative when reading a large number of qubits. It uses an additional $\lceil\log N\rceil+b$ clean qubits and $b\beta$ dirty qubits to implement the data-lookup oracle without garbage with $\beta \in [1,N]$ being the programmer's choice to trade-off gate-qubit cost. See \cite{low2018trading} for detailed implementation.}
	\label{fig:qroms_circuit}
\end{figure}
\begin{lem}\label{lem:qrom_superposition_error}
Assuming an exact rotation synthesis for $R$ and real positive amplitudes $a_x \ge 0$ for the target state, the error of the QROM state preparation method in \qromalgo~is $2^{-b}\pi n$.
\end{lem}
\begin{proof}
We follow the same inductive process of the algorithm to estimate the final error. Assuming an error state $\ket{\bm{\error}_w}$ with norm $\error_w$ for the inductive step up to $w \in  \{1,\ldots, n\}$, and an error state $\ket{\bm{\error}_{b,y}}$ with norm $\error_{b,y}$ for each rotation, the next step gives
\begin{align}
&\left(\sum_{y\in\{0,1\}^w}\sqrt{p_y}\ket{y}+\ket{\bm{\error}_w}\right)\ket{0}\ket{\bm{0}}  \underset{O_w}{\mapsto} \sum_{y\in\{0,1\}^w}\sqrt{p_y}\ket{y}\ket{0}\ket{\theta_y}+O_w\ket{\bm{\error}_w}\ket{0}\ket{\bm{0}} \underset{R_w}{\mapsto} \\  
&\sum_{y\in\{0,1\}^w}\sqrt{p_y}\ket{y}\left(\sqrt{\frac{p_{y0}}{p_y}}\ket{0}+\sqrt{\frac{p_{y1}}{p_y}}\ket{1} + \ket{\bm{\error}_{b,y}}\right)\ket{\theta_y}+R_wO_w\ket{\bm{\error}_w}\ket{0}\ket{\bm{0}} \underset{O_w^\dagger}{\mapsto} \\
&\left(\sum_{y\in\{0,1\}^{w+1}}\sqrt{p_y}\ket{y}\ket{\bm{0}} + \sum_{y\in\{0,1\}^{w}}\sqrt{p_y}\ket{y}\ket{\bm{\error}_{b,y}}\ket{\bm{0}}\right)  + O_w^\dagger R_wO_w\ket{\bm{\error}_w} \ket{0}\ket{\bm{0}}
\end{align}
We observe $\ket{\bm{\error}_{w+1}} = \sum_{y\in\{0,1\}^{w}}\sqrt{p_y}\ket{y}\ket{\bm{\error}_{b,y}}\ket{\bm{0}}  + O_w^\dagger R_wO_w\ket{\bm{\error}_w}$ with norm $\error_{w+1} \le (\sum p_y \error_{b,y}^2)^{1/2} +  \error_w \le  \error_b  + \error_w$, where $\error_b = \max\limits_{1\le w \le n,y \in \{0,1\}^w} \error_{b,y}$. Therefore, we have $\error_{n} \le n \error_b$. The statement follows as $\error_b \le 2^{-b}\pi$ (\cite[Eq. (55)]{su2021fault}).
\end{proof}

\begin{rmk}\label{rmk:rot_syn_cost}
To synthesize a rotation, one needs a so-called gradient state. Given access to a gradient state, the synthesis of an exact rotation has cost $b-3$ Toffolis \cite[Eq. (55)]{su2021fault}. We note that the gradient state used in all single-qubit rotations synthesis in our algorithm is precomputed. Even if we were to consider the error, the cost added is polylogarithmic and can be safely ignored in our resource estimations.
\end{rmk}
\begin{rmk}\label{rmk:last_iter_qrom_phase}
Assuming the amplitudes are not positive, there is one last iteration of the algorithm which we did not include in the previous estimate. If the phase $\phi_x = \text{arg}[a_x/|a_x|]$ is nontrivial, then a QROM and its inverse reading $n$ qubits along with a rotation are needed. Thus, in general, the error of this last rotation needs to be added to the estimate above. However, we explain below why there is no error for our specific cases.\\

Our applications are either part of the PREP subroutine or part of $\SEL_{NL}$. For $\SEL_{NL}$, the preparation of the Gaussian states $\ket{\Psi_{I,\sigma}}$ of certain types $\sigma$ has amplitudes with a phase that is $\pm 1$. In those cases, the rotation is a $Z$ gate which has no error. For the PREP subroutines, such as those in \cref{eq:prep_loc_state,eq:prep_nl_state} for the momentum state of $U_{loc}$ and the PREP state corresponding to $U_{NL}$, we notice a register $s_{loc},s_{NL}$ that holds a sign which is later used in the corresponding SEL subroutine. The cost of computing this is one last QROM that reads $n$ qubits to compute the sign. Again, there is no error in this last step.
\end{rmk}

\subsection{QROM gate and qubit costings}\label{app:QROM_gate_qubit_costings}
In this section, we list the costs that are relevant for performing resource estimation of various parts of the algorithm that use variants of QROM. \cref{table:qroms_cost} further below, which is a copy of \cref{table:qroms_cost_main_text} in the main text \ is used to derive the estimates, along with a careful examination of the cost of each step of \qromalgo. In some cases, we are simply recomputing some of the estimates in \cite[App. D.b]{low2018trading} but with more precision. First, we recall the Toffoli cost of a single QROM for our variants of interest:
\begin{align}\label{eq:qrom_gate_cost_output_sel_variant}
    \textsc{Select} \ &: \  2^{n} \\
    \label{eq:qrom_gate_cost_output_1d}
    \textsc{SelSwapDirty} \ &: \  3b \beta +2\left\lceil \frac{2^{\log(N)}}{\beta}\right\rceil.
\end{align}
Using the above for \textsc{SelSwapDirty} and recalling \cref{rmk:rot_syn_cost}, we estimate the Toffoli cost of \qromalgo~for $N=2^n$. It is given by
\begin{align}
    2\left(3b \beta n + 2\sum_{w=0}^{n}\left\lceil\frac{2^w}{\beta}\right\rceil\right) + (b-3)n,
\end{align}
which is bounded by
\begin{align}\label{eq:qrom_gate_cost_superposition_1d}
     2\left(3b \beta n  + 2\left\lceil\frac{2^{n +1}-1}{\beta}\right\rceil+2n\right) + (b-3)n.
\end{align} 
The leading factor of two is due to the application of $O_w$ and its inverse in the iterative process. Further, the sum over $w$ goes from $0$ to $n$, instead of $0$ to $n-1$, since in the last iteration we need to also output the phase $\phi_x = \text{arg}[a_x/|a_x|]$, see~\cref{rmk:last_iter_qrom_phase}. 

To find the optimal value of $\beta$, one has to also take into account the maximum number of available dirty qubits $n_{\text{dirty}}$. This number is set equal to the total number of qubits on the circuit, minus the ones already used as clean qubits in the QROM itself. As the QROM consumes $\beta b$ many dirty qubits, we must have $\beta b \le n_{\text{dirty}}$. Since the leading term in the cost is always $2(2^{n +1}-1)$, one can show that the optimal value of $\beta$ satisfying its constraint is
\begin{align}\label{eq:optimal_beta_cost}
    \beta = \left\lfloor \min\left(\sqrt{\frac{2(2^{n+1}-1)}{3bn}},\frac{n_{\text{dirty}}}{b}\right)\right\rfloor.
\end{align}
Lastly, the cost for preparing a state using the \textsc{Select} variant of QROM is
\begin{align}\label{eq:qrom_gate_cost_superposition_sel_variant}
    2(2^{n +1}-1) + (b-3) n.
\end{align}
The qubit costings are already described in \cref{table:qroms_cost}. Note the clean qubit cost $b+\lceil \log(N) \rceil $ is independent of $\beta$.
\begin{rmk}\label{rmk:why_select_qrom}
    The variant of choice for QROM when applied on very few qubits is \textsc{Select}. Its implementation, as shown in \cref{fig:qroms_circuit}, along with its cost analysis, is far simpler than that of \textsc{SelSwapDirty}.
\end{rmk}
{
	\renewcommand{\arraystretch}{1.1}
	\begin{table}[ht]
 \centering
		\begin{tabular}{|Sc|Sc|Sc|Sc|}
  \hline
			\shortstack{Operation\\$\;$}&\shortstack{Additional qubits\\$\;$}&\shortstack{Toffoli Depth\\$\;$}&\shortstack{Toffoli count\\ $\le\cdot+\mathcal{O}(\log\cdot)$}\\
			\hline
			\textsc{Select}&$b+\lceil\log{N}\rceil$&$N$&$N$
			\\\hline
			{\textsc{SelSwapDirty}}&$b(\beta+1) +\lceil\log{N}\rceil$&$2\lceil\frac{N}{\beta}\rceil+3\lceil\log{\beta\rceil}$&$2\lceil\frac{N}{\beta}\rceil+3b\beta$\\\hline
		\end{tabular}
		\caption{(\cite[Table II]{low2018trading}) Gate and qubit cost of \textsc{Select} and \textsc{SelSwapDirty} QROMs. The space-depth trade-off is determined by $\beta\in[1,N]$. Note that $b\beta$ qubits of the circuit shown in~\cref{fig:qroms_circuit}b are dirty qubits, while $b+\lceil \log(N) \rceil $ are clean qubits.}
		\label{table:qroms_cost}
	\end{table}
}

\subsection{QROM Toffoli depth}\label{app:qrom_parallelization}
\textsc{SelSwapDirty} offers considerable flexibility in controlling the Toffoli depth. As shown in \cref{fig:qroms_circuit}, the implementation of the operator SWAP in \textsc{SelSwapDirty} involves $\beta$-controlled swap operations on $\beta+1$ registers, each of size $b$. Without any parallelization, the depth would be $\beta b $. However, such an operation can be extensively parallelized on a circuit to a depth of $b\log(\beta)$ as explained in \cite[App. B.2.b]{low2018trading}, using simply controlled swap operations applied on two registers of size $b$. The parameters involved in this depth reduction and considered for the purpose of resource estimation, are 
\begin{itemize}
    \item $\kappa$: This factor is between $1$ and $b$ and determines the extent to which we further parallelize the circuit by implementing $\kappa$ of the $b$ many controlled swap operations simultaneously. This would bring down the depth to $(b/\kappa)\log\beta$. The assumption in \cite{low2018trading} is that $\kappa = b$, meaning all controlled swap operations are applied in parallel.
    \item $n_{\text{tof}}$: The maximum allowed number of simultaneous Toffoli application.
    \item $n_{\text{dirty}}$: Already defined in \cref{app:QROM_gate_qubit_costings}.
\end{itemize}
It can be shown that the depth of \textsc{SelSwapDirty} on $n$ qubits is:
\begin{align}\label{eq:qrom_gate_cost_output_1d_parallelized_v2}
    3\left\lceil \frac{b}{\kappa} \right\rceil \lceil \log\beta\rceil  +2\left\lceil \frac{2^{n}}{\beta}\right\rceil
\end{align}
As mentioned, the choice of $\kappa = b$ gives \cref{table:qroms_cost}, but we have the following constraints on the number of simultaneous Toffoli applications and dirty qubits: 
\begin{align}\label{eq:beta_nparallel}
&\beta \kappa \le n_{\text{tof}} \\\label{eq:beta_ndirty}
&\beta b \le n_{\text{dirty}}.
\end{align}
Finally, the depth of \qromalgo~is
\begin{align}\label{eq:qrom_gate_cost_superposition_1d_parallelized_v2}
        2\left(3\left\lceil \frac{b}{\kappa}\right\rceil \left\lceil\log\beta\right\rceil n+2\left\lceil\frac{2^{n +1}-1}{\beta}\right\rceil+2n\right) + (b-3) n.
\end{align}
Given the range of parameters in our case studies, we always have $b,n \ll 2^{n+1}-1$. Thus, the optimal value of $\beta$ satisfying the constraints is
\begin{align}\label{eq:optimal_beta_depth}
    \beta = \left\lfloor \min\left(\frac{n_{\text{dirty}}}{b}, \frac{n_{\text{tof}}}{\kappa}, \frac{2\cdot (2^{n+1}-1)}{3bn/\kappa} \log_e 2\right)\right\rfloor,
\end{align}
where $\log_e$ is the logarithm in the natural basis. Note that if $n_{\text{dirty}},n_{\text{tof}}$ are large enough, for example $\sim O(2^{n+2})$, then we could set $\kappa = b$ and the minimum possible depth would be achieved with $\beta = \lfloor \frac{2\cdot (2^{n+1}-1)}{3bn/\kappa} \log_e 2\rfloor $:
\begin{align}\label{eq:min_possible_depth}
    2\Big(3 \big(n+2-\log(3n)+\log(\log_e 2 )\big)n+\frac{3n}{\log_e 2}+2n\Big) + (b-3) n \ \sim \ O(n^2+(b-3)n).
\end{align}

\section{Inequality test} \label{app:Ineq_test}
In this section, we review and generalize to general lattices the inequality test technique used for the preparation of the momentum state superposition for the all-electron Hamiltonian. Recall that this state is shared by the PREP state of $U$ and $V$. The task is to prepare:
\begin{align}\label{eq:prep_U_V_common_momentum_state}
    \frac{1}{\sqrt{\lambda_{\nu}}}\sum_{\nu \in \mathcal{G}_0} \frac{1}{\|\bm{\nu}\|}\ket{\nu_x}\ket{\nu_y}\ket{\nu_z},
\end{align}
where $\lambda_\nu$ is a normalization factor. There are two main steps, with the first independent from the lattice structure.
\begin{enumerate}
    \item Prepare a unary-encoded register on $\mu=2$ to $n_p+1$, i.e.,
    $\frac{1}{\sqrt{2^{n_p+2}}}\sum_{\mu= 2}^{n_p+1}\sqrt{2^{\mu}}\ket{\mu}$. Then using controlled Hadamards over registers $\ket{\nu_x}$, $\ket{\nu_y}$, and $\ket{\nu_z}$, prepare a uniform superposition taking values from $-2^{\mu-1}+1$ to $2^{\mu-1}-1$ as signed integers. These superpositions are over a series of nested cubes $\mathcal{C}_\mu = \{\bnu :  \|\bnu\|_\infty < 2^{\mu-1} \}$, whose differences are denoted by $\mcB_\mu = \mathcal{C}_\mu \backslash \mathcal{C}_{\mu - 1}$.   
    In the previous preparation, $\ket{+0}$ and $\ket{-0}$ both appear, and the latter is flagged as failure. Furthermore, to avoid double-counting, for a $\bm{\nu}$ prepared for $\mu$, if $\bm{\nu} \notin \mcB_\mu$, it is flagged as failure.     
    Finally, we prepare a uniform superposition over a register $\ket{m}$ of size $n_M$, where $n_M$ is to be determined later by error analysis. Overall, we obtain the following state: 
    \begin{equation}
    \frac{1}{\sqrt{M2^{n_p+2}}}\sum_{\mu=2}^{n_p+1}\sum_{\bm{\nu}\in \mcB_{\mu}}\sum_{m=0}^{M-1} \frac{1}{2^\mu}\ket{\mu}\ket{\nu_x}\ket{\nu_y}\ket{\nu_z}\ket{m}\ket{00}_{\text{flag}} + \ket{\Phi^{\perp}},
    \end{equation}
    where $M=2^{n_M}$ and $\ket{\Phi^\perp}$ includes the basis states flagged as failure by any of the two qubits in register $flag$. See \cite[p. 4-5]{babbush2019quantum} for more details.
    
    \item For every $\ket{\mu}\ket{\bnu}\ket{m}$, we check whether:
    \begin{equation}\label{eq:ineq_test_m}
         m G_\nu^2 < (2^{\mu-2}b_{min})^2 M,
    \end{equation}
    where $b_{min} = \sigma_3(B)$ is the smallest singular value of the lattice matrix $B := (\bb_1,\bb_2,\bb_3)$ with $\bb_\omega$ the reciprocal lattice column vector. Notice $b_{min}$ coincides with $\min_\omega \|\bb_\omega\|$ for an orthogonal lattice. One must show the soundness of the inequality, in other words, as $m<M$, we need to show $\frac{M(2^{\mu-2}b_{min})^2}{G_\nu^2 } \le M$. Indeed, for $\bnu \in \mcB_\mu = \{\bnu : 2^{\mu-2} \le \|\bnu\|_\infty \le 2^{\mu-1}-1 \}$, we have $G_\nu^2 = \langle \bnu, B^TB \bnu \rangle \ge \sigma_3(B)^2\|\bnu\|^2 \ge \sigma_3(B)^22^{2(\mu-2)}$. To implement the inequality test, we opt to use a \textsc{SelSwapDirty} QROM circuit $O$ to compute:
    \begin{align}
        O\ket{\bnu} \ket{0} = \ket{\bnu}\ket{a_{\bnu}}, 
    \end{align}
    where $a_{\bnu} = \lceil  M(2^{\mu-2}b_{min}/G_\nu)^2 \rceil$ is calculated classically. This is followed by an $n_M$-bit comparator circuit to compare $a_{\bnu}$ to $m$ and store the success as $\ket{0}$ in the third flag qubit below:
    \begin{equation}\label{eq:PREP_U+V_result_state}
      \frac{1}{\sqrt{M2^{n_p+2}}}\sum_{\mu=2}^{n_p+1}\sum_{\bm{\nu}\in \mcB_{\mu}}\sum_{m=0}^{\lceil M(2^{\mu-2}b_{min}/G_\nu)^2\rceil-1} \frac{1}{2^\mu}\ket{\mu}\ket{\nu_x}\ket{\nu_y}\ket{\nu_z}\ket{m}\ket{000}_{\text{flag}} + \ket{\Psi^\perp}.
    \end{equation}
    Upon success, the amplitude for each $\bm{\nu}$ is as desired up to a uniform scale:
    \begin{equation}\label{eq:the_amps_ineq_test}
        \sqrt{\frac{\lceil M(2^{\mu-2}b_{min}/G_\nu)^2\rceil}{M2^{2\mu}2^{n_p+2}}}\approx \frac{b_{min}}{4\sqrt{2^{n_p+2}}}\frac{1}{G_\nu}, \ \text{as }M  \to \infty.
    \end{equation}
\end{enumerate}

Amplitude amplification for $\ket{000}_{\text{flag}}$ is the last step, amplifying its probability above some predetermined $p_{\text{th}}$. This finishes the preparation of \cref{eq:prep_U_V_common_momentum_state}.

\section{Exact amplitude amplification with known initial amplitude} \label{app:aa_fewer_steps}
In our implementation of $\SEL_{NL}$, we need an improved amplitude amplification that yields a success probability equal to one for preparing the state $\ket{\Psi_{I,2,0}}$. As mentioned in \cite[App. A]{mcardle2022quantum}, this technique is folklore knowledge but no reference could be found describing the costings of the method. This section gives an overview of the algorithm and discusses how to compute its cost.

\subsection{The algorithm}
Assume we want to amplify the probability of $\ket{0}$ in $\ket{\Phi} = p^{1/2}\ket{0}+(1-p)^{1/2}\ket{1}$ to one (or in practice very close to one), where $p$ is known. Let $l$ be the smallest nonnegative integer satisfying $p\ge \sin(\frac{\pi}{2(2l+1)})^2$. Choose $0<a\le 1$ such that $a^2p=q$ where $q = \sin(\frac{\pi}{2(2l+1)})^2$. Prepare a rotated ancilla $\ket{r} = a\ket{0}+(1-a^2)^{1/2}\ket{1}$ with sufficient precision, the costing of which is determined later. Then apply a circuit with Toffoli with $X$ gates to $\ket{\Phi}\ket{r}\ket{0}$ to make the state:
\begin{align}\label{eq:lin_trick_eq}
    ap^{1/2}\ket{00}\ket{0} + (1-a^2p)^{1/2}\ket{\psi}\ket{1}
\end{align}
where $\ket{\psi}$ is some unit state in the subspace orthogonal to $\ket{00}$. Then $l$ steps of amplitude amplification with the last qubit as the flag gives a success probability of $\sin^2((2l+1)\arcsin(q^{1/2})) = 1$.

\subsection{Cost and error}
The Toffoli cost in this case, assuming a priori access to a gradient state, is the one Toffoli used to prepare \cref{eq:lin_trick_eq}, plus $n_{AA}-3$ Toffolis to make the rotated ancilla with $n_{AA}$ bits of precision. The qubit cost is the two additional ones used in \cref{eq:lin_trick_eq} plus the $n_{AA}$ qubits to be used in the gradient state to achieve required precision. In our resource estimations, we choose a default value of $n_{AA}=35$, which is more than enough to ensure high precision in our cases. Notice this cost is to be added to the cost of preparing $\Phi$, and then multiplied by $(2l+1)$ to give the full amplitude amplification cost. Overall, requiring a lot more bits of precision is very cheap, given that the costs are many orders of magnitude less than the other subroutines involved in the preparation of $\ket{\Psi_{I,2,0}}$. 

Lastly, no error on the ground state energy estimation is induced from this step; the effective value of $\lambda$ changes very slightly, through the scaling of $\lambda_{NL}$ (which is by far the lowest contributor to $\lambda$) by a factor of about $(1-2^{-2n_{AA}})^{-1}$. Due to its negligible impact in our resource estimations, we ignore the exact calculation of this change but note that it can be easily included.

\section{The details of Prepare and Select operators}\label{app:the_details_of_prep_and_sel}
This section describes the implementation of prepare and select operators in more depth. We make use of the following notation:
\begin{itemize}
    \item $\mcT $ : the number of atomic species in the cell, and $\tau = \lceil \log(\mcT)\rceil$,
    \item $N_t$ : the number of nuclei with atomic type $t$ in the cell, and $n_t = \lceil \log(N_t) \rceil$.
\end{itemize}
Note that $\sum_{t=0}^{\mcT-1} N_t = L$ by definition, where $L$ is the number of atoms. Expressing the nuclei $I=(t,t_{j})$ in register $k'_{loc}$ or $k'_{NL}$ costs $\tau+ \max_t n_t$ qubits. 
\subsection{Prepare}\label{appssec:prep_implementation}
We first address the registers that have a known or straightforward implementation. 
\begin{enumerate}
\item In \cref{eq:master_prep_state}, for states of the registers $b,g,h,c,d,e$, the preparation is the same as in the OAE setting \cite{su2021fault}. 
\item For the states of the registers $\mcX$ (\ref{eq:prep_mcX}) and $f$ (\ref{eq:prep_b_omega_omega_prime}), as motivated by \cref{rmk:why_select_qrom}, we choose the \textsc{Select} variant of QROM in \qromalgo. The small subtlety in the case of $f$ is that \qromalgo~is used to prepare the superposition over $\ket{\omega,\omega'}_f$, and the fifth qubit $\ket{\sgn(\langle \bb_\omega,\bb_{\omega'}\rangle )}_f$ is computed directly by reading $\ket{\omega,\omega'}_f$ using a \textsc{Select} QROM. \item Finally, for the superposition in \cref{eq:prep_v_state} on $j_{V},k_{V}$, the exact same algorithm in \cref{app:Ineq_test} applies, where one only needs to change the notation from $M \to M_V$. For example, the inequality test is $m G_\nu^2 <  M_{V}(2^{\mu-2}b_{min})^2$.
\end{enumerate}
Next, we discuss the more involved superposition for the local and non-local term.

\subsubsection{Momentum state superposition for $U_{loc}$}\label{appsssec:prep_implementation_loc_state}

We rewrite the superposition in \cref{eq:prep_loc_state} on $k_{loc},k'_{loc},s_{loc}$:
\begin{align}\label{eq:prep_loc_state_appendix}
\frac{1}{\sqrt{\lambda_{\nu,loc}}}\sum_{I,\bm{\nu}\in \mathcal{G}_0}\frac{|\gamma_I(G_\nu)|^{1/2}}{G_\nu}\ket{\bm{\nu}}_{k_{loc}}\ket{I}_{k'_{loc}}\ket{\sgn(\gamma_I(G_\nu))}_{s_{loc}}
\end{align}
The preparation of the above state is done almost entirely using QROM. Recall that the function $\gamma_I$ is solely determined by the atomic type $t$ of $I=(t,t_j)$, where $0\le t \le \mcT-1$, and $0\le t_j\le N_t-1$ enumerates the type $t$ atoms in the cell. The state $\ket{t}_{k'_{loc}}$ is represented with a binary string of length $\tau = \lceil\log(\mcT) \rceil $, and the state $\ket{t_j}_{k'_{loc}}$ with one of length $n_t= \lceil\log(N_t) \rceil$. We prepare the state below over $\ket{\bm{\nu}}_{k_{loc}}\ket{t}_{k'_{loc}}$ using \qromalgo~with \textsc{SelSwapDirty} QROMs, followed by a final \textsc{SelSwapDirty} QROM that reads $\ket{\bm{\nu}}_{k_{loc}}\ket{t}_{k'_{loc}}$ and outputs $\ket{\sgn(\gamma_t(G_\nu))}_{s_{loc}}$ (\cref{rmk:last_iter_qrom_phase}):
\begin{align}\label{eq:prep_loc_n_t_ps}
\sum_{t,\bm{\nu}\in \mathcal{G}_0}\left(\frac{n_t}{\Ps(n_t,b_r)}\right)^{1/2}\frac{|\gamma_t(G_\nu)|^{1/2}}{G_\nu}\ket{\bm{\nu}}_{k_{loc}}\ket{t}_{k'_{loc}}\ket{\sgn(\gamma_t(G_\nu))}_{s_{loc}}.
\end{align}
Here, $\Ps(n,b)$ is the probability of success for the preparation of a uniform superposition over $n$ basis states using $b$ bits of precision for the involved rotation (\cite[App. J]{su2021fault}). In all of our implementations, we set $b_r=8$ which gives a very high success probability for all different values of $n$ in our simulations. 

The state we get from using \qromalgo~is not exactly that of \cref{eq:prep_loc_n_t_ps} and, similarly to the preparation of \cref{eq:prep_U_V_common_momentum_state}, it may include inadmissible states, such as $\ket{\bnu} = \ket{\pm 0}$ or $\ket{t}$ for a $t\ge \mcT$. Notice that for states prepared using QROM, the corresponding action of $\SEL$ on inadmissible basis states does not have to be a trivial action with an exactly computable scalar, as we are not attempting to shift the Hamiltonian by some known scalar. We take into account these inadmissible basis states by the error they induce throughout the block-encoding. Thus it is not necessary to flag these basis states, unless it is to make the action of $\SEL$ well defined. However, as shown by the definition of $\SEL$, the action there is already well defined: these basis states can only be an issue when computing $\bR_I$ or the phase action $e^{i\bG_\nu \cdot \bR_I}$. When $\bnu=\ket{\pm 0}$, $\bG_\nu$ is the origin vector, and when $t\ge \mcT$ leading to some undefined $\bR_I$, the QROM computing $\bR_I$ is programmed to output the all-zero state. Therefore, the action is not only well-defined but also trivial.

The last step is a direct application of the technique in \cite[App. A.2]{lee2021even} to prepare a uniform superposition over the $n_t$ nuclei of type $t$, meaning
\begin{align}\label{eq:I_loc_first_appear}
&\sum_{t,\bm{\nu}\in \mathcal{G}_0}\left(\frac{n_t}{\Ps(n_t,b_r)}\right)^{1/2}\frac{|\gamma_t(G_\nu)|^{1/2}}{G_\nu}\ket{\bm{\nu}}_{k_{loc}}\ket{t}_{k'_{loc}} \ket{\sgn(\gamma_I(G_\nu))}_{s_{loc}} \otimes \nonumber \\ &\left(\left(\frac{\Ps(n_t,b_r)}{n_t}\right)^{1/2}\ket{0}_{I,loc}\sum_{t_j=0}^{n_t-1}\ket{t_j}_{k'_{loc}} + (1-\Ps(n_t,b_r))^{1/2}\ket{1}_{I,loc}\ket{\psi_t^\perp}_{k'_{loc}}\right) ,
\end{align}
where $\ket{\psi_t^\perp}_{k'_{loc}}$ is some inadmissible state and $\ket{\cdot}_{I,loc}$ flags the desired uniform superposition. Simplifying the expression above for $\ket{0}_{I,loc}$ yields the desired state of \cref{eq:prep_loc_state_appendix}.

\subsubsection{PREP state for $U_{NL}$}\label{appsssec:prep_state_nl}

We wish to prepare the state of registers $k_{NL},k'_{NL},s_{NL}$ in \cref{eq:prep_nl_state}, rewritten below:
\begin{align}\label{eq:prep_nl_state_appendix}
    \frac{1}{\sqrt{\sum |c_{I,\sigma}|}} \sum_{I,\sigma} \sqrt{|c_{I,\sigma}|}\ket{\sigma}_{k_{NL}}\ket{I}_{k'_{NL}}\ket{\sgn(c_{I,\sigma})}_{s_{NL}}.
\end{align}
The process is very similar to $U_{loc}$. First, the state 
\begin{align}\label{eq:qrom_nl_state_a_o}
    \sum_{t=0}^{\mcT-1}\sum_{\sigma=0}^{10}\left(\frac{n_t|c_{t,\sigma}|}{\Ps(n_t,b_r)}\right)^{1/2} \ket{\sigma}_{k_{NL}}\ket{t}_{k'_{NL}},
\end{align}
is prepared using a \textsc{SelSwapDirty} QROM reading $\tau+4$ many qubits, where the four qubits encode $0\le \sigma \le 10$. This is followed by a uniform superposition over $n_t$ basis states with success probability $\Ps(n_t,b_r)$:
\begin{align}
    &\sum_{t,\sigma}\left(\frac{n_t|c_{t,\sigma}|}{\Ps(n_t,b_r)}\right)^{1/2} \ket{\sigma}_{k_{NL}}\ket{t}_{k'_{NL}} \otimes \nonumber \\
    &\left(\left(\frac{\Ps(n_t,b_r)}{n_t}\right)^{1/2}\ket{0}_{I,NL}\sum_{t_j=0}^{n_t-1}\ket{t_j}_{k'_{NL}} + (1-\Ps(n_t,b_r))^{1/2}\ket{1}_{I,NL}\ket{\psi_t^\perp}_{k'_{NL}}\right).
\end{align}
Finally, for the value in register $s_{NL}$, we need to know the atomic type $t$ and which of the three subgroups of the 11 types $\sigma$ refers to, i.e., $\sigma = 0$, or $\sigma=(1,\omega)$, or $\sigma = (2,x)$ where $x$ is either $0$ or $(\omega,\omega')$. This identification requires us to compute two bits of information, which can be done by inequality tests as the three subgroups correspond to the enumeration $\sigma = 0, 1\le \sigma \le 3,$ and $4\le \sigma \le 10$. Thus a \textsc{Select} QROM reading $\tau+2$ qubits can be used to finish the preparation of \cref{eq:prep_nl_state_appendix}.

\subsection{Select}\label{appssec:SEL_implementation}

Below, we go over the remaining details of $\SEL_{T},\SEL_{V},$ and $\SEL_{loc}$, and dedicate a separate section for $\SEL_{NL}$.\\

1. $\SEL_T$: We recall the transformation implemented by this operator
    \begin{align}\begin{split}
        \SEL_T: &\ket{0}_{T} \ket{+}_b\ket{j}_e \ket{\omega,\omega',\sgn(\langle \bb_\omega,\bb_{\omega'}\rangle)}_f \ket{r}_g \ket{s}_h \ket{\bp}_j \to\\
        &(-1)^{b(p_{\omega,r}p_{\omega',s}+1)+\sgn(\langle \bb_\omega,\bb_{\omega'}\rangle p_\omega p_{\omega'})} \ket{0}_{T} \ket{+}_b \ket{j}_e \ket{\omega,\omega',\sgn(\langle \bb_\omega,\bb_{\omega'}\rangle)}_f \ket{r}_g \ket{s}_h \ket{\bp}_j.
    \end{split}
    \end{align}
    In addition to the CSWAP that sends back and forth the $(j,\omega,r),(j,\omega',s)$ coordinates to an auxiliary register, we have a phase that is controlled on the following:
    \begin{itemize}
        \item The state $\ket{\rchi}$ is equal to $\ket{00}$,
        \item The register $c$ flags the success of $i=j$ in registers $d,e$ (i.e., $\ket{1}_c$),
        \item The ancilla attached to register $f$ flags the admissible states ($(\omega \neq 4) \land (\omega' \neq 4)$) in the approximate superposition prepared by QROM. This uses three Toffolis and two additional qubits.
    \end{itemize}
    With the help of three additional Toffolis and a single auxiliary register, we can record the success of all three above in terms of the state $\ket{0}_T$. This would indicate the success of the state preparation of $T$. Therefore the phase action of $\SEL_T$ is only controlled on a single auxiliary state $\ket{\cdot}_T$ when the lattice is orthogonal (\cref{rmk:lcu_t_b_getting_dropped}), and otherwise, it is controlled on an additional register $\ket{\cdot}_b$. Note this additional control is only for implementing $(-1)^{b(p_{\omega,r} p_{\omega',s})}$, where one needs two Toffolis.\\
   
2. $\SEL_V$: We recall the phase and controlled arithmetics carried out by this operator:
     \begin{align}\begin{split}
        \SEL_V:  &\ket{0}_{V} \ket{b}_b \ket{i}_d \ket{j}_e \ket{0}_c \ket{0}_{j_V} \ket{\bnu}_{k_V} \ket{\bp}_i \ket{\bq}_j \to \\
        &(-1)^{b([(\bp+\bnu) \not \in \mcG] \lor [(\bq-\bnu) \not\in \mcG])} \ket{0}_{V} \ket{b}_b \ket{i}_d \ket{j}_e \ket{0}_c \ket{0}_{j_V}\ket{\bnu}_{k_V} \ket{\bp+\bnu}_i \ket{\bq-\bnu}_j.
        \end{split}
    \end{align}
    The addition and subtraction along with the phase implementation must be controlled on the success of state preparation for $V$, which occurs when the state $(\ket{\rchi}_{\mcX}$ is equal to $\ket{01}_{\mcX})$ and the $j_v, c$ registers are in state  $\ket{0}_{j_V} \ket{0}_c$. This can be computed via three Toffolis into one single register $\ket{0}_{V}$. Thus addition and subtraction along with the phase implementation are controlled on a single auxiliary register.\\

3. $\SEL_{loc}$: We recall the overall action:
    \begin{align}\begin{split}
        \SEL_{loc}: &\ket{0}_{loc} \ket{b}_b \ket{j}_e \ket{\bnu}_{k_{loc}} \ket{I}_{k'_{loc}} \ket{\sgn(\gamma_I(G_\nu))}_{s_{loc}} \ket{\bm{0}}_{\bR}   \ket{\bp}_j \to \\
        &(-1)^{b[(\bp-\bnu) \not \in \mcG]} \ket{0}_{loc} \ket{b}_b \ket{j}_e \ket{\bnu}_{k_{loc}} \ket{I}_{k'_{loc}} \ket{\sgn(\gamma_I(G_\nu))}_{s_{loc}} \ket{\bR_I}_{\bR} \ket{\bp-\bnu}_j \to \\
        &e^{i \bG_\nu \cdot \bR_I} (-1)^{b[(\bp-\bnu) \not \in \mcG]+\sgn(\gamma_I(G_\nu))} \ket{0}_{loc} \ket{b}_b  \ket{j}_e \ket{\bnu}_{k_{loc}} \ket{I}_{k'_{loc}} \ket{\sgn(\gamma_I(G_\nu))}_{s_{loc}}  \ket{0}_{\bR} \ket{\bp-\bnu}_j.
    \end{split}
    \end{align}
    The subtraction of $\bnu$ from $\bp$ is controlled on $(\ket{\rchi}_{\mcX}$ being equal to $\ket{10}_{\mcX})$, and registers $c$ and $I,loc$ being in state $\ket{1}_c \ket{0}_{I,loc}$. This is computed via three Toffolis and recorded as $\ket{0}_{loc}$, indicating the success of the state preparation for $U_{loc}$. Recall $\ket{\cdot}_{I,loc}$ was computed as part of the preparation algorithm in \cref{eq:I_loc_first_appear}.
    
    To implement $e^{i \bG_\nu \cdot \bR_I}$, we employ a \textsc{Select} QROM that outputs $\bR_I$ into a register denoted by $\bR$ and of size $n_R$. This QROM reads $\ket{x}_{loc}\ket{I}_{k'_{loc}}\ket{\bm{0}}_{\bR}$ and outputs $\bR_I$ in register $\bR$ if $x=0$ and leaves $\ket{\bm{0}}_{\bR}$ unchanged otherwise. Note that in case $I$'s type is inadmissible (as discussed in \cref{appsssec:prep_implementation_loc_state}), one may simply leave $\bR = \bm{0}$. The implementation of the phase $e^{i \bG_\nu \cdot \bR_I}$ does not need to be controlled as for any case other than the local and non-local operators being qubitized, the register $\bR$ is all-zero. Once the phase is implemented, we erase the register $\bR$ (by applying the inverse of the QROM) and apply $\text{Z}_{s_{loc}}$ controlled on $\ket{0}_{loc}$. \\

4- $\SEL_{NL}$: The costliest and most involved subroutine of $\SEL_{NL}$ is the preparation of $\ket{\Psi_{I,\sigma}}$. For the sake of illustration, we first assume that the lattice of our model is orthogonal, as is the case for \dis~and \limnfo. Recall that the states $\ket{\Psi_{I,\sigma}}$ are all (derivatives of) three-dimensional Gaussian superposition states. Thus, assuming an orthogonal lattice, they decompose to the tensor product of three one-dimensional Gaussian states. Below, we show this decomposition for all the different types $\sigma$ of states $\ket{\Psi_{I,\sigma}}$: $(a) \ket{\Psi_{I,0}}$, $(b) \ket{\Psi_{I,1,\omega}}$, $(c) \ket{\Psi_{I,2,0}}$, $(d) \ket{\Psi_{I,2,(\omega,\omega')}}$ where $1\le \omega\neq\omega' \le 3$. Due to orthogonality, $\bG_p = \sum_\omega p_\omega \bb_\omega = (G_{p,1},G_{p,2},G_{p,3})$ where $G_{p,i} = p_i(\bm{b}_i)_i$. We drop the normalization factors to avoid cluttering:
\begin{align}\label{eq:psi_I_o_decomp}
(a)\ \sum_{\bp}  e^{-G_p^2r_0^2/2}\ket{\bp} =  \bigotimes_{i = 1,2,3}\Big(\sum_{\bp_i}  e^{-(G_{p,i})^2r_0^2/2}\ket{\bp_i}\Big) 
\end{align}
\begin{align}
(b)\ \sum_{\bp}  G_{p,\omega} e^{-G_p^2r_1^2/2}\ket{\bp} =
\bigotimes_{i\neq \omega }  \Big(\sum_{\bp_i}   e^{-(G_{p,i})^2r_1^2/2}\ket{\bp_i}\Big) \otimes \Big(\sum_{\bp_\omega} G_{p,\omega}e^{-(G_{p,\omega})^2r_1^2/2}\ket{\bp_\omega}\Big) 
\end{align}
\begin{align}\label{eq:psi_I_o_decomp_20}
(c)\ \sum_{\bp}  G_p^2 e^{-G_p^2r_2^2/2}\ket{\bp} = \sum_\omega \bigotimes_{i\neq \omega }  \Big(\sum_{\bp_i}   e^{-(G_{p,i})^2r_2^2/2}\ket{\bp_i}\Big) \otimes \Big(\sum_{\bp_\omega} G_{p,\omega}^2e^{-(G_{p,\omega})^2r_2^2/2}\ket{\bp_\omega}\Big)
\end{align}
\begin{align}
(d)\  \sum_{\bp}  (G_{p,\omega}G_{p,\omega'})e^{-G_p^2r_2^2/2}\ket{\bp}  =&  \Big(\sum_{\bp_{\omega''}}   e^{-(G_{p,{\omega''}})^2r_2^2/2}\ket{\bp_{\omega''}}\Big) \otimes  \Big(\sum_{\bp_\omega} G_{p,\omega}e^{-(G_{p,\omega})^2r_2^2/2}\ket{\bp_\omega}\Big)\otimes \nonumber \\
&\Big(\sum_{\bp_{\omega'}} G_{p,{\omega'}}e^{-(G_{p,{\omega'}})^2r_2^2/2}\ket{\bp_{\omega'}}\Big)
\end{align}

The QROM-based state preparation breaks to three separate ones, applied in parallel and each reading $n_p+(\tau+4)$ qubits, significantly reducing the cost. While the orthogonal decomposition works as intended to implement the reflection on states of type (a), (b), (d) above, the decomposition for (c) does not lead to the reflection on $\sum_{\bp} G_p^2 e^{-G_p^2r_2^2/2}\ket{\bp}$ but rather to the sum of that on the three states in its decomposition, meaning $\sum_\omega (\mathbbm{1}-2\ket{\psi_\omega}\bra{\psi_\omega})$ instead of $I-2\ket{\psi}\bra{\psi}$ where $\psi=\sum_\omega \psi_\omega$. For this reason, we separately treat further the states of type (c) below.

\textit{Implementation of the reflection onto $\ket{\Psi_{I,2,0}}$. }The preparation $U_{I,\sigma}$ must act by identity if $\sigma$ is of type (c), and is followed by another operator $V_{I,\sigma}$. This operator specifically takes care of preparing a state of type (c), and is different from the identity only if $\sigma=(2,0)$. $V_{I,\sigma}$ acts as follows:
\begin{enumerate}
    \item The \textsc{Select} QROM-based preparation of the state $\sum_{i=1}^{3} \sqrt{\frac{(\bb_i)_i^4}{\sum_j (\bb_j)_j^4}} \ket{i}$, where $\ket{i}$ is one-hot-encoded in three qubits.
    \item The parallel \textsc{SelSwapDirty} QROM-based preparation of three states, each QROM reading $\ket{\bp_\omega}\ket{t}_{k'_{NL}}\ket{\cdot}_{NL,c}\ket{i}_\omega$, a total of $n_p+\tau+1+1$ qubits. We ensure that the state to be prepared is indeed of type (c) by reading $\ket{\cdot}_{NL,c}$. This register is computed using one Toffoli calculating the AND of $\ket{\cdot}_{NL}$ and $\ket{\sigma==(2,0)}$, where the latter is determined using three Toffolis and three qubits by checking the condition $\sigma==(2,0)$ (recall $0\le \sigma\le 10$ is represented using four qubits).
    \item Following the previous step, we have prepared $\sum \ket{\Psi_{i,I,\sigma}}\ket{i}$ where $\sum \ket{\Psi_{i,I,\sigma}} = \ket{\Psi_{I,\sigma}}$ is the desired state. Therefore, we need to disentangle the $\ket{i}$ register from the system register:
    \begin{itemize}
        \item Apply the Hadamard gate on each three qubits of $\ket{i}$. This leads to the state $\ket{\Psi_{I,\sigma}} \ket{000} + \ldots$. We use two Toffolis to flag success by one flag qubit as $\ket{\Psi_{I,\sigma}}\ket{000}\ket{0}_{\text{flag}} + \ldots$ .
        \item The probability of success for $\ket{0}_{\text{flag}}$ is $\frac{1}{8}$, which is larger than $\sin\left(\frac{\pi}{2(2\cdot 2 +1)}\right)^2\approx 0.095$. Therefore, using the technique in \cref{app:aa_fewer_steps}, we can ensure that after $l=2$ steps of amplitude amplification, the success probability is very close to one, so that the error in our implementation of the reflection is negligible. 
    \end{itemize} 
\end{enumerate}

\textit{The case of non-orthogonal lattices.} We have two materials in \cref{table:materials} with a non-orthogonal lattice. However, as one of the lattice vectors is orthogonal to the other two, the states $\ket{\Psi_{I,\sigma}}$ admit a decomposition into states of size $n_p$ and $2n_p$ qubits. Therefore, the QROM cost will need to change accordingly for all states of type (a), (b), and (d). For type (c), after a suitable change of axis, we always have $\bb_1 = (x_1,0,0), \bb_2 = (0,x_2,0),\bb_3=(0,x_4,x_3)$ for some $x_i>0$. Hence, the preparation of $\ket{\Psi_{I,2,0}}$ can be adapted as follows: prepare a superposition $\sqrt{\frac{x}{x + y}}\ket{01} + \sqrt{\frac{y}{x + y}}\ket{10}$, where $x = (\sum_{p_1} G_{p,1}^4 e^{-G_{p,1}^2r_2^2})(\sum_{p_2,p_3} e^{-(G_{p,2}^2+G_{p,3}^2)r_2^2})$ and $y=(\sum_{p_1} e^{-G_{p,1}^2r_2^2})(\sum_{p_2,p_3} (G_{p,2}^2+G_{p,3}^2)^2e^{-(G_{p,2}^2+G_{p,3}^2)r_2^2})$. The amplitude amplification will need to be done on a qubit with success probability $\frac{1}{4}$. Thus only one amplification is necessary to get exact success probability 1, and there is no need to use the technique in \cref{app:aa_fewer_steps}.

\section{Derivation of $\lambda$}\label{app:lambdas}
In this section, we follow the guideline of \cref{ssec:effective_lambda} for computing the effective values of $\lambda_T$, $\lambda_V$, $\lambda_{NL}$ and $\lambda_{loc}$. This includes finding the relevant success probabilities in state preparation, which flag the admissible states in the superposition. Then one needs to find the approximated amplitudes implemented by the algorithm for the admissible states.

\begin{table*}[ht]
    \centering
    \begin{tabular}{|Sc|Sc|}
        \hline \cref{eq:lambda_T} & $\lambda_T = \eta  2^{2n_p-3} \Ps(\eta,b_r)^{-2}\sum_{\omega,\omega' \in \{1,2,3\}} |\langle \bb_\omega, \bb_{\omega'} \rangle|$ \\ \hline
        \cref{eq:lambda_NL} & $\lambda_{NL} = \eta  \Ps(\eta,b_r)^{-2}\sum_{I,\sigma} |c_{I,\sigma}|$ \\ \hline
        \cref{eq:lambda_V} & $\lambda_V = \frac{2\pi \eta (\eta-1)\lambda_{\nu,V}}{\Omega P_{\nu,V}^{amp} \Ps(\eta,b_r)^{2}}$, where $\lambda_{\nu,V} =  \sum_{\mu = 2}^{n_p+1} \sum_{\bnu \in \mcB_\mu}\frac{\lceil M_V(2^{\mu-2}b_{min}/G_\nu)^2\rceil }{M_V(2^{\mu-2}b_{min})^2}$ \\\hline
        \cref{eq:lambda_loc} & $\lambda_{loc}=\frac{4\pi \eta \lambda_{\nu,loc}}{\Omega\Ps(\eta,b_r)^{2}}$, where  $\lambda_{\nu,loc} = \sum_{I,\bnu \in \mcG_0} \frac{|\gamma_I(G_\nu)|}{G_\nu^2}$ \\
         \hline
         \cref{eq:lambda_val} & $\lambda = \lambda_T+\lambda_{loc}+\lambda_{NL}+\lambda_V$\\
         \hline
    \end{tabular}
    \caption{Values of $\lambda_T$, $\lambda_V$, $\lambda_{NL}$, $\lambda_{loc}$ and total $\lambda$ according to calculations in \cref{app:lambdas}. The corresponding qubit numbers, e.g. $n_T$ are computed in \cref{app:errors}. The probability $P_{\nu,V}^{amp}$ given by \cref{eq:p_nu_v_amp_dfn} is an amplification of the initial probability $P_{\nu,V}$ defined in \cref{eq:p_nu_v}. Each term must include the success probability $\Ps(\eta,b_r)^2$ for preparing the superposition over electron pairs in \cref{eq:prep_b_c_d_e_eta}.}
    \label{tab:lambda_vals}
\end{table*}

\subsection{$\lambda_T$}\label{sssec:lambda_T}
For the kinetic term $T$, there are two success probabilities to consider. One is for preparing the states of registers $g,h$ in \cref{eq:prep_r_s_binary_T}, which brings an adjustment by a probability of $\frac{(2^{n_p-1}-1)^2}{2^{2n_p-2}}$, as demonstrated in \cite[Eq. (71)]{su2021fault}. The other is for preparing \cref{eq:prep_b_c_d_e_eta}, which involves creating two uniform superposition over electrons, yielding the adjustment by $\Ps(\eta,b_r)^2$. So we need to replace the theoretical value $\lambda_T = \eta\sum\limits_{\omega,\omega' \in \{1,2,3\}} |\langle \bb_\omega, \bb_{\omega'} \rangle| (2^{n_p-1}-1)^2/2$ by 
\begin{align}\label{eq:lambda_T}
    \lambda_T = \frac{\eta  2^{2n_p-2}\sum\limits_{\omega,\omega' \in \{1,2,3\}} |\langle \bb_\omega, \bb_{\omega'} \rangle|}{2 \Ps(\eta,b_r)^2}.
\end{align}
If the lattice is orthogonal, $\lambda_T$ is half the above value (\cref{rmk:lcu_t_b_getting_dropped}).

\subsection{$\lambda_{NL}$}\label{sssec:lambda_NL}
Given the implementation in \cref{appsssec:prep_state_nl}, the amplitudes are correctly scaled such that the success probabilities $\Ps(n_t,b_r)$ are canceled out. Thus, only the success probability $\Ps(\eta,b_r)^2$ for the superposition over pairs of electrons must be taken into account. Therefore the theoretical value of $\lambda_{NL} = \eta\sum_{I,\sigma} |c_{I,\sigma}|$ is adjusted as follows:
\begin{align}\label{eq:lambda_NL}
    \lambda_{NL} = \Ps(\eta,b_r)^{-2}\eta  \sum_{I,\sigma} |c_{I,\sigma}|.
\end{align}

\subsection{$\lambda_{V}$}\label{sssec:lambda_V}
We discussed in \cref{appssec:prep_implementation} the implementation of the PREP state for $V$. Here, the derivation of $\lambda_V$ is very similar to the OAE case \cite[Eq. (124)]{su2021fault}. We simply review it by making the small changes needed for general lattices. Let us recall the momentum state superposition:
\begin{align}
    \sqrt{\frac{P_{\nu,V}}{\lambda_{\nu,V}}} \ket{0}_{j_V}\sum_{\bm{\nu}\in \mathcal{G}_0}\frac{1}{G_\nu}\ket{\bm{\nu}}_{k_V}+\sqrt{1-P_{\nu,V}}\ket{1}_{j_V}\ket{\bm{\nu}^\perp}_{k_V}.
\end{align}
Here, we need to consider multiple adjustments to the theoretical value of $\lambda_V$. The first one is the probability of success $P_{\nu,V}$, flagged by $\ket{0}_{j_V}$. We shall amplify it above some set threshold probability $p_{\text{th}}$. Furthermore, similar to $U,V$ in the OAE case \cite[Eq. (124)]{su2021fault}, the amplitudes implemented by the inequality test in \cref{appssec:prep_implementation} are not exactly $\frac{1}{G_\nu}$. Indeed, while theoretically $\lambda_{\nu,V} = \sum_{\nu \in \mcG_0} 1/G_\nu^2$, the effective amplitudes are as mentioned in \cref{eq:the_amps_ineq_test}. Therefore, we have an adjustment for the normalization of the success state flagged by $\ket{0}_{j_V}$:
\begin{align}
    \lambda_{\nu,V} =  \sum_{\mu = 2}^{n_p+1} \sum_{\bnu \in \mcB_\mu}\frac{\lceil M_V(2^{\mu-2}b_{min}/G_\nu)^2\rceil }{M_V(2^{\mu-2}b_{min})^2}.
\end{align}
This is then used to adjust the value of $\lambda_V$, along with the amplified probability and the usual $\Ps(\eta,b_r)^2$ for the electron pairs superposition:
\begin{align}\label{eq:lambda_V}
    \lambda_V = \frac{2\pi \eta (\eta-1)\lambda_{\nu,V}}{\Omega P_{\nu,V}^{amp}\Ps(\eta,b_r)^2}.
\end{align}
To complete our derivation, we recall the expression for the amplified probability $P_{\nu,V}^{amp}$, where
\begin{align}\label{eq:p_nu_v_amp_dfn}
    P_{\nu,V}^{amp} = \sin^2((2a_V+1)\arcsin(\sqrt{P_{\nu,V}})),
\end{align}
given $a_V$ many amplitude amplifications to reach a set success probability threshold $p_{\text{th}}$. Lastly, following the \cref{eq:the_amps_ineq_test}, $P_{\nu,V}$ can be shown to be given by
\begin{align}\label{eq:p_nu_v}
    P_{\nu,V} = \frac{\lambda_{\nu,V}b_{min}^2}{2^{n_p+6}}.
\end{align}

\subsection{$\lambda_{loc}$}\label{sssec:lambda_loc}
We recall the implementation of the momentum state for $U_{loc}$ in \cref{appsssec:prep_implementation_loc_state}, where the QROM scaled the amplitudes by $\Ps(n_t,b_r)^{-1}$, ensuring this success probability gets canceled after preparing the superposition over the nuclei of each atomic type. Thus, similar to $\lambda_{NL}$, we only need to adjust by the usual electron pairs success probability preparation:
\begin{align}\label{eq:lambda_loc}
        \lambda_{loc} = \frac{4\pi \eta \lambda_{\nu,loc}}{\Omega \Ps(\eta,b_r)^2},
\end{align}
where $\lambda_{\nu,loc} = \sum_{I,\bnu }\frac{|\gamma_I(G_\nu)|}{G_\nu^2}$.

\subsection{The effective value of $\lambda$} \label{sssec:lambda}
Gathering the previous results in \cref{eq:lambda_T,eq:lambda_NL,eq:lambda_V,eq:lambda_loc}, the effective value of $\lambda$ is 
\begin{align}\label{eq:lambda_val}
\lambda = \lambda_T+\lambda_{loc}+\lambda_{NL}+\lambda_V.
\end{align}
\setlength\cellspacetoplimit{5pt}
\setlength\cellspacebottomlimit{5pt}
\begin{table*}[ht]
    \centering
    \begin{tabular}{|Sc|Sc|}
    \hline
    \cref{eq:error_rchi_deriv} & $\error_{\rchi} \le \frac{4\pi }{2^{n_{\rchi}}}\lambda $
    \\
    \hline
    \cref{eq:error_B_deriv} &  $\error_B \le 4\frac{\pi \eta  2^{2n_p-2}}{2^{n_B}}\sum_{\omega,\omega'} |\langle\bb_\omega,\bb_{\omega'}\rangle| $ \\
    \hline
    \cref{eq:error_k_deriv} &
    $\error_{NL} \le 2\frac{(\tau+4) \pi  }{2^{n_{NL}}} \eta \sum_{t,\sigma} |\frac{c_{t,\sigma}n_t}{\Ps(n_t,b_r)}| $\\
    \hline 
    \cref{eq:error_M_V_derivation} &
    $\error_{M_V} \le \frac{8\pi\eta(\eta-1)}{\Omega 2^{n_{M_V}} b_{min}^2}(7\times 2^{n_p+1}-9n_p - 11-3\times 2^{-n_p}) $  \\
    \hline
    \cref{eq:error_M_loc_deriv} &
    $\error_{M_{loc}} \le \frac{8 \pi^2 \eta \max_t(\frac{n_t}{\Ps(n_t,b_r)}) (3n_p+\tau) \sum_{t,\bnu} \frac{|\gamma_t(G_\nu)|}{G_\nu^2} }{2^{n_{M_{loc}}}\Omega}$ \\
    \hline
    \cref{eq:error_R_deriv_pessimistic} &
    $\error_R  \le \frac{2\eta\pi \max \|\bm{a}_i\|}{2^{n_R}\Omega} \sum_I (\sum_{\bnu \in \mcG_0}\frac{|\gamma_I(G_\nu)|}{G_\nu} + \sum_{\bnu \in \mcG}G_\nu F_{I,\bnu}) $\\
    \hline
    \cref{eq:error_NL_deriv} &    
    $\error_{\Psi} \le \frac{18(n_p+4+\tau)\pi\eta}{2^{n_{\Psi}}} \sum_{I,\sigma} |c_{I,\sigma}| $\\
    \hline
    \end{tabular}
    \caption{Equations derived in \cref{app:errors} linking the target errors $\error_X$ and their associated finite size registers $n_X$. By replacing the inequality with equality, one obtains the $n_X$ that achieves a target error $\error_X$. In the estimation for $\error_R$, $F_{I,\bnu}$ (\cref{eq:F_NL_I_bnu}) is some expression bounding the entries of the non-local term.}
    \label{tab:error_profile}
\end{table*}
\setlength\cellspacetoplimit{2pt}
\setlength\cellspacebottomlimit{2pt}

\section{Error Analysis}\label{app:errors}
We estimate the errors listed in \cref{sec:overview_of_errors}. To do so, we make the following basic observation. Assume that an LCU of the form $\sum_a \alpha_a U_a$ is approximated by $\sum_a \xi_a V_a$. Here, $\xi_a$ is obtained after a series of approximations due to choosing finite size registers $S_\alpha=(n_{s_1},\ldots,n_{s_k})$ and similarly for $V_a$, where we have a series of approximations using finite size registers $S_U=(m_{s_1},\ldots , m_{s_l})$. We estimate $\|\sum_a \alpha_a U_a - \sum_a \xi_a V_a\|$ using the triangle inequality, by building the following LCU series:
\begin{itemize}
    \item $\sum_a \alpha_a U_{a,j}$ where $0\le j\le l$ means we perform the series of approximations up to $m_{s_t}$. Note that $U_{a,l}=V_a$ and $U_{a,0}:=U_a$.
    \item $\sum_a \alpha_{a,i} V_a$, where $0\le i\le k$ means we use only the finite size registers up to $n_{s_i}$. Note that $\alpha_{a,k} = \xi_a$ and $\alpha_{a,0}=\alpha_a$.
\end{itemize}
Then we can estimate $\sum_a \alpha_a U_{a,j}$ by $\sum_a \alpha_a U_{a,j+1}$ for $0\le j\le l-1$, and $\sum_a \alpha_{a,i} V_a$ by  $\sum_a \alpha_{a,i+1} V_a$ for $0 \le i \le k-1$. Hence, triangle inequality gives us:
\begin{align}\label{eq:series_approx_order}
    \|\sum_a \alpha_a U_a - \sum_a \xi_a V_a\| \le \sum_{j=0}^l \sum_a \alpha_a \|U_{a,j} - U_{a,j+1}\| + \sum_{i=0}^k \sum_a |\alpha_{a,i}-\alpha_{a,i+1}|.
\end{align}
For each of the four operator $\rchi=T,V, U_{NL}, U_{loc}$, we must identify the order in which the approximations must be introduced. For the selection probabilities, in addition to the choice $n_{\rchi}$ which is the first to be made for all of the four operators, there is only one other approximation. For example, for $T$, the order of approximation is $(n_{\rchi},n_{B})$ while for $U_{NL}$, it is $(n_{\rchi},n_{NL})$. There is also at most one choice for the unitaries for all four operators, with the exception of the non-local term; there, the order of approximations is $(n_R,n_{\Psi})$.

We finish this discussion with a lemma that is essential in getting an accurate estimate of the errors made by QROM when scaling a qubitized operator by its $\lambda$.
\begin{lem}\label{lem:qrom_denormalized_bd}
Assume the unit state $\ket{\psi} = \frac{1}{\sqrt{\lambda}} \sum_a \alpha_a \ket{a}$ with $\alpha_a \in \mbbC$, is approximated by the unit state $\widetilde{\ket{\psi}} = \sum_a \xi_a\ket{a}$ up to error:
$$\| \ket{\psi} - \widetilde{\ket{\psi}} \| = \sqrt{\sum_a \left|\xi_a-\frac{\alpha_a}{\sqrt{\lambda}}\right|^2} \le \epsilon.$$ 
Then
\begin{align}
    \sum_a ||\xi_a|^2 \lambda - |\alpha_a|^2 | \le 2\epsilon \lambda.
\end{align}
\end{lem}
\begin{proof}
We use Cauchy-Schwarz and triangle inequality $\sum_a ||\xi_a|^2 \lambda - |\alpha_a|^2 | =$
\begin{align}
    &\sum_a ||\xi_a| \sqrt{\lambda} - |\alpha_a|| \cdot ||\xi_a| \sqrt{\lambda} + |\alpha_a||  \le (\sum_a (|\xi_a| \sqrt{\lambda} - |\alpha_a|)^2)^{1/2} \cdot (\sum_a (|\xi_a| \sqrt{\lambda} + |\alpha_a|)^2)^{1/2} \le \\
    &(\sum_a |\xi_a \sqrt{\lambda} - \alpha_a|^2)^{1/2} \cdot (\sum_a |\xi_a|^2 \lambda + \sum_a |\alpha_a|^2 + 2\sum_a |\xi_a\alpha_a| \sqrt{\lambda} )^{1/2} \le \epsilon\sqrt{\lambda} \cdot (2\lambda + 2\sum_a |\xi_a\alpha_a| \sqrt{\lambda} )^{1/2} \le \\
    &\epsilon\sqrt{\lambda} \cdot (2\lambda + 2 (\sum_a |\xi_a|^2)^{1/2} (\sum_a|\alpha_a|^2)^{1/2} \sqrt{\lambda} )^{1/2} = \epsilon\sqrt{\lambda} \cdot (2\lambda + 2 \sqrt{\lambda} \sqrt{\lambda} )^{1/2} = 2\epsilon \lambda.
\end{align}
The first equality is the conjugate identity, the inequality after is Cauchy-Schwartz. It is followed by a triangle inequality for $(|\xi_a| \sqrt{\lambda} - |\alpha_a|)^2 \le(|\xi_a \sqrt{\lambda} - \alpha_a|)^2 $ and the expansion of the term $(|\xi_a| \sqrt{\lambda} + |\alpha_a|)^2$. Then we use directly the assumption to bound the first term, while the second term expansion simplifies since $\sum_a |\xi_a|^2=1, \sum_a |\alpha_a|^2=\lambda$. The rest is another application of Cauchy-Schwartz.
\end{proof}
While the order of approximation in \cref{eq:series_approx_order} starts with the unitaries and then the selection probabilities, we found it more instructive to first discuss the errors related to PREP, i.e., the selection probabilities.
\subsection{Errors in PREP}

\subsubsection{$\error_{\rchi}$}\label{sssec:error_rchi}
The register $\mcX$ is a superposition made by QROM with target amplitudes $\lambda_{\rchi}/\lambda$. As shown in \cref{lem:qrom_superposition_error}, the error in estimating the normalized state is $\epsilon =\frac{n\pi}{2^{n_{\rchi}}}$ where $n$ is the number of qubits in register $\mcX$. Since we have four operators, $n=2$.  Thus the equation determining $\error_{\rchi}$ after taking into account the normalization $\lambda$ and using \cref{lem:qrom_denormalized_bd} is:
\begin{align}\label{eq:error_rchi_deriv}
    \error_{\rchi} \le 2\cdot \frac{2\pi }{2^{n_{\rchi}}}\lambda \implies n_{\rchi}=\lceil \log(\frac{4\pi \lambda}{\error_{\rchi}})\rceil .
\end{align}
Note that in the OAE setting, $\error_T$ \cite[Eq. (D29)]{su2021fault} is the closest analog to our $\error_{\rchi}$.

\subsubsection{$\error_B$}\label{sssec:error_B}
This error is derived similarly to the previous one. It approximates the normalized state in register $f$ (\cref{eq:prep_b_omega_omega_prime}) up to error $\epsilon = \frac{4\pi}{2^{n_B}}$ as we use $4$ qubits to denote the two coordinates $\omega,\omega'$. The normalization factor $\lambda$ (in the context of  \cref{lem:qrom_denormalized_bd}) is $\frac{\eta  2^{2n_p-2}}{2}\sum_{\omega,\omega'} |\langle\bb_\omega,\bb_{\omega'}\rangle|$, and thus the error induced is 
\begin{align}\label{eq:error_B_deriv}
    \error_B \le 2\frac{4\pi}{2^{n_B}} \frac{\eta  2^{2n_p-2}}{2}\sum_{\omega,\omega'} |\langle\bb_\omega,\bb_{\omega'}\rangle| \implies n_B = \lceil \log(\frac{4\pi \eta 2^{2n_p-2}}{\error_B}\sum_{\omega,\omega'} |\langle\bb_\omega,\bb_{\omega'}\rangle|) \rceil.
\end{align}
If the lattice is orthogonal, then $\error_B$ is bounded by half the estimate above (\cref{rmk:lcu_t_b_getting_dropped}).

\subsubsection{$\error_{NL}$}\label{sssec:error_k}
The superposition over $\ket{t}_{k_{NL}'}\ket{\sigma}_{k_{NL}}\ket{\sgn(c_{t,\sigma})}_{s_{NL}}$ in \cref{eq:prep_nl_state} is made by a QROM reading $\tau+4$ qubits representing $\ket{t}_{k_{NL}'}\ket{\sigma}_{k_{NL}}$. According to \cref{lem:qrom_superposition_error}, this leads to an error $\epsilon = \frac{\pi(\tau+4)}{2^{n_{NL}}}$ in preparing the normalized state. By \cref{lem:qrom_denormalized_bd}, the error induced on the selection probabilities is
\begin{align}\label{eq:error_k_deriv}
    &\error_{NL} \le 2\frac{(\tau+4) \pi  }{2^{n_{NL}}} \eta \sum_{t,\sigma} |\frac{c_{t,\sigma}n_t}{\Ps(n_t,b_r)}| \implies \\
    &n_{NL} = \lceil \log(\frac{2(\tau+4) \pi \eta \sum_{t,\sigma} |\frac{c_{t,\sigma}n_t}{\Ps(n_t,b_r)}| }{\error_{NL}}) \rceil 
\end{align}

Let us explain the factor $\eta \sum_{t,\sigma} |\frac{c_{t,\sigma}n_t}{\Ps(n_t,b_r)}| $, which is supposed to be the factor $\lambda$ in \cref{lem:qrom_denormalized_bd}. First, note that $\eta$ is simply taking into account the sum over the $\eta$ electrons. For $ \sum_{t,\sigma} |\frac{c_{t,\sigma}n_t}{\Ps(n_t,b_r)}| $, recall that the QROM in \cref{appsssec:prep_state_nl} gives amplitudes $\xi_a$ for $a=(t,\sigma)$, approximating $\alpha_a = (\frac{n_t|c_{t,\sigma}|}{\Ps(n_t,b_r)})^{1/2}$. Further, we needed to prepare the uniform superposition over $n_t$ basis states enumerating nuclei of atomic type $t$. As a result, the estimation of our error is more relaxed than the one in \cref{lem:qrom_denormalized_bd} appears: instead of estimating $\sum_{a=(t,\sigma)} ||\xi_a|^2 \lambda - |\alpha_a|^2 |$, one has to estimate $\sum_{a=(t,\sigma)} p_t||\xi_a|^2 \lambda - |\alpha_a|^2 |$ where $p_t = \Ps(n_t,b_r) \le 1$; thus the same bound still applies, where we substitute for $\lambda = \sum_a \frac{n_t|c_{t,\sigma}|}{\Ps(n_t,b_r)}$ and $\epsilon = \frac{(\tau+4) \pi  }{2^{n_{NL}}}$.

\subsubsection{$\error_{M_V}$}\label{sssec:error_M_V}
$\error_{M_V}$ has a similar derivation to $\error_M$ in the OAE case \cite[Eq. (111)]{su2021fault}, and we follow that very closely while generalizing it to arbitrary lattice:
\begin{align}\label{eq:error_M_V_derivation}
    &\error_{M_V}=||V-\widetilde{V}|| \le \frac{2\pi\eta(\eta-1)}{ \Omega} \sum_{\mu=2}^{n_p+1}\sum_{\bnu \in \mcB_\mu} |\frac{1}{G_\nu^2}-\frac{1}{G_\nu'^2}| \le \\
    &\frac{2\pi\eta(\eta-1)}{ \Omega} \sum_{\mu=2}^{n_p+1}\sum_{\bnu \in \mcB_\mu} \frac{16}{M_V2^{2\mu}b_{min}^2} \le \frac{2\pi \eta(\eta-1)}{ \Omega}\frac{4}{M_Vb_{min}^2}(7\times 2^{n_p+1}-9n_p - 11-3\times 2^{-n_p}) \implies \\
    &\error_{M_V} \le \frac{8\pi\eta(\eta-1)}{\Omega M_V b_{min}^2}(7\times 2^{n_p+1}-9n_p - 11-3\times 2^{-n_p}) \implies \\
    &n_{M_V} = \lceil \log\left(\frac{8\pi\eta(\eta-1)}{\error_{M_V}\Omega b_{min}^2}(7\times 2^{n_p+1}-9n_p - 11-3\times 2^{-n_p})\right) \rceil
\end{align}
Note the replacement of $\frac{1}{G_\nu'^2}$ by $\frac{16(M_V2^{2\mu}b_{min}^2/(16G_\nu^2) +1)}{M_V2^{2\mu}b_{min}^2}$, which is the estimate made by the inequality test method for the target amplitude $\frac{1}{G_\nu}$. This substitution follows the same reasoning in \cite[Eq. (113)]{su2021fault} when picking $\alpha=1$ in \cite[Eq. (109)]{su2021fault}. Also note that $b_{min} \le 2\pi \Omega^{-1/3}$ with equality in the orthonormal case, which is a sanity check as it shows we can recover \cite[Eq. (111)]{su2021fault} when combined with \cite[Eq. (113)]{su2021fault}.

\subsubsection{$\error_{M_{loc}}$}\label{sssec:error_M_loc}

The error analysis here is similar to $\error_{NL}$ in \cref{sssec:error_k}, as the preparation method of the momentum state for the local term is also based on QROM followed by a uniform superposition over $n_t$ basis states. The coefficients estimated by the QROM are $\alpha_a = \frac{\gamma_t(G_\nu)^{1/2}n_t^{1/2}}{G_\nu \Ps(n_t,b_r)^{1/2}}$, where $a = (t,\bnu)$. After applying QROM, superpositions over $n_t$ nuclei of atomic species $t$ are created which introduce an amplitude of $\sqrt{\frac{\Ps(n_t,b_r)}{n_t}}$. So we need to bound the error $\sum_{a=(t,\bnu)} \frac{n_t}{\Ps(n_t,b_r)} |\xi_a-\alpha_a|^2$. We use the simple bound $\error = \max_t(\frac{n_t}{\Ps(n_t,b_r)}) \cdot \frac{\pi(3n_p+\tau)}{2^{n_{M_{loc}}}}$ where the latter term is the bound on $\sum_{a=(t,\bnu)} |\xi_a-\alpha_a|^2$ given by the QROM approximation of the normalized state (\cref{lem:qrom_superposition_error}). Therefore, by virtue of \cref{lem:qrom_denormalized_bd} with $\lambda$ in that lemma set as $\frac{4\pi\eta}{\Omega}\sum_{t,\bnu} \frac{|\gamma_t(G_\nu)|}{G_\nu^2} $, we obtain
\begin{align}\label{eq:error_M_loc_deriv}
    &\error_{M_{loc}} = \frac{8 \pi^2 \eta \max_t(\frac{n_t}{\Ps(n_t,b_r)}) (3n_p+\tau) \sum_{t,\bnu} \frac{|\gamma_t(G_\nu)|}{G_\nu^2} }{2^{n_{M_{loc}}}\Omega} \implies \\
    &n_{M_{loc}} = \lceil \log\left(\frac{8 \pi^2 \eta \max_t(\frac{n_t}{\Ps(n_t,b_r)}) (3n_p+\tau) \sum_{t,\bnu} \frac{|\gamma_t(G_\nu)|}{G_\nu^2} }{\error_{M_{loc}}\Omega} \right)\rceil.
\end{align}

\subsection{Errors in SEL}
\subsubsection{$\error_R$}\label{sssec:error_R}
For $\error_{R,loc},\error_{R,NL}$, we need to follow the same estimations in \cite[Eqs. (101-103)]{su2021fault}, applied to $U_{loc},U_{NL}$. We let $\widetilde{U_{loc}}$ be the approximation of $U_{loc}$ as a result of using $n_R$ bits to compute the approximation $\widetilde{\bR_I}$ of $\bR_I$, and define $\delta_R = \max_I \|\bR_I - \widetilde{\bR_I}\|$. We have $\delta_R \le \frac{\max \|\bm{a}_i\|}{2^{n_R+1}}$ as $\bR_I = \sum_{i=1}^3 \bm{a}_ir_{i,I}$, where $0\le r_{i,I} \le 1$ are the given fractional coordinates of the nuclei in the cell. 
Given the LCU of $U_{loc}$ in \cref{eq:lcu_loc}, we have:
\begin{align}
    &\error_{R,loc} = \| U_{loc} - \widetilde{U_{loc}}\| \le \frac{4\eta\pi}{\Omega} \sum_{\bnu \in \mcG_0,I} \frac{|\gamma_I(G_\nu)|}{G_\nu^2} |e^{-i\bG_\nu \cdot \bR_I} - e^{-i \bG_\nu \cdot \widetilde{\bR_I}} | \le \\
    &\frac{4\eta\pi}{\Omega} \sum_{\bnu \in \mcG_0,I} \frac{|\gamma_I(G_\nu)|}{G_\nu^2} \|\bG_\nu\| \cdot  \|\bR_I - \widetilde{\bR_I}\| \le
    \frac{2\eta\pi \max \|\bm{a}_i\|}{2^{n_R}\Omega} \sum_{\bnu \in \mcG_0,I} \frac{|\gamma_I(G_\nu)|}{G_\nu}
\end{align}
The LCU of $U_{NL}$ can not be used like its local counterpart to facilitate the estimation of $\error_{R,NL}$. Instead we have to first derive an estimate for the entries of $U_{NL}$. Below, we provide two estimates, the first is the tighter one, the second is more pessimistic but easier to compute and is used for the purpose of resource estimation. Recall
\begin{align}
 f_{I}(\bp,\bq) = \Bigg\{& 4r_0^3B_0~e^{-(G_p^2+G_q^2)r_0^2/2} + \frac{16r_1^5B_1}{3}~(\bm{G}_p \cdot \bm{G}_q)~e^{-(G_p^2+G_q^2)r_1^2/2} \nonumber \\
    & + \left[ \frac{32 r_2^7B_2}{15}~(\bm{G}_p \cdot \bm{G}_q)^2 + \frac{32r_2^7B_2}{45}~(G_p G_q)^2\right]e^{-(G_p^2+G_q^2)r_2^2/2} \Bigg\}
\end{align}
Then, 
$$U_{NL} = \frac{4\pi \eta}{\Omega}\sum_{\bp,\bq \in \mcG, I} e^{-i(\bG_q-\bG_p)\cdot \bR_I} f_{I}(\bp,\bq) \ket{\bp}\bra{\bq}$$
and we get
\begin{align}
    \error_{R,NL} = \|U_{NL} - \widetilde{U_{NL}}\|\le \frac{4\eta\pi \delta_{R}}{\Omega} \sum_{\bnu \in \mcG} G_\nu \sum_I \|\sum_{\bq\in\mcG} f_{I}(\bq-\bnu,\bq) \ket{\bq-\bnu}\bra{\bq}\|
\end{align}
Notice the matrix $\sum_{\bq \in \mcG} f_{I}(\bq-\bnu,\bq) \ket{\bq-\bnu}\bra{\bq}$ is a shift of a diagonal matrix, thus its norm is the maximum entry $F_{I,\bnu} = \max_{\bq} |f_{I}(\bq-\bnu,\bq)|$. It follows:
\begin{align}\label{eq:F_NL_I_bnu}
    \error_{R,NL} \le \frac{4\eta\pi \delta_{R}}{\Omega} \sum_{\bnu \in \mcG, I} G_\nu F_{I,\bnu}
\end{align}
Substituting for $\delta_R$, the total bound is
\begin{align}
    \error_R \le \error_{R,loc}+\error_{R,NL} \le \frac{2\eta\pi \max \|\bm{a}_i\|}{2^{n_R}\Omega} \sum_I (\sum_{\bnu \in \mcG_0}\frac{|\gamma_I(G_\nu)|}{G_\nu} + \sum_{\bnu \in \mcG}G_\nu F_{I,\bnu}) \implies \\
    n_R = \lceil \log\left( \frac{2\eta\pi \max \|\bm{a}_i\|}{\error_R\Omega} \sum_I (\sum_{\bnu \in \mcG_0}\frac{|\gamma_I(G_\nu)|}{G_\nu} + \sum_{\bnu \in \mcG}G_\nu F_{I,\bnu}) \right) \rceil
\end{align}
Computing  $\sum_{\bnu \in \mcG}G_\nu F_{I,\bnu}$ may be time-consuming as the number of entries to compute scales with $N^2$. Thus we opt for an easier to compute bound, by simply adding the absolute value of all entries instead of the above grouping:
\begin{align}
    \error_{R,NL} = \|U_{NL} - \widetilde{U_{NL}}\|\le \frac{4\eta\pi \delta_{R}}{\Omega} \sum_I \sum_{\bp,\bq \in \mcG} \|\bG_q-\bG_p\|  |f_{I}(\bp,\bq)| \le \frac{4\eta\pi \delta_{R}}{\Omega} \sum_I \sum_{\bp,\bq \in \mcG} (G_p+G_q)|f_{I}(\bp,\bq)|
\end{align}
This is followed by the approximation below, where all summations are over $\mcG$:
\begin{align}
    \sum_{\bp,\bq} (G_p+G_q)|f_{I}(\bp,\bq)| \le     2\cdot  |4r_0^3B_0| \big((\sum e^{-G_p^2r_0^2/2})(\sum G_pe^{-G_p^2r_0^2/2}) - \sum G_p e^{-G_p^2r_0^2}\big) + \\
    2 \cdot |\frac{16r_1^5B_1}{3}|  \big((\sum G_p^2e^{-G_p^2r_1^2/2})(\sum G_pe^{-G_p^2r_1^2/2})  - \sum G_p^3 e^{-G_p^2r_1^2}  \big) + \\
    2 \cdot |\frac{128r_2^7B_2}{45}| \big((\sum G_p^3e^{-G_p^2r_2^2/2})(\sum G_p^2e^{-G_p^2r_2^2/2})  - \sum G_p^5 e^{-G_p^2r_2^2}    \big)
\end{align}
where we have used triangle inequality and Cauchy-Schwartz for all applicable expressions involved in $f_{I}$. We have also leveraged the projector nature of the pseudopotential matrix entries to write the estimation above in such a way that it would be easier to compute on a classical computer. Denoting the above estimation by $F_{I}$ we derive the pessimistic bound
\begin{align}\label{eq:error_R_deriv_pessimistic}
    \error_R \le \frac{2\eta\pi \max \|\bm{a}_i\|}{2^{n_R}\Omega} \sum_I ( F_{I}+\sum_{\bnu \in \mcG_0}\frac{|\gamma_I(G_\nu)|}{G_\nu}) \implies \\
    n_R = \lceil \log\left( \frac{2\eta\pi \max \|\bm{a}_i\|}{\error_R\Omega} \sum_I (F_{I} + \sum_{\bnu \in \mcG_0}\frac{|\gamma_I(G_\nu)|}{G_\nu} ) \right) \rceil
\end{align}

\subsubsection{$\error_{\Psi}$}\label{sssec:error_NL}
\begin{table*}[!ht]
\centering
\begin{tabular}{| m{11cm} | m{5cm} |}
\hline
\centering Procedure for $\PREP$ &  \hspace{16mm} Toffoli cost \\
 \hline
Preparing the superposition for register $\mcX$; see \cref{app:gate_cost_reg_mcX}.  & $2\cdot[2 (2^{2+1}-1)+ (n_{\rchi}-3)2]$ \\ 
\hline
The $c,d,e$ registers are equal superpositions over $\eta$ values of $i$ and $j$ in unary; see \cite[Eq. (62)]{su2021fault}. &
$14n_\eta+8b_r-36$ \\
\hline
The $f,g,h$ registers used for $T$; see \cite[Eq. (70)]{su2021fault} for the $g,h$ preparation cost and \cref{app:gate_cost_reg_f} for register $f$. & $2\cdot [2(2^{4+1}-1) + (n_B-3)4 + 2^4 + (n_p-2)]$ \\ 
\hline
The two QROMs used for outputting $\bR_I$ in register $\bR$; see \cref{app:gate_cost_reg_R}. & $2\cdot[2(2^{\tau+\max_t n_t +1})]$\\ 
\hline
Making the uniform superposition on the nuclei of each type $t$ in $k'_{loc},k_{NL}'$ registers; see \cref{sssec:unif_sup_k_loc_k_NL}. & $2\cdot [2(3\max_t n_t - 3v_2(\max_t n_t)+2b_r-9 + 2\cdot2^{\tau})]$\\ 
\hline
The $(k_{NL},k_{NL}',s_{NL})$ register superposition prepared using QROM. $\beta_{NL}$ defined in \cref{eq:beta_k_cost}. & $2\cdot [2( 2\lceil \frac{2^{\tau+4+1}-1}{\beta_{NL}}\rceil +3n_{NL}(\tau+4)\beta_{NL}+2(\tau+4)) + (n_{NL}-3)(\tau+4) + 2^{\tau+2}]+12$\\ 
\hline
Preparing the superposition for the register $(j_V,k_V)$ with amplitudes $1/G_\nu$ using QROM in inequality test; $\beta_V$ defined in \cref{eq:beta_V_cost}. & $(2a_V+1)\cdot [2(2\lceil \frac{2^{3n_p}}{\beta_V}\rceil+ 3\beta_Vn_{M_V})+8(n_p-1)+6n_p+2+n_{M_V}]$ \\ 
\hline
Preparing the superposition for the register $(k_{loc},k'_{loc},s_{loc})$ using QROM; $\beta_{loc}$ defined in \cref{eq:beta_loc_cost}. & $2\cdot [2(2\lceil \frac{2^{3n_p+\tau+1}}{\beta_{loc}}\rceil+ 3\beta_{loc}n_{M_{loc}}(3n_p)+2(3n_p))+(n_{M_{loc}}-3)(3n_p+\tau+1)]$ \\ 
\hline
Toffolis used to compute the registers $\ket{\cdot}_{\rchi}$ for $\rchi \in \{T,V,loc,NL\}$; counted in \cref{ssec:the_sel_subroutine_of_qubitization}. & $6+3+3+\tau+9$\\
\hline
Toffolis used to compute the registers $\ket{\cdot}_{NL,c}$; see \cref{appssec:SEL_implementation}. & $4$\\
 \hline \hline
\centering Procedure for $\SEL$ &  \hspace{16mm} Toffoli cost \\
 \hline 
Controlled swaps of the $p$ and $q$ registers into and out of ancillae (which is used for all four operators); see \cite[Eq. (72)]{su2021fault}. & $12\eta n_p+4\eta-8$ \\ 
\hline
The $\SEL$ cost for $T$; see \cite[Eq. (73)]{su2021fault}. & $5(n_p-1)+2$ \\ 
\hline
Controlled additions and subtractions of $\nu$ into the momentum registers for $U_{loc},V$; see \cite[Eq. (93)]{su2021fault}. & $48n_p$ \\
\hline 
Phasing by $-e^{-i \bG_\nu\cdot \bR_I}$ for $U_{loc}$; see \cite[Eq. (97)]{su2021fault}. & $6n_p n_R$\\ 
\hline
Phasing by $-e^{-i (\bG_q-\bG_p)\cdot \bR_I}$ for $U_{NL}$; see \cite[Eq. (97)]{su2021fault}. & $12n_p n_R$\\ 
\hline
Cost of reflection on $\ket{\Psi_{I,\sigma}}$ for $U_{NL}$ where $n_\qb = n_p+\tau+4$. $\beta_{\Psi}$ defined in \cref{eq:beta_NL}. & $ 6\cdot [2(2\lceil \frac{2^{n_\qb+1}-2^{n_\qb-n_p}}{\beta_{\Psi}}\rceil + 3\beta_{\Psi} n_{\Psi} n_p+2n_p) + (n_{\Psi}-3)n_p]+ (3n_p-1)$\\ 
\hline
Cost of reflection on $\ket{\Psi_{I,2,0}}$ where $n_\qb' = n_p+\tau+2$. $\beta_{\Psi}'$ defined in \cref{eq:beta_NL_prime_cost}. & $5\cdot2\cdot [2\cdot(2^{3+1}-1)+3(n_{\bb}-3)+ 3\cdot[2(2\lceil \frac{2^{n_\qb'+1}-2^{n_\qb'-n_p}}{\beta_{\Psi}'}\rceil + 3\beta_{\Psi}' n_{\Psi} n_p+2n_p) +(n_{\Psi}-3)n_p] + n_{AA}] $\\ 
 \hline \hline
\centering Reflection on state preparation qubits &  \hspace{16mm} Toffoli cost \\
 \hline 
 Reflection on the qubits used in state preparation; see \cref{app:gate_costings_ROT_costing}.  & $2n_\eta+9n_p+n_{M_V}+35+2(\tau+\max_t n_t)$ \\
 \hline
\end{tabular}
\caption{The costs involved in implementing the qubitization $Q = (2\ket{\bm{0}}\bra{\bm{0}} - \mathbbm{1}) \PREP_H^\dagger \cdot \SEL_H \cdot  \PREP_H$, which involves the block-encoding of the Hamiltonian ($\PREP_H$ with its uncomputation $\PREP_H^\dagger$, and $\SEL_H$), and the reflection $(2\ket{\bm{0}}\bra{\bm{0}} - \mathbbm{1})$ on state preparation qubits.}
\label{tab:gatecosts}
\end{table*}

Here, we estimate the error induced by the QROM-based preparation of the Gaussian states $\ket{\Psi_{I,\sigma}}$. First, notice that any approximation $||\ket{\psi} -\widetilde{\ket{\psi}}||\le \epsilon $ give the following estimate on the projection operator $||\ket{\psi}\bra{\psi} - \widetilde{\ket{\psi}}\widetilde{\bra{\psi}}|| \le 2\epsilon + \epsilon^2 \le 3\epsilon$. Assuming the lattice is orthogonal, we apply three QROMs, one for each coordinate, to implement $\mathbbm{1}-2\ket{\Psi_{I,\sigma}}\bra{\Psi_{I,\sigma}}$. Therefore, the error in approximating $\ket{\Psi_{I,\sigma}}$, up to first order, is $\epsilon = 3\epsilon'$ for $\epsilon' = \frac{n\pi}{2^{n_{\Psi}}}$ with $n=n_p+\tau+4$ or $n=n_p+\tau+2$ when we prepare states $\ket{\Psi_{I,\sigma}}$ of type (c) in \cref{eq:psi_I_o_decomp_20}. We will not consider the second and higher order of errors in our approximation as their impact is too small, and we have already a pessimistic estimate above by taking $\epsilon^2 \le \epsilon$. The case for partially orthogonal lattices is simpler, as there are two QROMs and therefore two associated errors, however one must select $n=2n_p+\tau+4$, or $n=2n_p+\tau+2$. Thus, for the reflections, we have:
$$|| (\mathbbm{1}-2\ket{\psi}\bra{\psi}) - (\mathbbm{1}-2\ket{\widetilde{\psi}}\bra{\widetilde{\psi}}) || \le 18\epsilon'$$
Hence the error is estimated as:
\begin{align}\label{eq:error_NL_deriv}
    &\error_{\Psi} = ||U_{NL}- \widetilde{U_{NL}}|| \le \frac{18(n_p+4+\tau)\pi\eta}{2^{n_{\Psi}}} \sum_{I,\sigma} |c_{I,\sigma}| \implies \\
    &n_{\Psi} = \lceil \log\left(\frac{18(n_p+4+\tau)\pi\eta}{\error_\Psi} \sum_{I,\sigma} |c_{I,\sigma}| \right)\rceil.
\end{align}
when the lattice is orthogonal, and we simply replace $(n_p+4+\tau)$ by $(2n_p+4+\tau)$ if the lattice is partially orthogonal.

As explained in \cref{appssec:SEL_implementation}, the implementation of the reflection onto $\ket{\Psi_{I,2,0}}$ requires an additional QROM to prepare a one-hot-encoded superposition $\sum_{i=1}^{3}\sqrt{\frac{(\bb_i)_i^4}{\sum_j (\bb_j)_j^4}}\ket{i}$. Denoting by $n_{\bb}$ the number of qubits used by the \qromalgo~rotations to prepare the superposition, the associated error $\error_{NL'}$ satisfies $\error_{NL'} \le 3\pi \cdot 2^{-n_{\bb}}$. However, to make the analysis easier for our case-studies while also retaining accuracy later on in our resource estimations, we choose $n_{\bb}$ so large that it gives the superposition with a negligible error. By choosing $n_{\bb} = 50$, the error is of order $8e-15$, which is small enough to be safely ignored in our analysis. Even for larger materials than those in our case studies with a much larger $\lambda_{NL}$, one can always increase $n_{\bb}$ by a small amount without any significant accrued gate and qubit cost.

\section{Gate costings}\label{app:gate_costings}

In this section, we derive the Toffoli cost expressions for the algorithm. A summary of the results is given in  \cref{tab:gatecosts}. We make a few general remarks on this table:

1. The cost calculated for the PREP subroutines always includes the uncomputation part by $\PREP^\dagger$. Hence, most costs have a leading factor of two.

2. The cost for the reflections on $\ket{\Psi_{I,\sigma}}$ for all $\sigma$ will need to change slightly for the materials with non-orthogonal lattices (see \cref{ssec:SEL_costings}).

3. While we list the Toffoli cost in \cref{tab:gatecosts}, we also study Toffoli depth, and calculating the latter mostly involves replacing the QROM costs expressions in \cref{tab:gatecosts} by their depth formulae in \cref{table:qroms_cost}.

4. The parameters $\beta_-$ determine the space-depth tradeoff of the QROM (\cref{table:qroms_cost}). Optimizing the expressions in \cref{table:qroms_cost} in terms of $\beta_-$ generally leads to a much higher total number of qubits compared to the AE case. Thus we determine $\beta_-$ in a way that satisfies constraints on the number $n_{\text{dirty}}$ of dirty qubits that can be used. When estimating depth, there will be an additional constraint posed by the maximum allowed number $n_{\text{tof}}$ of simultaneous Toffoli applications.

Finally, note that the gate and qubit costings in the AE case for general lattices is the same as OAE in \cite{su2021fault} with two exceptions:
\begin{itemize}
\item The costing for preparing the state of register $f$, over $\ket{\omega,\omega',\sgn(\langle \bb_\omega,\bb_{\omega'} \rangle)}$,
\item The costing for preparing the momentum state superposition, which is identical to preparing the momentum state for $V$ in the PP-based algorithm.
\end{itemize}
The gate and qubit estimates for these are derived further below.

\subsection{Toffoli cost of Prepare}
\subsubsection{Register $\mcX$}\label{app:gate_cost_reg_mcX}
The superposition $\sum_{\rchi \in \{0,1\}^2} \sqrt{\frac{\lambda_{\rchi}}{\lambda}}\ket{\rchi}_{\mcX}$ on two qubits is prepared using the \textsc{Select} QROM-based \qromalgo, reading $2$ qubits and using a register of size $n_{\rchi}$ for the precision of the rotations. Its gate cost is directly derived from \cref{eq:qrom_gate_cost_superposition_sel_variant}, substituting $n=2$ and $b=n_{\rchi}$. Notice that the inverse of the operation in $\PREP^\dagger$ is responsible for the doubling of the cost, yielding $2\cdot[2 (2^{2+1}-1)+ (n_{\rchi}-3)2]$.

\subsubsection{Register $f$}\label{app:gate_cost_reg_f}
We use a \textsc{Select} QROMs to prepare the superposition in register $f$ and the same variant to output $\sgn(\langle \bb_\omega, \bb_{\omega'}\rangle)$ in the fifth qubit. According to \cref{eq:qrom_gate_cost_superposition_sel_variant}, the former has cost $2(2^{4+1}-1)+(n_B-3)4$ as $4$ qubits are read, while the latter has cost $2^4$ (\ref{eq:qrom_gate_cost_output_sel_variant}). The inverse of these operations for $\PREP^\dagger$ has the same cost.

\subsubsection{Register $\textit{\textbf{R}}$}\label{app:gate_cost_reg_R}
We used two \textsc{Select} QROMs to output the nuclei coordinates $\bR_I$ into register $\bR$, one for each of the local and non-local term. Each reads the nuclei type and its enumeration, i.e. $\tau + \max_t n_t$ qubits. They further read the qubit of the register $loc$ and $NL$ to effectively control their output. By a direct application of \cref{eq:qrom_gate_cost_output_sel_variant}, the cost is $2\cdot [2\cdot 2^{\tau + \max_t n_t +1}]$, which is doubled due to $\PREP^\dagger$.

\subsubsection{Uniform superpositions in $k'_{loc}, k_{NL}'$}\label{sssec:unif_sup_k_loc_k_NL}
We apply the algorithm and cost estimate in \cite[App. A.2]{lee2021even} for preparing the uniform superposition over $n_t$ basis states enumerating the nuclei of type $t$ in register $k'_{loc}$ and $k'_{NL}$, giving the Toffoli cost $2\cdot[3 \max_t n_t - 3v_2(\max_t n_t)+2b_r-9]$ for each, where $v_2(x)$ is the largest power of two dividing $x \in \mbbZ$. The uncomputation has the same cost, therefore doubling the said amount.

There is one small subtlety that we did not address when implementing the PREP states of $U_{loc}$ and $U_{NL}$. Our application of \cite[App. A.2]{lee2021even} assumes that the number of qubits $n_t$ needed for creating the superposition over $n_t$ many basis states, and the number $n_t$ itself, are both stored in some registers. These two registers are computed using the \textsc{Select} variant of QROM that needs to read only the atomic type, and has cost $2\cdot 2^{\tau}$, which is further doubled due to the inverse of PREP. 

\subsubsection{QROM-based preparation of the superposition over $\ket{\sigma}_{k_{NL}}\ket{t}_{k_{NL}'}\ket{\sgn(c_{t,\sigma})}_{s_{NL}}$}\label{app:gate_cost_k_NL_s_NL}
We apply the \textsc{SelSwapDirty} QROM-based \qromalgo, reading $\tau+4$ qubits with rotation precision $n_{NL}$ (\ref{eq:error_k_deriv}). Thus, following \cref{eq:qrom_gate_cost_superposition_1d}, the cost is 
\begin{equation}
    2\left( 2\left\lceil \frac{2^{\tau+4+1}-1}{\beta_{NL}}\right\rceil +3n_{NL}(\tau+4)\beta_{NL} + 2(\tau+4)\right) + (n_{NL}-3)(\tau+4).
\end{equation}
Following \cref{eq:optimal_beta_cost}, we have
\begin{align}\label{eq:beta_k_cost}
    \beta_{NL} = \left\lfloor \min\left(\sqrt{\frac{2(2^{\tau+4+1}-1)}{3n_{NL}(\tau+4)}},\frac{n_{\text{dirty}}}{n_{NL}}\right) \right\rfloor.
\end{align}
Here we use the material and the notation in our review of circuit depth of QROM in \cref{app:qrom_parallelization} to make the estimates. Recall that the \textsc{SelSwapDirty} QROM uses $\beta_{NL}n_{NL}$ dirty ancillae and we must have $\beta_{NL}n_{NL} \le n_{\text{dirty}}$.

In addition, there is the cost of computing $\sgn(c_{t,\sigma})$ done by a \textsc{Select} QROM, which is $2^{\tau+2}$, as it only reads the atomic type along with two bits that determine to which of the three subgroups $\{(0), (1,\omega), \{(2,0),(2,(\omega,\omega'))\}\}$ does $\sigma$ belong to, which corresponds to $\sigma=0$, $1\le \sigma\le 3$, $4\le \sigma\le 10$. Those two bits are computed by inequality tests and require $4\times 3=12$ Toffolis. While they can be uncomputed without any Toffolis, the rest of the cost is doubled as we implement the inverse of PREP, yielding the total cost 
\begin{equation}
2\cdot\left[2( 2\left\lceil \frac{2^{\tau+4+1}-1}{\beta_{NL}}\right\rceil +3n_{NL}(\tau+4)\beta_{NL}+2(\tau+4)) + (n_{NL}-3)(\tau+4) + 2^{\tau+2}\right]+12.
\end{equation}

\textit{Circuit Depth. }  We take into account the maximum simultaneous Toffoli application $n_{\text{tof}}$. Let us denote by $\kappa_{NL}$ the parallelization factor we wish to use for this computation (see \cref{app:qrom_parallelization} for the exact definition), for which $\beta_{NL} \kappa_{NL} \le n_{\text{tof}}$. This is in addition to the previous dirty qubit constraint. Then the depth according to \cref{eq:qrom_gate_cost_superposition_1d_parallelized_v2} becomes:
\begin{align}
    2\cdot\left[2\left( 2\lceil \frac{2^{\tau+4+1}-1}{\beta_{NL}}\right\rceil +3\left\lceil \frac{n_{NL}}{\kappa_{NL}}\right\rceil(\tau+4)\left\lceil \log(\beta_{NL})\right\rceil +2(\tau+4)) + (n_{NL}-3)(\tau+4) + 2^{\tau+2}\right]+12
\end{align}
where fllowing \cref{eq:optimal_beta_depth}
\begin{align}\label{eq:beta_k}
    \beta_{NL} = \left\lfloor \min\left(\frac{n_{\text{dirty}}}{n_{NL}}, \frac{n_{\text{tof}}}{\kappa_{NL}}, \frac{2\cdot (2^{\tau+4+1}-1)}{3n_{NL}(\tau + 4)/\kappa_{NL}} \log_e(2)\right)\right\rfloor .
\end{align}

\subsubsection{Toffoli cost for preparing the momentum state superposition for $V$}
All the subroutines used in the inequality test procedure, such as preparing the superposition over $\ket{m,\mu}$, or  checking the signs of $\bnu$ and testing $\bnu \neq 0$, remain exactly the same as in \cite{su2021fault}. Their Toffoli cost totals $8(n_p-1)+6n_p+2+n_{M_V}$. Next, we compute the cost for the \textsc{SelSwapDirty} QROM that reads $\bnu$ (i.e. $3n_p$ qubits) and outputs $\lceil\frac{M_V(2^{\mu-2}b_{min})^2}{G_\nu^2}\rceil$ (\cref{appssec:prep_implementation}) with precision $n_{M_V}$. 
This cost is obtained as in (\cref{eq:qrom_gate_cost_output_1d}). The uncomputation of this QROM in $\PREP^\dagger$ doubles this, yielding a total of $2(2\lceil \frac{2^{3n_p}}{\beta_V}\rceil+ 3\beta_Vn_{M_V})$. By a derivation similar to \cref{eq:optimal_beta_cost}, the optimal value for $\beta_V$ is
\begin{align}\label{eq:beta_V_cost}
    \beta_V = \left\lfloor \min\left(\sqrt{\frac{2(2^{3n_p})}{3n_{M_V}}},\frac{n_{\text{dirty}}}{n_{M_V}}\right) \right\rfloor.
\end{align}
The rest of the cost is derived in \cite{su2021fault} and totals $8(n_p-1)+6n_p+2+n_{M_V}$. Finally, the $a_V$ amplitude amplifications multiplies the total by $2a_V+1$.

\textit{Circuit Depth. } Using a similar notation to the previous part, the constraints are $\beta_V \kappa_V \le n_{\text{tof}}, \beta_V n_{M_V} \le n_{\text{dirty}}$. Then, following \cref{eq:qrom_gate_cost_output_1d_parallelized_v2}, the circuit depth is:
\begin{align}
    (2a_V+1)\left[2\left(2\left\lceil \frac{2^{3n_p}}{\beta_V}\right\rceil+ 3\lceil \log(\beta_V)\rceil \left\lceil \frac{n_{M_V}}{\kappa_V}\right\rceil \right) + 8(n_p-1)+6n_p+2+n_{M_V}\right]
\end{align}
where
\begin{align}\label{eq:beta_V}
    \beta_V =  \left\lfloor \min\left(\frac{n_{\text{dirty}}}{n_{M_V}}, \frac{n_{\text{tof}}}{\kappa_V}, \frac{2\cdot 2^{3n_p}}{3n_{M_V}/\kappa_V} \log_e(2)\right)\right\rfloor
\end{align}

\subsubsection{Toffoli cost for the momentum state superposition for $U_{loc}$}
Recall that we used \textsc{SelSwapDirty} QROMs to directly prepare the superposition over $\ket{\bnu}_{k_{loc}}\ket{t}_{k'_{loc}}\ket{\sgn(\gamma_t(G_\nu))}_{s_{loc}}$. The QROMs eventually read $3n_p+\tau$ qubits, and we assume the precision of the rotations to be $n_{M_{loc}}$ bits. Notice that in the superposition preparation scheme, the last QROM oracle outputs $\sgn(\gamma_I(G_\nu))$ into the $s_{loc}$ register, thus uses a register of size $n_{M_{loc}}+1$ for its output. As a result the cost formula is slightly changed. Given $\PREP^\dagger$, the total cost is 
\begin{equation}
    2\cdot \left[2\left(2\left\lceil \frac{2^{3n_p+\tau+1}-1}{\beta_{loc}}\right\rceil+ 3\beta_{loc}n_{M_{loc}}(3n_p-1)+3\beta_{loc}(n_{M_{loc}}+1)+2(3n_p)\right)+(n_{M_{loc}}-3)(3n_p+\tau)\right],
\end{equation}
where we set 
\begin{align}\label{eq:beta_loc_cost}
    \beta_{loc} =  \left\lfloor \min\left(\sqrt{\frac{2(2^{3n_p+\tau+1}-1)}{3(n_{M_{loc}}+1)(3n_p)}},\frac{n_{\text{dirty}}}{n_{M_{loc}}}\right) \right\rfloor.
\end{align}

\textit{Circuit Depth. } Following the same strategy in the previous case, we consider the constraints  $\beta_{loc} \kappa_{loc} \le n_{\text{tof}}, \beta_{loc}(n_{M_{loc}}+1) \le n_{\text{dirty}}$. We are using $n_{M_{loc}}+1$ instead of $n_{M_{loc}}$ for the very last QROM, therefore, to get an upper bound of the resource estimate, we use $n_{M_{loc}}+1$ in the constraints. The depth is
\begin{align}
    2\cdot &\big[ 2\left(2\left\lceil \frac{2^{3n_p+\tau+1}-1}{\beta_{loc}} \right\rceil + 3\left\lceil \log(\beta_{loc})\right\rceil \left\lceil \frac{n_{M_{loc}}}{\kappa_{loc}} \right\rceil (3n_p-1)+3\left\lceil \log(\beta_{loc})\right\rceil \left\lceil\frac{n_{M_{loc}}+1}{\kappa_{loc}}  \right\rceil + 2(3n_p)\right) \nonumber \\
    &+ (n_{M_{loc}}-3)(3n_p+\tau)\big]
\end{align}
where
\begin{align}\label{eq:beta_loc}
    \beta_{loc} = \left\lfloor \min\left(\frac{n_{\text{dirty}}}{n_{M_{loc}}+1}, \frac{n_{\text{tof}}}{\kappa_{loc}}, \frac{2\cdot (2^{3n_p+\tau+1}-1)}{3(n_{M_{loc}}+1)(3n_p)/\kappa_{loc}} \log_e(2)\right)\right\rfloor
\end{align}

\subsubsection{Toffolis to compute the selection qubit registers and $\ket{\cdot}_{NL,c}$}
In \cref{appssec:SEL_implementation}, we mentioned how to compute the register $\ket{\cdot}_{\rchi}$ that flags the success of the state preparation for the operator $\rchi$. Their Toffoli costs are $6,3,3,\tau+9$, for $\rchi=T,V,U_{loc},U_{NL}$, respectively. The inverse of PREP  in this case can be done via measurements and Clifford gates. Similarly, the cost for $\ket{\cdot}_{NL,c}$, defined in the implementation of $\SEL_{NL}$, is four Toffolis, and can be uncomputed via measurements and Clifford gates.

\subsection{Toffoli cost of Select}\label{ssec:SEL_costings}
Below we briefly go over cost estimates that are very similar to the OAE setting.

\textit{CSWAPs and the SEL cost for $T$. }Recall that at the beginning and end of all $\SEL_{\rchi}$ operators, there is a shared circuit of CSWAPs. The cost estimate used in \cite[Eq. (72-73)]{su2021fault} applies without any change, to perform the CSWAPs on the plane wave vectors for all four operators and the bits $r,s$ of coordinates $\omega,\omega'$ for the operator $T$, copying them back and forth to an auxiliary register and implementing the necessary phases for $T$.

\textit{Controlled addition/subtraction of the momentum state vector. } While the same cost in the OAE setting \cite[Eq. (93)]{su2021fault} was computed as $24n_p$, here one needs to take into account two separate application of this operation for $U_{loc},V$ yielding $48n_p$.

\textit{Phasings by the nuclei coordinates. } There are two such phasings, one for the local part, which cost $6n_pn_R$ is computed exactly as in \cite[Eq, (97)]{su2021fault}, and another for the non-local, which cost is $ 12n_pn_R$, as we apply the phase once for $\bG_q$ and then for $\bG_p$ after the reflection on $\ket{\Psi_{I,\sigma}}$. 

\subsubsection{Reflection on $\ket{\Psi_{I,\sigma}}$}
There are two costs to be estimated. One is the preparation of $\ket{\Psi_{I,\sigma}}$ by the operator $U_{I,\sigma}$ (and its inverse) and the other is the reflection onto $\ket{0}_{NL}\ket{0}^{\otimes 3n_p}$. The latter's Toffoli cost is $(3n_p+1)-2$. Below we compute the costs for a material with an orthogonal lattice, and end with a remark on the changes required for the non-orthogonal case.

There are three sets of \textsc{SelSwapDirty} QROMs used for each coordinate, along with their inverse that follows the reflection on $\ket{0}_{NL}\ket{0}^{\otimes 3n_p}$. This means a factor of six. The QROMs read $n_\qb = n_p+\tau +4 $, i.e. the number of bits in one coordinate $\bp_\omega$ of the plane wave vector $\bp$, the type of the nuclei and Gaussian state $(t,\sigma)$. But since $\tau+4$ of these qubits are already determined by the PREP state, the iterative process of QROM to build the superposition happens only $n_p$ times. Thus the reflection costs 
\begin{equation}
6\left[2\left(2\left\lceil \frac{2^{n_\qb+1}-2^{n_\qb-n_p}}{\beta_{\Psi}}\right\rceil + 3\beta_{\Psi} n_{\Psi} n_p+2n_p\right) + (n_{\Psi}-3)n_p\right]+3n_p-1
\end{equation}
with
\begin{align}\label{eq:beta_NL_cost}
    \beta_{\Psi}=  \left\lfloor \min\left(\sqrt{\frac{2(2^{n_\qb+1}-2^{n_\qb-n_p})}{3n_{\Psi}n_p}},\frac{n_{\text{dirty}}}{n_{\Psi}}\right) \right\rfloor.
\end{align}
To the cost above, one needs to add the one for implementing $V_{I,\sigma}$ for the reflection onto $\ket{\Psi_{I,2,0}}$. We refer to \cref{appssec:SEL_implementation} for the relevant notations.  We can summarize this cost as $5\cdot  2(QROM_i+3\cdot QROM_{\Psi}+(n_{AA}-2)+2)$, where 
\begin{itemize}
    \item the factor of five is because of the exact amplitude amplification,
    \item $n_{AA}-2$ is due to using the trick in \cref{app:aa_fewer_steps},
    \item the additional $2$ is to compute the flag qubit out of the three hot encoded qubits $\ket{i}$, and
\end{itemize} 
\begin{align}
    QROM_i &= 2\cdot (2^{3+1}-1) + 3 (n_{\bb}-3),\\\label{eq:QROM_Psi}
    QROM_{\Psi} &= 2\left(2\left\lceil \frac{2^{n_\qb'+1}-2^{n_\qb'-n_p}}{\beta_{\Psi}'}\right\rceil + 3\beta_{\Psi}' n_{\Psi} n_p+2n_p\right)+(n_{\Psi}-3)n_p,
\end{align}
with $n_\qb' = n_p+\tau+2$ and $n_{\bb} = 50$ (\cref{sssec:error_NL}). $QROM_i$ is the cost of preparing the one-hot-encoded superposition $\sum_{i=1}^{3}\sqrt{\frac{(\bb_i)_i^4}{\sum_j (\bb_j)_j^4}}\ket{i}$, and $QROM_{\Psi}$ is the cost for preparing the superposition $\ket{\Psi_{i,I,\sigma}}$ for each given $i$. According to \cref{eq:optimal_beta_cost} the expression for $\beta_{\Psi}'$ is
\begin{align}\label{eq:beta_NL_prime_cost}
    \beta_{\Psi}' = \left\lfloor \min\left(\sqrt{\frac{2(2^{n_\qb'+1}-2^{n_\qb'-n_p})}{3n_{\Psi}n_p}},\frac{n_{\text{dirty}}}{n_{\Psi}}\right) \right\rfloor.
\end{align}
Notice that while we are using a different $\beta_{\Psi}'\neq \beta_{\Psi}$, the same $n_{\Psi}$ in \cref{eq:beta_NL_cost} is used. This enables us to also take into account the error in preparing $\ket{\Psi_{I,2,0}}$  when analyzing the error $\error_{\Psi}$ due to the choice $n_{\Psi}$ (\cref{sssec:error_NL}).

\textit{Non-orthogonal case. }The non-orthogonal lattices in our case studies allow a decomposition of the Gaussian states into a 1D and 2D factor. The cost for $U_{I,\sigma}$ will change to include that of two different QROMs reading $n_p+\tau+4$ and $2n_p+\tau+4$ qubits. Further, for the implementation of $V_{I,\sigma}$, we only need one amplitude amplification, and the hot-encoded superposition above is over two qubits. The necessary changes to the cost formulae are straightforward. For example, for $V_{I,\sigma}$, the cost changes to $3\cdot2 (QROM_{\Psi,1}+QROM_{\Psi,2}+QROM_{i}'+2)$, where $QROM_{\Psi,1}$ is the same as (\ref{eq:QROM_Psi}), and
\begin{align}
    QROM_i' &= 2\cdot (2^{2+1}-1) + 3 (n_{\bb}-3),\\ \label{eq:cost_QROM_Psi_2}
    QROM_{\Psi,2} &= 2\left(2\left\lceil \frac{2^{n_\qb'+1}-2^{n_\qb'-n_p}}{\beta_{\Psi}'}\right\rceil + 3\beta_{\Psi}' n_{\Psi} (2n_p)+4n_p\right)+(n_{\Psi}-3)2n_p,
\end{align}
where we note the substitution $n_\qb' = 2n_p+\tau+2$ and replacing $n_p$ with $2n_p$ where appropriate.

\textit{Circuit Depth. }Given the decomposition of QROM to three parallel QROMs, we can parallelize the computation more so than in the previous procedures. We apply the three sets of QROMs in parallel, in addition to reducing their depth using \cref{eq:qrom_gate_cost_superposition_1d_parallelized_v2}, yielding a circuit depth of
\begin{align}\label{eq:runtime_cost_parallel_NL}
    2\left[2\left(2\left\lceil \frac{2^{n_\qb+1}-2^{n_\qb-n_p}}{\beta_{\Psi}}\right\rceil + 3\lceil \log(\beta_{\Psi})\rceil \left\lceil \frac{n_{\Psi}}{\kappa_{\Psi}}\right\rceil n_p + 2n_p\right) + (n_{\Psi}-3)n_p\right]
\end{align}
to implement $\ket{\Psi_{I,\sigma}}$ for $\sigma \neq (2,0)$. However note that the number of dirty qubits used in this case is $3\beta_{\Psi}n_{\Psi}\le n_{\text{dirty}}$ with the factor $3$ due to simultaneously preparing the three 1D Gaussian states. Similarly we have $3\beta_{\Psi}\kappa_{\Psi} \le n_{\text{tof}}$. These constraints imply the following optimization
\begin{align}\label{eq:beta_NL}
    \beta_{\Psi} = \left\lfloor\min \left(\frac{2(2^{n_\qb+1}-2^{n_\qb- n_p})}{3n_{\Psi}n_p/\kappa_{\Psi}} \log_e(2) , \frac{n_{\text{dirty}}}{3n_{\Psi}},  \frac{n_{\text{tof}}}{3\kappa_{\Psi}}\right) \right\rfloor .
\end{align}
The same arguments applies for the circuit depth of $QROM_{\Psi}$ :
\begin{align}
    2\left(2\left\lceil \frac{2^{n_\qb'+1}-2^{n_\qb'-n_p}}{\beta_{\Psi}'}\right\rceil + 3\lceil \log(\beta_{\Psi}')\rceil \left\lceil \frac{n_{\Psi}}{\kappa_{\Psi}}\right\rceil n_p + 2n_p\right) + (n_{\Psi}-3)n_p
\end{align}
where 
\begin{align}\label{eq:beta_NL_prime}
    \beta_{\Psi}' = \left\lfloor\min \left(\frac{2(2^{n_\qb'+1}-2^{n_\qb'- n_p})}{3n_{\Psi}n_p/\kappa_{\Psi}'} \log_e(2) , \frac{n_{\text{dirty}}}{3n_{\Psi}},  \frac{n_{\text{tof}}}{3\kappa_{\Psi}'}\right) \right\rfloor .
\end{align}

In the non-orthogonal case studies, the circuit depth of the 2D Gaussian state is always the larger one, and therefore it is the only one that needs to be taken into account. For example, for $V_{I,\sigma}$, this means the depth formula is $3\cdot 2(QROM_{\Psi,2,d}+QROM_{i}'+2)$ where $QROM_{\Psi,2,d}$ is the depth of the circuit with cost $QROM_{\Psi,2}$ (\ref{eq:cost_QROM_Psi_2}). Again, the changes are straightforward to calculate. For example, to compute the part related to preparing the 2D Gaussian state, $n_\qb$ in \cref{eq:runtime_cost_parallel_NL} must be changed to $2n_p+\tau+4$, and the conditions for $n_{\text{dirty}},n_{\text{tof}}$ change to $2\beta_{\Psi}\kappa_{\Psi} \le n_{\text{tof}}, \ 2\beta_{\Psi}\kappa_{\Psi} \le n_{\text{tof}}$ (similarly for $\beta_{\Psi}'$). 

\subsection{Toffoli cost of the reflection on the preparation qubits}\label{app:gate_costings_ROT_costing}

The qubitization $Q$ applies a reflection on the qubits used in the state preparation. As argued in \cite[Eq. (98)]{su2021fault}, the number of these qubits equals the Toffoli cost of this reflection. In our case, borrowing from \cite[App. C]{su2021fault} in some cases, the qubits that need to be reflected upon are
\begin{itemize}
    \item The two qubits that are rotated to select between the operators $\rchi$.
    \item There are $n_p$ qubits for each of $r$ and $s$, for a total of $2n_p$.
    \item There are five qubits for register $f$, for a total of 5. Note the flag qubits for ineligible states are rezeroed by $\PREP^\dagger$, so no reflection needed on them.
    \item There are $2n_\eta = 2\lceil\log \eta \rceil$ qubits for registers $d,e$ with 2 qubits that are rotated, for a total of $2n_\eta+2$ qubits (the flag qubits are rezeroed, so no reflection on them).
    \item Qubits used in the momentum state preparation for $U_{loc},V$, specifically the following for $V$:
        \begin{enumerate}
            \item $3(n_p+1)$ qubits for $\ket{\nu}$.
            \item $n_p$ qubits for the unary-encoded $\ket{\mu}$.
            \item $n_{M_{V}}$ qubits for $\ket{m}$.
        \end{enumerate}
   and the following for $U_{loc}$:
    \begin{enumerate}
        \item $3n_p$ qubits for $\ket{\nu}$ .
        \item One qubit for $\ket{\sgn(\gamma_I(G_\nu))}$,
        \item the $\tau+\max_t n_t+2$ qubits used for the two uniform superpositions over type and enumeration of nuclei, along with the two rotated ancillae.
    \end{enumerate} 
    All for a total of $7n_p+n_{M_V}+6+\tau+\max_t n_t$.
    \item We have the arithmetic overflow qubits for the addition/subtraction, which were computed to be 6 in \cite{su2021fault}, and is 9 for us since the subtraction is done separately for $U_{loc}$ and $V$.
    \item The qubits used in registers $k_{NL},k_{NL}',s_{NL}$, of which there are $\tau+\max_t n_t+5$, and the rotation ancilla needed for $\ket{t_j}$ for the uniform superposition over $n_t$ nuclei, for a total of $\tau+\max_t n_t+6$.
    \item Two qubits used for the trick in \cref{app:aa_fewer_steps} for implementing the reflections onto states of type $\sigma=(2,0)$ in $\SEL_{NL}$. This trick is not used for materials with non-orthogonal lattices.
    \item Three qubits for the one-hot-encoded superposition $\sum_{i=1}^{3}\sqrt{\frac{(\bb_i)_i^4}{\sum_j (\bb_j)_j^4}} \ket{i}$. This is two for non-orthogonal lattices.
    \item All ancilla qubits used by all QROMs are either dirty and from the circuit itself, which are returned to their initial state, or are clean (such as in the \textsc{Select} variant of QROM), which are uncomputed either by the procedure itself or measurement and Clifford gates after the QROM. 
\end{itemize}
Note that all other flag or ancilla qubits not mentioned above, such as the qubits $\ket{\cdot}_{\rchi}$s, are rezeroed. Overall, the total number of qubits and Toffoli cost for the reflection is at most
\begin{align}
    2n_\eta+9n_p+n_{M_V}+35+2(\tau+\max_t n_t)
\end{align}
and three less for the non-orthogonal cases.

\section{Qubit costings}\label{app:qubit_cost}
We list the entire qubit cost below borrowing from \cite[App. C]{su2021fault} in parts where the subroutines involved stay the same.
\begin{enumerate}
    \item The system register has size $3\eta n_p$.
    \item The control register for the phase estimation needs $\left\lceil \log\left(\left\lceil \frac{\pi\lambda}{2\epsilon_{\rm pha}} \right\rceil\right)\right\rceil$ qubits.
    \item The phase gradient state that is used for the phase rotations. There are $\max(n_R+1,n_{\rchi},n_B,n_{NL},n_{\bb},n_{\Psi})$ bits used in the phasing; each of these $n_X$'s is the number of bits of a phase gradient state used within the subroutine building the superposition on the corresponding register.
    \item One qubit for the $\ket{T}$ state used catalytically for controlled Hadamards.
    \item Two qubits for register $\mcX$.
    \item Four qubits for the four registers $\ket{\cdot}_{\rchi}$.
    \item The $2n_\eta+5$ qubits from the preparation of the superpositions over $i$ and $j$; $\eta$ qubits for each of these registers, 2 qubits for the rotation preparing the superpositions, 2 qubits that flag the success of the two preparations, and 1 qubit that flags whether $i=j$.
    \item Eight qubits for register $f$. Five qubits in register $f$ along with three additional qubits needed to flag the eligible basis states. 
    \item Five qubits in total used by the two QROMs to make the superposition in register $f$. Note five are used and rezeroed immediately, before four of them are reused to make the $\sgn(\omega,\omega')$ (and rezeroed again).
    \item The states $r$ and $s$ are prepared in unary, and need $n_p$ qubits each, for a total of $2n_p$.
    \item The register $\bR$ itself uses $3n_R$ qubits.
    \item The two QROMs used for computing the register $\bR$ each use and immediately clean $\tau+\max_t n_t+1$ ancilla qubits.
    \item The register $k'_{loc}$ uses $\tau +\max_t n_t$ qubits.
    \item The registers $k_{NL},k_{NL}',s_{NL}$ use  $\tau +\max_t n_t + 4 + 1$ qubits.
    \item Four qubits for the uniform superposition over the $n_t$ nuclei of each type $t$ in $k'_{loc},k_{NL}'$. One for the rotation for each and one for the success flag.
    \item Making the superposition on $k_{NL},k_{NL}',s_{NL}$ consumes $\beta_{NL}n_{NL}$ dirty qubits and $n_{NL}+(\tau+4)$ clean qubits which are all returned to their initial state.
    \item Flagging the eligible states in $k_{NL}',k_{NL}$ uses four qubits, one for the type, one for the $\sigma$ index, and two to compute whether these two flags and the flag for $n_t$ superposition are successful.
    \item Three inequality tests using temporarily $4$ qubits to determine where $\sigma$ lies ($\sigma=0, 1\le \sigma \le 3, 4\le\sigma\le 10$).
    \item The same three inequality tests use two qubits to determine $\sigma$'s associated type for computing $s_{NL}$.
    \item Outputting the sign into $s_{NL}$ uses and rezeroes $\tau+2$ ancilla qubits.
    \item The preparation of the momentum state superposition for $V$:
        \begin{enumerate}
            \item Storing $\bnu$ requiring $3(n_p+1)$ qubits.
            \item $\mu$ needs $n_p$ qubits.
            \item $n_{M_V}$ qubits for the equal superposition state.
            \item $3n_p+2$ qubits for testing $\ket{-0}$, including the flag qubit.
            \item $2n_p+1$ qubits used in signaling whether $\bnu$ is outside $\mcB_\mu$, including the flag qubit. 
            \item The $n_{M_V}$ bits required by QROM to compute one side of the inequality test, and the $\beta_Vn_{M_V}$ dirty and $n_{M_V}+3n_p$ clean ancilla qubits it uses and immediately returns to initial state to compute that side.
            \item The qubit $j_V$ resulting from the inequality test.
            \item Two qubits, one flagging success of all three of  inequality test, no negative zero and $\bnu$ not outside $\mcB_\mu$, and the other an ancilla qubit used to produce the triply controlled Toffoli.
        \end{enumerate}
    \item The preparation of the momentum state superposition for $U_{loc}$:
        \begin{enumerate}
            \item Storing $\bnu$ requires $3n_p$ qubits.
            \item Storing $I$ requires $\tau +\max_t n_t $ qubits.
            \item One qubit for the $s_{loc}$ register.
            \item The $n_{M_{loc}}$ bits required by QROM with the sign into $s_{loc}$, and the $\beta_{loc}(n_{M_{loc}}+1)$ dirty and $(n_{M_{loc}}+1)+(3n_p+\tau)$ clean ancilla qubits it uses and immediately returns to initial state to compute that RHS.
        \end{enumerate}
    \item $3+3+3+\tau+9$ ancilla qubits used to compute $\ket{\cdot}_{\rchi}$, along with $4$ more to compute and store $\ket{\cdot}_{NL,c}$.
    \item The temporary ancillae used in the addition and subtraction of $\bnu$ for the $U_{loc}, V$.  This is identical to the OAE case and we simply recall it to be thorough. The cost here is given by items (a) and (c), giving a total of $5n_p+1$:
        \begin{enumerate}
        \item In implementing the $\SEL$ operations, we need to control a swap of a momentum register into an ancilla, which takes $3n_p$ qubits for the output. The $n_\eta-1$ temporary ancillae for the unary iteration on the $i$ or $j$ register can be ignored because they are fewer than the other temporary ancillae used later.
        \item We use $2n_p+3$ temporary qubits to implement the block encoding of $T$, where we copy components $\omega,\omega'$ of the momentum into an ancilla, copy out two bits of these components of the momentum, then perform a controlled phase with those two qubits as control as well as the qubit flagging that $T$ is to be performed.
        \item For the controlled addition or subtraction by $\bnu$ in the $\SEL$ operations for $U_{loc}$ and $V$, we use $n_p$ bits to copy a component of $\bnu$ into an ancilla, then there are another $n_p+1$ temporary qubits used in the addition, for a total of $2n_p+1$ temporary qubits in this part. Even though the momentum $\bnu$ registers are different for $U_{loc},V$, the same temporary ancillae can be used since the computations are not done in parallel for the two.
        \item There are also temporary qubits used in converting the momentum back and forth between signed and two's complement, but these are fewer than those used in the previous step.
        \end{enumerate}
    \item There are 2 overflow qubits obtained every time we add or subtract a component of $\bnu$ into a momentum.
    All these qubits must be kept, giving a total of $9$.
    \item There are also temporary qubits used in the arithmetic to implement $e^{-i \bG_\nu \cdot \bR_I},e^{-i \bG_q \cdot \bR_I},e^{i \bG_p \cdot \bR_I}$. The arithmetic requires a maximum of $2(n_R-2)$ qubits. Note that the $\bR_I$ can be output by the QROM, the phase factor applied, and the $\bR_I$ erased, after (or before) the arithmetic for addition/subtraction of $\bnu$ is performed. So we only need to take the maximum of the $5n_R-4$ qubits used in this item and item 11, and the $5n_p+1$ temporary qubits used in item 24.
    \item Two qubits used to control between adding and subtracting $\bnu$ in order to make $\SEL$ self-inverse. This is required to employ the techniques of \cite{babbush2018encoding} to avoid controlled application of the qubitization operator.
    \item The clean qubit cost of the QROM used for $\ket{\Psi_{I,\sigma}}$ is $3n_{\Psi}+3(n_p+\tau+4)$ which are rezeroed, and the dirty qubit cost of the QROM is $3\beta_{\Psi}n_{\Psi}$. This changes to $2n_{\Psi}+3n_p+2\tau+8$ and $2\beta_{\Psi}n_{\Psi}$ for non-orthogonal lattices. Notice this is assuming simultaneous application of the QROMs for the coordinates. Otherwise, the costs just listed would become $n_{\Psi}+(n_p+\tau+4),\beta_{\Psi}n_{\Psi},n_{\Psi}+2n_p+\tau+4,\beta_{\Psi}n_{\Psi}$, respectively.
    \item The reflection after $\ket{\Psi_{I,\sigma}}$ preparation uses $3n_p-1$ temporary qubits.
    \item The one-hot-encoded superposition for $\ket{\Psi_{I,2,0}}$ requires $3$ qubits, along with the QROM preparing that superposition requiring an additional $3$ clean qubits which are immediately rezeroed. Both of these requirements become two instead of three when the lattice is non-orthogonal.
    \item The clean qubit cost of the QROM used for $\ket{\Psi_{I,2,0}}$ is $3n_{\Psi}+3(n_p+\tau+2) $ which are rezeroed, and the dirty qubit cost of the QROM is $3\beta_{\Psi}'n_{\Psi}$. This changes to $2n_{\Psi}+3n_p+2\tau+4$ and $2\beta_{\Psi}'n_{\Psi}$ for non-orthogonal lattices. Similar to a previous item, these are listed assuming simultaneous application of the QROMs for the coordinate. Otherwise, the costs listed become $n_{\Psi}+(n_p+\tau+2),\beta_{\Psi}'n_{\Psi},n_{\Psi}+2n_p+\tau+2,\beta_{\Psi}'n_{\Psi}$, respectively. 
    \item The (controlled) rotation following the entire qubitization was computed previously to need $2n_\eta+10n_p+n_{M_V}+35+2(\tau+\max_t n_t)$ Toffolis which is the same as the number of temporary qubits it needs.
\end{enumerate}
\begin{rmk}
In case we use the technique in \cref{app:aa_fewer_steps}, we need to change the phase gradient qubit cost (item 3) to $\max(n_R+1,n_{\rchi},n_{AA},n_B,n_{NL},n_{M_{loc}},n_{\bb},n_{\Psi})$. Furthermore, we need to add 2 each time we use this technique, which we do once for $\SEL_{NL}$ when the lattice is orthogonal. 
\end{rmk}
To sum up the above, we need to take into account which temporary ancillae can be reused, and take the maximum of the dirty qubits and clean qubits to get the total number of qubits used in the algorithm. 

First, we list the subroutines and the number of temporary ancillae they need. All ancillae below are clean and temporary unless mentioned otherwise.
\begin{enumerate}
    \item Five ancillae for the QROM in register $f$. 
    \item Register $\bR$ with $3n_R$ qubits can be used and cleaned immediately before other procedures as described in the listing above.
    \item $2(\tau+\max_t n_t+1)$ qubits for the QROMs computing register $\bR$.
    \item $\beta_{NL}n_{NL}$ dirty and $n_{NL}+(\tau+4)$ clean qubits for the QROM on $k_{NL},k_{NL}',s_{NL}$.
    \item Four qubits used for the inequality tests $\sigma=0,1\le \sigma\le 3, 3<\sigma$.
    \item $\tau +2$ ancillae for the QROM to output $s_{NL}$.
    \item $\beta_Vn_{M_{V}}$ dirty and $n_{M_{V}}+3n_p$ clean qubits for the QROM on $k_{V}$.
    \item $\beta_{loc}(n_{M_{loc}}+1)$ dirty and $(n_{M_{loc}}+1)+(3n_p+\tau)+n_{M_{loc}}$ clean qubits for the QROM on $k_{loc},k'_{loc}$.
    \item $5n_p+1$ and the $2(n_R-2)$ ancilla mentioned in the previous listing in items 24 and 26.
    \item $3n_{\Psi}+3(n_p+\tau+4)$ clean and $3\beta_{\Psi}n_{\Psi}$ dirty qubit cost of the QROM for implementing $\SEL_{NL}$. This changes to $2n_{\Psi}+3n_p+2\tau+8$ and $2\beta_{\Psi}n_{\Psi}$ for non-orthogonal lattices. We also recall the comment on our assumption of simultaneous applications of QROMs for the coordinates.
    \item $3n_p-1$ qubits for the reflection for $\SEL_{NL}$.
    \item $3$ clean ancilla along with $3n_{\Psi}+3(n_p+\tau+2)$, and $3\beta_{\Psi}'n_{\Psi}$ dirty qubits for the QROM for implementing $\SEL_{NL}$. This changes to $2+2n_{\Psi}+3n_p+2\tau+4$ and $2\beta_{\Psi}'n_{\Psi}$ for non-orthogonal lattices. We also recall the comment on our assumption of simultaneous applications of QROMs for the coordinates.
    \item $2n_\eta+9n_p+n_{M_V}+35+2(\tau+\max_t n_t)$ temporary ancillae used for the reflection.
\end{enumerate}
Regarding the dirty qubits requirement, since none of the operations above happen in parallel, we can simply take the maximum of them all to obtain $n_{\text{dirty}}$ as the dirty qubits requirement. 

From item 1 to item 8, all calculations are for PREP. With the exception of item 2, their clean ancillae are rezeroed immediately and thus the clean qubit requirement is the maximum of all the requirements: $n_{clean,PREP} = \max(5, 2(\tau+\max_t n_t), n_{NL}+(\tau+4), 4 , \tau +2, n_{M_V}+3n_p, (n_{M_{loc}}+1)+(3n_p+\tau)+n_{M_{loc}})$.

Once PREP is done, the SEL operations take over (items 9-12), and can use the $n_{clean,PREP}$ qubits freed up. Some SEL and PREP operations happen in specific orders as described for example in item 26 in the previous listing. The temporary clean qubit requirement $n_{tmp,clean}$ is obtained by
\begin{align}\label{eq:n_tmp_clean_H}
n_{tmp,clean,H} =& \max(5n_p+1,5n_R-4) + \nonumber\\
&\max(n_{clean,PREP},3n_p-1,3n_{\Psi}+3(n_p+\tau+4), 3n_{\Psi}+3(n_p+\tau+2)+3)\\
n_{tmp,clean} =& \max(n_{tmp,clean,H}, 2n_\eta+9n_p+n_{M_V}+35+2(\tau+\max_t n_t)),
\end{align}
and the total qubit cost is $n_{total} = \max(n_{\text{dirty}}, n_{tmp,clean}+n_{clean})$, where $n_{clean}$ is all the clean qubit costs that were not in the temporary clean list above.

The formula above is for orthogonal lattices, and when the lattice is non-orthogonal, we use $2n_{\Psi}+3n_p+2\tau+8$ and $2n_{\Psi}+3n_p+2\tau+4+2$ instead of the corresponding terms in the equation for \cref{eq:n_tmp_clean_H}.

\section{Resource estimation configuration and detailed results}\label{app:res_est_det_results}

\subsection{Parameters for resource estimation}\label{sssec:depth_params}

According to \cref{eq:min_possible_depth}, the minimum possible circuit depth of \qromalgo~using \textsc{SelSwapDirty} QROMs is $O(n^2+(b-3)n)$. However this is at the expense of exponentially many dirty qubits and simultaneous Toffoli applications. To derive a more reasonable depth, we set the limits with which QROM can optimize its circuit depth. These limits are values we set for the parameters $n_{\text{dirty}}$ and $n_{\text{tof}}$. The two set the constraints \cref{eq:beta_nparallel,eq:beta_ndirty} on the space-depth trade-off parameter $\beta$ (\ref{eq:optimal_beta_depth}). Notice that the Toffoli cost has (only) the dirty qubit constraint (\ref{eq:optimal_beta_cost}). As a result it is the latter that needs to be determined first. 

To do so, we first run a simulation to compute the clean qubit cost of the all-electron algorithm. Notice that the clean qubit cost is independent of the trade-off parameter $\beta$. Therefore, for each $N$, it is well-defined to set $n_{\text{dirty}}$ as the number of clean qubits that the all-electron algorithm needs for $N$ many plane waves. This allows for a fair comparison as for a fixed number of plane waves, both algorithms have the same available number of dirty qubits to optimize the depth of their QROM computations.

Given the fact that the AE algorithm always consumes far more clean qubits, we would like our PP-based circuits with optimized depth to use as much as possible the dirty qubits available. Choosing a small value for $n_{\text{tof}}$ can prevent that and we found that setting $n_{\text{tof}} = 500$ is approximately the smallest value that satisfies this requirement for all of our case studies.

As defined in \cref{app:qrom_parallelization}, there is another parameter $\kappa$ involved in the QROM depth calculation that is subroutine-dependent (we have $\kappa_{\Psi},\kappa_{loc},$ etc.). However we set a uniform $\kappa = 1$ value for all our estimations. Changing this value has shown insignificant or worsening impact on the depth.

The success probability threshold for the amplitude amplification involved in preparing $\sum_{\bnu} \frac{1}{G_\nu}\ket{\bnu}$ (\ref{eq:prep_v_state}) has no impact on the qubit cost. However, it changes the Toffoli depth as $\lambda_V \propto p_{\text{th}}^{-1}$. We set $p_{\text{th}}=0.75$ for all pseudopotential experiments. For the all-electron setting, given that the algorithm in \cite{su2021fault} has slight variations according to the value of the initial probability of success (denoted by $p_\nu$ in \cite[Thm. 4]{su2021fault}), we select a $p_{\text{th}}$ that gives the lowest Toffoli depth. 

Lastly, we need to determine the errors listed in \cref{app:errors} while targeting the chemical accuracy $\error = 0.043\text{eV}$. Recall that we have to satisfy:
\begin{align}
    \error^2 \ge \error_{\text{QPE}}^2 + (\error_{\rchi}+\error_B+\error_{NL}+\error_R+\error_{M_V}+\error_{M_{loc}} + \error_{\Psi})^2.
\end{align}
As the cost formula $\left\lceil \frac{\pi \lambda}{\error_{\text{QPE}}}\right\rceil (2\PREP_{cost} + \SEL_{cost})$ suggests, among all errors, the inverse of $\error_{\text{QPE}}$ contributes directly to the cost, as others only do so polylogarithmically. Therefore, we allocate the vast majority (99.5\%) of the error to $\error_{\text{QPE}}$, while distributing the rest equally among all other errors: $\error_{\rchi}=\error_B=\error_{NL}=\error_R=\error_{M_V}=\error_{M_{loc}} = \error_{\Psi} = \frac{\sqrt{0.5\%} \times \error}{7}$.

\subsection{Clean and total qubit cost}

As mentioned in the main text, the highest contribution to the qubit cost comes from the encoding of the plane waves, needing $3\eta n_p$ many clean qubits. The qubit cost in the \cref{table:materials_NPP_vs_NAE_depth,table:mat_dis,table:mat_limnfo,table:mat_limnnio,table:mat_limno} is measured in two parts, clean and total. The clean cost is lower than the total cost for the pseudopotentials, but they are equal in the AE setting, as enforced by the definition of $n_{\text{dirty}}$ in \cref{sssec:depth_params}. Notice these numbers are reported for the $n_{\text{tof}}=500$ runs optimizing the depth of the circuit, and not for the optimized costs. Further, this is only relevant to the PP-based algorithm, and the clean and dirty costings when optimizing the Toffoli cost of the PP-based algorithm are even smaller.

{
	\renewcommand{\arraystretch}{1.95}
\begin{table}[!h]
\centering
	\begin{tabular}{|Sc||Sc|Sc||Sc|Sc||Sc|Sc|}
 \hline
		\multirow{2}{*}{Material}& \multicolumn{2}{Sc||}{Clean qubit}&\multicolumn{2}{Sc||}{Total qubit}&\multicolumn{2}{Sc|}{Toffoli depth}\\
        \cline{2-7}
        & PP & AE &  PP & AE & PP & AE\\
		\hline
		\limno  &  \textbf{9808} & 24974 & \textbf{15136} & 24974 & $\mathbf{1.01\times 10^{15}}$ & 2.13$\times 10^{19}$  
        \\\hline
        \limnnio &  \textbf{11130} & 29784 & \textbf{18017} & 29784
 & \textbf{9.59$\mathbf{\times 10^{14}}$} & 3.59$\times 10^{19}$   
		\\\hline
		\limnfo &  \textbf{10260} & 26629 & \textbf{16121} & 26629 & \textbf{8.55$\mathbf{\times 10^{14}}$} & 1.16$\times 10^{19}$  
		\\\hline
        \dis &  \textbf{2650} & 4859 & \textbf{2847} & 4859 & \textbf{1.93$\times \mathbf{10^{13}}$} & 1.59$\times 10^{17}$ \\
        \hline
	\end{tabular}
	\caption{Depth and clean qubit cost estimation for $N^{\text{PP}}$ and $N^{\text{AE}}$ plane waves (\cref{table:materials}), when optimizing for Toffoli depth. This is in contrast to \cref{table:materials_NPP_vs_NAE_cost} where the optimized quantity was Toffoli cost. Better numbers are indicated in bold. The total qubit cost in \cref{table:materials_NPP_vs_NAE_cost} includes dirty qubits. The clean qubit cost is lower than the total qubit cost for the pseudopotentials, but they are  equal in the AE setting, as enforced by the definition of $n_{\text{dirty}}$ in this section.}
	\label{table:materials_NPP_vs_NAE_depth}
\end{table}
}

{
	\renewcommand{\arraystretch}{1.1}
\begin{table}[!h]
\centering
	\begin{tabular}{|Sc||Sc|Sc||Sc|Sc||Sc|Sc||Sc|Sc|}
 \hline
		\multirow{2}{*}{$N$}& \multicolumn{2}{Sc||}{Toffoli depth}& \multicolumn{2}{Sc||}{Clean qubit}&\multicolumn{2}{Sc||}{Total qubit}&\multicolumn{2}{Sc|}{Toffoli cost}\\
        \cline{2-9}
        & PP & AE &  PP & AE & PP & AE & PP & AE\\\hline
			$10^3$ & 5.85$\times 10^{12}$ & \textbf{4.86$\mathbf{\times 10^{12}}$} & \textbf{2278} & 2366 & \textbf{2278} & 2366 & 1.54$\times 10^{13}$ & \textbf{6.78$\mathbf{\times 10^{12}}$}
			\\\hline
			$10^4$ & 1.93$\times 10^{13}$ & \textbf{1.46$\mathbf{\times 10^{13}}$} & \textbf{2650} & 2867 & \textbf{2847} & 2867 & 6.38$\times 10^{13}$ & \textbf{2.94$\mathbf{\times 10^{13}}$}
			\\\hline
			$10^5$ & 9.09$\times 10^{13}$ & \textbf{8.07$\mathbf{\times 10^{13}}$} &  \textbf{2938} & 3365 & \textbf{3362} & 3365  & 2.68$\times 10^{14}$ & \textbf{1.44$\mathbf{\times 10^{14}}$}\\\hline
		\end{tabular}
		\caption{Resource estimation for dilithium iron silicate (\dis). This material resource estimation was studied earlier in \cite{batterypaper}. Better numbers are indicated in bold. The algorithm is using almost the entire dirty qubit capacity ($n_{\text{dirty}}$), as the total qubit count of both algorithms are almost equal.}
		\label{table:mat_dis}
	\end{table}
}

{
	\renewcommand{\arraystretch}{1.1}
\begin{table}[!h]
\centering
	\begin{tabular}{|Sc||Sc|Sc||Sc|Sc||Sc|Sc||Sc|Sc|}
 \hline
		\multirow{2}{*}{$N$}& \multicolumn{2}{Sc||}{Toffoli depth}& \multicolumn{2}{Sc||}{Clean qubit}&\multicolumn{2}{Sc||}{Total qubit}&\multicolumn{2}{Sc|}{Toffoli cost}\\
        \cline{2-9}
        & PP & AE &  PP & AE & PP & AE & PP & AE\\\hline
			$10^3$ & \textbf{9.08$\mathbf{\times 10^{13}}$} & 3.03$\times 10^{14}$ & \textbf{7602} & 10906 & \textbf{7602} & 10906 & \textbf{2.16$\mathbf{\times 10^{14}}$} & 3.58$\times 10^{14}$
			\\\hline
			$10^4$ & \textbf{3.09$\mathbf{\times 10^{14}}$} & 8.21$\times 10^{14}$ & \textbf{8937} & 13525 & \textbf{13524} & 13525 & \textbf{9.64$\mathbf{\times 10^{14}}$} &       1.18$\times 10^{15}$
			\\\hline
			$10^5$ & \textbf{8.55$\mathbf{\times 10^{14}}$} & 2.17$\times 10^{15}$ &  \textbf{10260} & 16147 & \textbf{16121} & 16147  & \textbf{3.87$\mathbf{\times 10^{15}}$} &      4.41$\times 10^{15}$\\\hline
		\end{tabular}
		\caption{Resource estimation for lithium manganese oxyfluoride (\limnfo). Better numbers are indicated in bold. Depth circuit optimization consumes almost the entire $n_{\text{dirty}}$ available for $N=10^4,10^5$, meaning it is determined by the dirty qubit constraint, instead of the Toffoli parallelization limits. For $N=10^3$, the optimization is fully achieved, i.e. $n_{\text{dirty}},n_{\text{tof}}$ are both large enough that we obtain the minimum possible depth for the optimized QROM circuit depths.}
		\label{table:mat_limnfo}
	\end{table}
}

{
	\renewcommand{\arraystretch}{1.1}
	\begin{table}[!h]
 \centering
	\begin{tabular}{|Sc||Sc|Sc||Sc|Sc||Sc|Sc||Sc|Sc|}
 \hline
		\multirow{2}{*}{$N$}& \multicolumn{2}{Sc||}{Toffoli depth}& \multicolumn{2}{Sc||}{Clean qubit}&\multicolumn{2}{Sc||}{Total qubit}&\multicolumn{2}{Sc|}{Toffoli cost}\\
        \cline{2-9}
        & PP & AE &  PP & AE & PP & AE & PP & AE\\\hline
			$10^3$ & \textbf{1.13$\mathbf{\times 10^{14}}$} &  2.75$\times 10^{14}$ & \textbf{8244} & 12171 & \textbf{8244} & 12171 & \textbf{2.38$\mathbf{\times 10^{14}}$} &       3.45$\times 10^{14}$
			\\\hline
			$10^4$ & \textbf{3.30$\mathbf{\times 10^{14}}$} & 7.18$\times 10^{14}$ & \textbf{9699} & 15105 & \textbf{15087} & 15105 & \textbf{1.08$\mathbf{\times 10^{15}}$} &       1.36$\times 10^{15}$
			\\\hline
			$10^5$ & \textbf{9.59$\mathbf{\times 10^{14}}$} & 2.01$\times 10^{15}$  &  \textbf{11130} & 18045 & \textbf{18017} & 18045  & \textbf{4.84$\mathbf{\times 10^{15}}$} &       5.59$\times 10^{15}$\\\hline
		\end{tabular}
		\caption{Resource estimation for LLNMO (\limnnio). Better numbers are indicated in bold. The comparison between clean and total qubit count shows a situation somewhat similar to \cref{table:mat_limnfo}.}
		\label{table:mat_limnnio}
	\end{table}
}

{
	\renewcommand{\arraystretch}{1.1}
	\begin{table}[!h]
	\centering
			\begin{tabular}{|Sc||Sc|Sc||Sc|Sc||Sc|Sc||Sc|Sc|}
 \hline
		\multirow{2}{*}{$N$}& \multicolumn{2}{Sc||}{Toffoli depth}& \multicolumn{2}{Sc||}{Clean qubit}&\multicolumn{2}{Sc||}{Total qubit}&\multicolumn{2}{Sc|}{Toffoli cost}\\
        \cline{2-9}
        & PP & AE &  PP & AE & PP & AE & PP & AE\\\hline
			$10^3$ & \textbf{1.09$\mathbf{\times 10^{14}}$} & 2.10$\times 10^{14}$ & \textbf{7273} & 10248 & \textbf{7273} & 10248 & \textbf{2.37$\mathbf{\times 10^{14}}$} & 2.62$\times 10^{14}$
			\\\hline
			$10^4$ & \textbf{3.37$\mathbf{\times 10^{14}}$} & 5.80$\times 10^{14}$ & \textbf{8551} & 12702 & \textbf{12696} & 12702 & 1.12$\times 10^{15}$ & \textbf{9.82$\mathbf{\times 10^{14}}$}
			\\\hline
			$10^5$ & \textbf{1.01$\mathbf{\times 10^{15}}$} & 1.63$\times 10^{15}$  &  \textbf{9808} & 15156 & \textbf{15136} & 15156  & 5.00$\times 10^{15}$ & \textbf{4.12$\mathbf{\times 10^{15}}$}\\\hline
		\end{tabular}
		\caption{Resource estimation for lithium manganese oxide (\limno). Better numbers are indicated in bold. The situation is similar to \cref{table:mat_limnfo}.}
		\label{table:mat_limno}
	\end{table}
}

\subsection{Results}\label{ssec:res_est_results_app}
In all cases, taking into account the better accuracy, we can conclude that the PP-based is the better alternative. However, this difference is most clear in \cref{table:materials_NPP_vs_NAE_depth} where we choose the right number of plane waves $N^{\text{PP}}$ and $N^{\text{AE}}$ to hit chemical accuracy with both PP and AE calculations. There are multiple reasons why the depth and cost is competitive, even for the same number of plane waves:
\begin{itemize}
    \item The number of electrons is about half of the all-electron case.
    \item In our simulations, we have observed how $\lambda_{loc}+\lambda_{NL}$ compares to $\lambda_U$ in the AE case, and their difference is multiple times more than what could be justified by the previous item alone. Indeed, $\lambda_{loc}$ (\cref{eq:lambda_loc}) involves the exponentially decaying factor $e^{-G_\nu^2r_{loc}^2/2}$, while $\lambda_{NL}$ has similar factors (\cref{eq:lambda_NL}), aided by the fact that the number of unitaries involved in the LCU for $U_{NL}$ is smaller compared to other operators ($\sim \eta L$).
    \item Even though the total PREP and SEL cost for qubitizing the pseudopotential Hamiltonian are larger than those in the AE case (2-4 times), it was important that our algorithm manages to control the qubitization cost, despite the more complicated expressions defining the pseudopotential matrix entries. This is accomplished thanks to our LCU and subroutine choices like QROM.
\end{itemize}

\section{List of notations}\label{appsec:list_notations}
\begin{itemize}
    \item $\bm{p}$, $\bm{q}$, $\bm{\nu}$ -- Plane wave indices as integer vectors. Normal font version is used for indexing other variables.
    \item $\omega$ -- Index with three values $\omega= 1,2,3$
    \item $\bm{a}_\omega$ -- Primitive lattice vectors
    \item $\bb_\omega$ -- Reciprocal lattice vectors
    \item $b_{min}$ -- The smallest singular value ($\sigma_3(B)$) of the lattice matrix $B := (\bb_1,\bb_2,\bb_3)$
    \item $\mathcal{G}$  -- Set of reciprocal lattice vectors $\sum_\omega \bb_\omega p_\omega$,~\eqref{eq:pw_p}
    \item $\mathcal{G}_0$ -- $\mathcal{G}_0 = \mathcal{G} \setminus (0,0,0)$ 
    \item $\bm{G}_p$, $\bm{G}_q$, $\bm{G}_\nu $, etc.  -- Reciprocal lattice vectors corresponding to plane wave indices~\eqref{eq:pw_g}
    \item ${G}_p$ -- Length of $\bm{G}_p$, equal to $\|\bm{G}_p\|$
    \item $\hat{\bm{G}}_p$ -- angular component of $\bm{G}_p$.
    \item $\bm{R}$ -- Nuclear coordinates, also denoting the register $\bR$ storing those coordinates
    \item $N$ -- Number of plane waves
    \item $n_p$ -- The number of qubits used for each coordinate in the plane wave register, $n_p = \lceil \log_2(N^{1/3}+1) \rceil$
    \item $n_{M_{loc}}$ -- Defined as the number of qubits used by the QROM rotations for $\PREP_{loc}$, equal to $\lceil \log(M_{loc}) \rceil$
    \item $n_{M_V}$ -- Defined as the precision used in the inequality test for $\PREP_V$, equal to $\lceil \log(M_{V}) \rceil$
    \item $n_{NL}$ -- Number of qubits used for the QROM rotations for preparing register $NL$'s, part of $\PREP_{NL}$
    \item $n_B$ -- Number of qubits used for the QROM rotations for preparing register $f$, part of $\PREP_T$
    \item $n_{\rchi}$ -- Number of qubits used for the QROM rotations for preparing register $\rchi$
    \item $n_R$ -- Number of qubits used to represent nuclei coordinates
    \item $n_{\Psi}$ -- Number of qubits used for the QROM rotations for preparing the Gaussian states
    \item $n_t$ -- Defined as $\lceil \log(N_t) \rceil$, where $N_t$ is the number of nuclei with atomic type $t$ in the cell
    \item $n_{\text{dirty}}$ -- Number of dirty qubits available for QROM 
    \item $n_{\text{tof}}$ -- The maximum allowed number of simultaneous Toffoli applications
    \item $\tau$ -- Defined as $\lceil \log(\mcT) \rceil$, where $\mcT$ is the number of atomic species in the cell
    \item $\eta$ -- Number of electrons
    \item $v_2(\cdot)$ -- The highest power of two dividing an integer
    \item $n_\eta$ -- Defined as $\lceil \log(\eta) \rceil$
    \item $b_r$ -- The number of bits used in rotating an ancilla to prepare a uniform superposition over $n_r$ bits with success probability $\Ps(n_r,b_r)$
    \item $a_V$ -- Number of amplitude amplifications associated to $\PREP_V$
    \item  $Z_I, Z_{\text{ion}_I}$ -- Nuclear charge, ionic charge
    \item $\Omega$ -- Cell volume
    \item $\lambda$ -- The LCU induced one-norm of the Hamiltonian
    \item $\lambda_X$ -- Usually the LCU induced one-norm of the term $X$ in the Hamiltonian
    \item $\error$ -- The total error in the energy estimation, usually set to $0.043\text{eV}$ corresponding to chemical accuracy
    \item $\error_{\operatorname{QPE}}$ -- The Quantum Phase Estimation (QPE) error
    \item $\error_{\square}$ -- Error $\error$ associated to finite size register or process $\square$, defined in \cref{sec:overview_of_errors}
    \item $\varphi_p$, $\phi_{lm}$ -- Plane wave, pseudo wave function
    \item $\beta_i$ -- Projector from the non-local potential
    \item $B_{ij}$ -- Projector matrix coefficients associated to the pseudopotential
    \item $\beta_X$ -- The QROM space-time trade-off parameter associated to $X$, where $X=loc,NL,\Psi$, corresponding to $\PREP_{loc},\PREP_{NL},\SEL_{NL}$
    \item $T$, $U$, $u^{\text{loc}}$, $u^{\text{NL}}$ , $V$  -- Kinetic, external potential, local potential, non-local potential, and electron-electron interaction Hamiltonian terms
    \item $\Psi_{I,\sigma}$ -- Gaussian states defined in~\eqref{eq:Psi_I,0}, \eqref{eq:Psi_I,1}, \eqref{eq:Psi_I,2,0}, \eqref{eq:Psi_I,2,w}
    \item $\gamma_I$ -- Function in the local term, defined in \eqref{eq:dfn_gamma_I} 
    \item $I$ -- Nuclei index
    \item $L$ -- The number of nuclei
    \item $l,m$ -- Angular momentum indices
    \item  $r_{loc}$, $r_i$, $\alpha$, $C_i$, $A_i$ -- HGH pseudopotential parameters that depend on the atom
    \item $p_{\text{th}}$ -- Chosen threshold for the amplitude amplification for success probability $P_{\nu,V}$ amplified to $P_{\nu,V}^{amp}$, corresponding to $\PREP_V$
\end{itemize}
\end{document}